\documentclass[11pt]{article}
\usepackage[colorlinks = true, linkcolor = blue, filecolor = blue, urlcolor = blue, citecolor = blue]{hyperref}
\usepackage{graphicx}
\usepackage{color}
\usepackage{amsmath, amssymb, amsthm, comment, enumerate, ascmac}
\usepackage{latexsym, mathrsfs, mathtools, psfrag, setspace,lscape}
\usepackage{newtxtext}
\usepackage{natbib}
\usepackage[stable]{footmisc}
\usepackage{multirow}
\usepackage[margin = 1truein, nohead]{geometry}

\newcommand{\mbf}{\mathbf}
\newcommand{\mcl}{\mathcal}
\newcommand{\mbb}{\mathbb}
\newcommand{\mrm}{\mathrm}

\renewcommand{\hat}{\widehat}
\renewcommand{\tilde}{\widetilde}
\newcommand{\E}{\mathbb{E}}
\newcommand{\bgamma}{\boldsymbol \gamma}

\newtheorem{theorem}{Theorem}[section]

\newtheorem{example}[theorem]{Example}
\newtheorem{remark}[theorem]{Remark}

\newtheorem{lemma}[theorem]{Lemma}

\newtheorem{assumption}{Assumption}[section]

\makeatletter

\@addtoreset{figure}{section}

\@addtoreset{table}{section}
\makeatother

\numberwithin{equation}{section}
 
\allowdisplaybreaks[1]

\DeclareMathOperator*{\argmin}{\arg\!\min}
\DeclareMathOperator*{\argmax}{\arg\!\max}

\title{A Pairwise Strategic Network Formation Model with Group Heterogeneity: With an Application to International Travel}

\author{Tadao Hoshino\thanks{
    School of Political Science and Economics, Waseda University, 1-6-1 Nishi-waseda, Shinjuku-ku, Tokyo 169-8050, Japan. Email: \href{mailto:thoshino@waseda.jp}{thoshino@waseda.jp}.
    This work is supported financially by JSPS Grant-in-Aid for Scientific Research C-20K01597.
    }
}

\begin{document}
\maketitle
	
\begin{abstract}

    In this study, we consider a pairwise network formation model in which each dyad of agents strategically determines the link status between them.
    Our model allows the agents to have unobserved group heterogeneity in the propensity of link formation.
    For the model estimation, we propose a three-step maximum likelihood (ML) method.
    First, we obtain consistent estimates for the heterogeneity parameters at individual level using the ML estimator.
    Second, we estimate the latent group structure using the binary segmentation algorithm based on the results obtained from the first step.
    Finally, based on the estimated group membership, we re-execute the ML estimation. 
    Under certain regularity conditions, we show that the proposed estimator is asymptotically unbiased and distributed as normal at the parametric rate.
    As an empirical illustration, we focus on the network data of international visa-free travels.
    The results indicate the presence of significant strategic complementarity and a certain level of degree heterogeneity in the network formation behavior.
		
\bigskip
		
\end{abstract}
	
\newpage 
\section{Introduction}\label{sec:introduction}

Empirical modeling of network formation is an important research topic that has been studied for several decades.
While most of these models has been developed in the mathematical statistics literature, as the importance of network structure in many economic activities has been increasingly recognized, there is currently a growing number of econometric studies that focus on network formation in conjunction with the significant advancement in the related econometric techniques.\footnote{
    For recent developments regarding econometric approaches for analyzing network formation, we refer readers to, for example, \cite{chandrasekhar2016econometrics} and \cite{de2020econometric}.
    }

Econometric studies on network formation can be classified into two types: those that attempt to explicitly incorporate the interaction of individuals in the realizing network structure endogenously affecting the network formation behavior (e.g., \citealp{leung2015two}; \citealp{mele2017structural}; \citealp{sheng2020structural}) and those that do not account for such simultaneous interactions but emphasize modeling a flexible form of individual heterogeneity (e.g., \citealp{graham2017econometric}; \citealp{jochmans2018semiparametric}; \citealp{dzemski2019empirical}).
For the former type, network formation is modeled as a game in which agents strategically form links to maximize their payoffs.
Although this game-theoretic approach is (economic) theoretically well-underpinned, we often encounter serious analytical difficulties due to the presence of multiple equilibria.
To circumvent these difficulties, we typically need to introduce some ad hoc behavioral assumptions into the network formation process, or we simply discontinue point-identifying the models and resort to partial identification.
Compared with the former, the latter is more ``descriptive'' than ``structural'', but has great flexibility in the model specification.
These models are relatively easy to implement and, thus, are appealing to empirical researchers.
However, they are not suitable for analyzing the interactions of agents in network formation, which should be an essential factor in economic and social network data.

These two types of econometric models have their own advantages.
Hence, it is an ingenuous idea to construct a new model that has the advantages of both approaches by combining them.
However, to my best knowledge, there are only a few papers that address this way of model extension (e.g., 
\citealp{graham2016homophily}; \citealp{graham2020testing}; \citealp{pelican2020optimal}).
\cite{graham2016homophily} gets around the multiple equilibria issue by considering the ``dynamic'' (rather than the instantaneous) interdependencies in network formation.
The latter two papers consider a quite general framework that incorporates both a general form of strategic interaction and unobserved degree heterogeneity into a single model.
However, they mainly focus on testing the presence of interactions, but not on the estimation of the models.

In this paper, we propose a new ``pairwise'' network formation model that is empirically tractable while retaining the nice properties of the above-mentioned approaches.
More specifically, we assume that each link connection is determined by the strategic interaction solely between the corresponding dyad of agents, without affecting or being affected by other dyads, rather than regarding the realized network as a consequence of a large $n$-player game.
Although ignoring such network externalities limits the range of applications of our model, it would still cover a fairly large number of interesting empirical situations, and, most importantly, the multiple equilibria problem can be greatly mitigated.
In addition, we allow each agent's payoff to depend on two unobserved preference parameters that represent his/her outgoing and incoming propensity, which we refer to as the sender and the receiver effect, respectively.

For estimating our pairwise network formation model, assuming that the model error terms follow some parametric distribution, we propose using the maximum likelihood (ML) method.
Although multiple equilibria can occur even in our binary game situation, we can address this issue in the same manner as in \cite{bresnahan1990entry} and \cite{berry1992estimation}.
As in \cite{dzemski2019empirical} and \cite{yan2019statistical}, we treat the agent-specific preference heterogeneity parameters as the fixed-effect parameters to be estimated.
However, note that if we include these heterogeneity parameters directly into the model, as the dimension of the parameters increases proportionally to the sample size, our ML estimator suffers from the \textit{incidental parameter} problem.
As will be confirmed in the numerical simulations reported later, the incidental parameter bias can be severe.
Thus, to avoid the incidental parameter problem, as a second novel part of this study, we focus on a situation where the individuals can be classified into several unknown groups and their specific effects are homogeneous within each group.
If the number of the latent groups is fixed, we can expect that the ML estimator becomes asymptotically unbiased at the parametric rate and hence the standard inference procedure can be applied.

In the literature on statistical network data analysis, uncovering such latent group structures in networks has been intensively studied (e.g., \citealp{newman2004finding}; \citealp{bickel2009nonparametric}; \citealp{fortunato2010community}; \citealp{karrer2011stochastic}; \citealp{rohe2011spectral}; \citealp{abbe2017community}).
Indeed, putting the strategic interaction effect aside, our network formation model can be regarded as a type of the \textit{stochastic block model}, one of the major modeling approaches in the above literature, with additional covariates, as in \cite{zhang2019node}.
In panel data analysis also, identification of unobserved grouped heterogeneity is one of the most active research areas (e.g., \citealp{bonhomme2015grouped}; \citealp{ke2016structure}; \citealp{su2016identifying}; \citealp{wang2018homogeneity}).
Although the application of these grouping methods to econometric network models has been relatively limited thus far, it is a promising approach, as discussed in \cite{bonhomme2020econometric}.
If we can recover the true group structure with probability approaching one (w.p.a.1) by using any method, the individual fixed-effect parameters can be estimated at a faster rate than the case of no grouping.
In this study, among several alternative methods, we adopt the binary segmentation (BS) method (see, e.g., \citealp{bai1997estimating}; \citealp{ke2016structure}; \citealp{lian2019homogeneity}; \citealp{wang2020identifying}).
Compared to the other grouping methods, the BS method has several favorable properties including fast computation speed and robustness as it does not require us to set initial values.

The whole estimation procedure is divided into three steps.
The first step is to obtain the ML estimator without considering the latent group structures.
Although this estimator suffers from the incidental parameter bias, it is still possible to produce consistent estimates for the heterogeneity parameters.
In the second step, we apply the BS method with respect to the estimated sender effect parameters and the receiver effect parameters separately to identify each agent's group memberships.
In the third step, we re-estimate the model using the ML method given the estimated group structure.
Under certain regularity conditions, we show that the proposed estimator is asymptotically unbiased and normal at the parametric rate.
Furthermore, the estimator is asymptotically equivalent to the ``oracle'' estimator that is obtained based on the (unknown) true group memberships.

To illustrate our model empirically, we investigate the formation of international visa-free travel networks, where the dependent variable of interest is defined as follows: $g_{i,j} = 1$ if country $i$ allows the citizens in country $j$ to visit $i$ without visas and $g_{i,j} = 0$ if not.
We apply our model framework to the network of 57 countries selected mainly from Asia, the Middle East, the former USSR, and Oceania.
As expected, we find the presence of a certain level of degree heterogeneity in terms of both the sender and the receiver effects.
Interestingly, there seems to be a negative correlation between the sender effects and the receiver effects (in other words, there is a tendency that a country's sender effect increases as its receiver effect decreases).
Our estimation result also suggests that there is a significant strategic complementarity in the network formation behavior.
Another interesting finding is that the countries are \textit{homophilous} -- tend to connect with similar others -- in terms of the political system.
These findings would highlight the usefulness of the proposed model and method.
	
\paragraph{Organization of the paper:}
The remainder of the paper is organized as follows.
In Section \ref{sec:model}, we formally introduce the model investigated in this study.
In this section, we demonstrate that our pairwise model exhibits multiple equilibria and discuss the conditions under which the model can be point-identified. 
Section \ref{sec:3stepML} provides a detailed explanation about our three-step ML estimator.
We also investigate the asymptotic properties of the proposed estimator in this section.
In Section \ref{sec:MC}, we present a set of Monte Carlo experiments to evaluate the finite sample performance of the proposed estimator.
Section \ref{sec:empiric} presents our empirical analysis, and, finally, Section \ref{sec:conclusion} concludes.
All the technical details are relegated to Appendix.

\paragraph{Notation:}
For a natural number $n$, $I_n$ denotes an $n \times n$ identity matrix.
$\mbf{1}\{\cdot\}$ denotes the indicator function, which is one if its argument is true and zero otherwise.
For a matrix $A$, we use $||A||$ to denote its Frobenius norm: $||A|| =\sqrt{\mrm{tr}\{AA^\top\}}$, where $\mrm{tr}\{\cdot\}$ is a trace of a matrix.
When $A$ is a square matrix, we use $\lambda_{\mrm{min}}(A)$ to denote its smallest eigenvalue.
For a vector $\mbf{a} = (a_1, \ldots, a_k)^\top$, $||a||_\infty$ denotes its maximum norm: $||\mbf{a}||_\infty = \max_{1 \le i \le k}|a_i|$.
For a general set $\mcl{X}$, we use $\mcl{X}^{\mathrm{int}}$ to denote its interior.
In addition, $|\mcl{X}|$ denotes the cardinality of $\mcl{X}$.
$c$ (possibly with subscript) denotes a generic positive constant whose exact value may vary per case.


\section{Model Setup and Identification}\label{sec:model}

\subsection{Pairwise strategic network formation model}

Suppose that we have a sample of $n$ agents that form social networks whose connections are represented by an $n \times n$ adjacency matrix $G_n = (g_{i,j})_{1 \le i,j \le n}$.
These agents can be individuals, firms, municipalities, or nations depending on the context.
The network is directed; that is, regardless of the value of $g_{j,i}$, we observe $g_{i,j} = 1$ if agent $i$ links to $j$ and $g_{i,j} = 0$ otherwise.
There are no self-loops; that is, the diagonal elements of $G_n$ are all zero.
Throughout the paper, we assume that the status of $(g_{i,j}, g_{j,i})$ is determined solely by the pair of agents $(i,j)$, without considering the status of other network links.
Specifically, for each pair $(i,j)$, suppose that $i$'s marginal payoff of forming a link with $j$ given $g_{j,i} = q$ is written as
\begin{align*}
    u_{i,j}(q) = Z_{i,j}^\top \beta_0 + \alpha_0 q + A_{0,i} + B_{0,j} - \epsilon_{i,j}, \;\; \text{for} \;\; i \neq j.
\end{align*}
Here, $Z_{i,j} \in \mbb{R}^{d_z}$ is a vector of observed covariates, $A_{0,i} \in \mbb{R}$ is agent $i$'s individual specific effect as a ``sender'', $B_{0,j}\in\mbb{R}$ is $j$'s individual specific effect as a ``receiver'', $\epsilon_{i,j} \in \mbb{R}$ is an unobservable payoff component, and $\beta_0 \in \mbb{R}^{d_z}$ and $\alpha_0 \in \mbb{R}$ are unknown coefficient vector and the interaction effect parameter, respectively.
The covariates $Z_{i,j}$ and $Z_{j,i}$ may contain common elements; however, in the later discussion, we require that they must have some agent specific elements that can vary across the partners.
The individual specific effects $A_{0,i}$ and $B_{0,j}$, which we call the sender and the receiver effect, respectively, can be interpreted as the level of $i$'s willingness to create connections with others and the popularity of $j$, respectively, that generate degree heterogeneity across the agents.
Following \cite{dzemski2019empirical} and \cite{yan2019statistical}, we treat $\{(A_{0,i}, B_{0,i})\}$ as fixed-effect parameters to be estimated.

We assume that the agents have complete information; that is, the realizations of $(Z_{i,j}, Z_{j,i})$ and $(\epsilon_{i,j}, \epsilon_{j,i})$ are common knowledge to both $i$ and $j$.
Then, if we assume that the observed network $G_n$ is formed by a collection of Nash equilibrium actions, we obtain the following econometric model:
\begin{align}\label{eq:model}
\begin{split}
    g_{i,j} & = \mbf{1}\left\{ Z_{i,j}^\top \beta_0 + \alpha_0 g_{j,i} + A_{0,i} + B_{0,j} \ge \epsilon_{i,j} \right\} \\
    g_{j,i} & = \mbf{1}\left\{ Z_{j,i}^\top \beta_0 + \alpha_0 g_{i,j} + A_{0,j} + B_{0,i} \ge \epsilon_{i,j} \right\}, \;\; \text{for} \;\; i \neq j.
\end{split}
\end{align}

The following are the two examples to which the above framework can be potentially applied.

\begin{example}[Online social networking]\upshape
    The analysis of online social networking behavior is an active research topic in network science.
    For some social networking sites, users can easily establish links to others (become a \textit{follower}) without mutual consent.
    In addition, whether a person becomes a follower of someone is often irrelevant to who else they are following.
    Thus, this would be a situation where our framework reasonably fits.
\end{example}

\begin{example}[International visa-free network]\upshape
    In the research on international migration and tourism, investigating the determinants and impacts of visa policies is one of the central interests (e.g., \citealp{neiman2009impact}; \citealp{neumayer2010visa}; \citealp{mckay2018tall}).
    As bilateral visa policies is naturally observed as a consequence of strategic (economic and/or political) interactions between the two countries, our model would be an appropriate analytical tool here.
    In our empirical study presented in Section \ref{sec:empiric}, by setting $g_{i,j} = 1$ if country $i$ allows visa-free entry for the citizens of country $j$ and $g_{i,j} = 0$ if not, we will show that the magnitude of the bilateral interaction in visa policies is significant. \label{ex:visa}
\end{example}

In these examples, we can naturally imagine that the strategic interaction effect $\alpha_0$ is positive (i.e., strategic complements).
We assume that strategic complementarity would be reasonable for most empirical situations of network formation games.
Then, throughout the paper, we impose this assumption: $\alpha_0 > 0$.
Under strategic complementarity, each pair's Nash equilibrium action can be summarized in Figure \ref{fig:nash}.
As shown in the figure, the space of $(\epsilon_{i,j}, \epsilon_{j,i})$ cannot be partitioned into non-overlapping regions associated with the four alternative realizations of $(g_{i,j}, g_{j,i})$.
That is, both $(g_{i,j}, g_{j,i}) = (1,1)$ and $(g_{i,j}, g_{j,i}) = (0,0)$ can occur in the shaded area in the figure, and the link status is not uniquely determined in this area (i.e., multiple equilibria exist).
This non-uniqueness of model-consistent decisions is called \textit{incompleteness} and has been extensively studied in the literature on simultaneous equation models for discrete outcomes  (e.g., \citealp{tamer2003incomplete}; \citealp{lewbel2007coherency}; \citealp{ciliberto2009market}; \citealp{chesher2020structural}).

\begin{figure}[!h]
	\centering
	\includegraphics[width = 10cm]{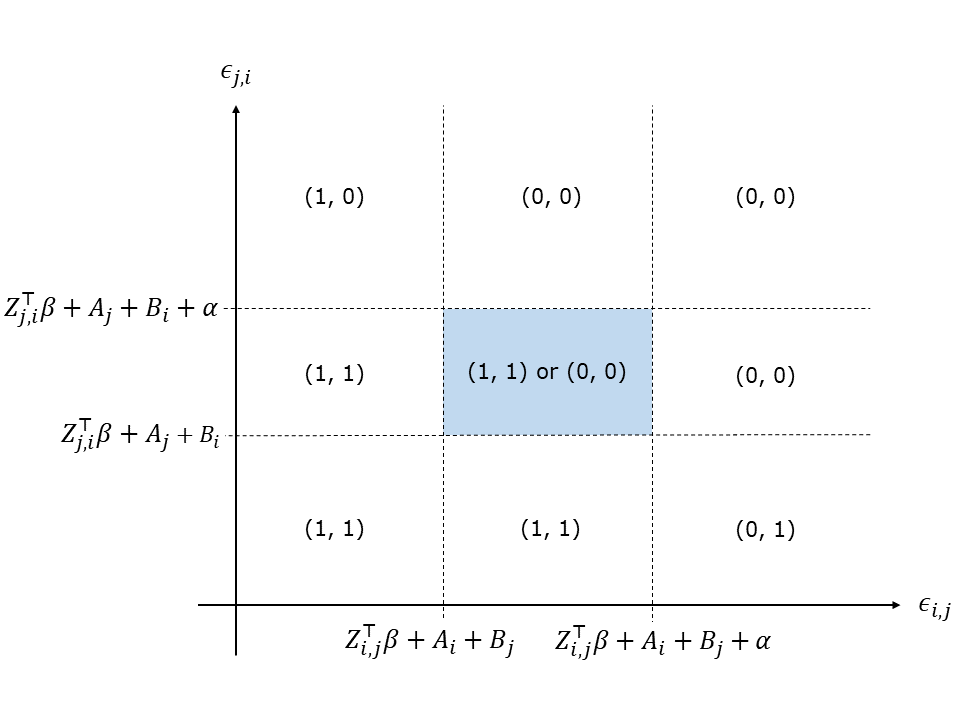}
	\caption{Pure strategy Nash equilibrium}
	\label{fig:nash}
\end{figure}

There are several approaches to handle this incompleteness issue in the literature.\footnote{
    Refer to \cite{de2013econometric} for a comprehensive survey on this topic.
    }
Among them, this study adopts the traditional approach developed by \cite{bresnahan1990entry} and \cite{berry1992estimation} that focuses only on the unique equilibrium outcomes.
That is, we consider estimating the model based only on the information about ``one-way links'' in the network.

In the following, we assume that $\{\epsilon_{i,j}\}$ are identically distributed with a known cumulative distribution function (CDF) $F$.
Further, we assume that the pairs $\{(\epsilon_{i,j}, \epsilon_{j,i})\}$ are independent and identically distributed (i.i.d.) across pairs, and their joint distribution is represented by $H(\cdot, \cdot; \rho_0)$ such that $\Pr(\epsilon_{i,j} \le a_1, \epsilon_{j,i} \le a_2) = H(a_1, a_2 ; \rho_0)$, where $\rho_0 \in \mbb{R}$ is a parameter controlling the correlation between $\epsilon_{i,j}$ and $\epsilon_{j,i}$.
Define $y^{(1,0)}_{i,j} \equiv \mbf{1}\{(g_{i,j}, g_{j,i}) = (1,0)\}$ and $y^{(0,1)}_{i,j} \equiv \mbf{1}\{(g_{i,j}, g_{j,i}) = (0,1)\}$.
Let $\chi_{n,a}$ be the $a$-th column of $I_n$, $\mbf{A}_0 = (A_{0,1}, \ldots, A_{0,n})^\top$, and $\mbf{B}_0 = (B_{0,1}, \ldots, B_{0,n})^\top$, so that we can write 
\begin{align*}
    Z_{i,j}^\top \beta_0 + A_{0,i} + B_{0,j}= W_{i,j}^\top \Pi_0, 
\end{align*}
where $W_{i,j} = (Z_{i,j}^\top, \chi_{n,i}^\top, \chi_{n,j}^\top)^\top$ and $\Pi_0 = (\beta_0^\top, \mbf{A}_0^\top, \mbf{B}_0^\top)^\top$.
In addition, we denote $\theta_0 = (\beta_0^\top, \alpha_0, \rho_0)^\top$ and $\bgamma_0 = (\mbf{A}_0^\top, \mbf{B}_0^\top)^\top$.
Then, the conditional probabilities of $\{y_{i,j}^{(1,0)} = 1\}$ and $\{y_{i,j}^{(0,1)} = 1\}$ are respectively given as follows:
\begin{align*}
    P^{(1,0)}_{i,j}(\theta_0, \bgamma_0) 
    & \equiv F(W_{i,j}^\top \Pi_0) - H(W_{i,j}^\top \Pi_0, W_{j,i}^\top \Pi_0 + \alpha_0; \rho_0), \\
    P^{(0,1)}_{i,j}(\theta_0, \bgamma_0) 
    & \equiv F(W_{j,i}^\top \Pi_0) - H(W_{i,j}^\top \Pi_0 + \alpha_0, W_{j,i}^\top \Pi_0; \rho_0).
\end{align*}
Here, note that the equalities $y^{(1,0)}_{i,j} = y^{(0,1)}_{j,i}$ and $P^{(1,0)}_{i,j}(\theta, \bgamma) = P^{(0,1)}_{j,i}(\theta, \bgamma)$ hold.
Thus, the likelihood function can be concentrated with respect to $(y_{i,j}^{(1,0)}, P_{i,j}^{(1,0)}(\theta, \bgamma))$; hereinafter, we omit the superscripts and denote $y_{i,j} = y_{i,j}^{(1,0)}$ and $P_{i,j}(\theta, \bgamma) = P_{i,j}^{(1,0)}(\theta, \bgamma)$ when there is no confusion.
Then, the log-likelihood function can be written as
\begin{align}\label{eq:llfunc}
    \begin{split}
    \mcl{L}_n(\theta, \bgamma) 
    & = \frac{2}{N}\sum_{i = 1}^n \sum_{j > i} \left[ y_{i,j} \ln P_{i,j}(\theta, \bgamma) + y_{j,i} \ln P_{j,i}(\theta, \bgamma) + (1 - y_{i,j} - y_{j,i}) \ln (1 - P_{i,j}(\theta, \bgamma) - P_{j,i}(\theta, \bgamma) )\right] \\
    & = \frac{1}{N}\sum_{i = 1}^n \sum_{j \neq i} \left[ y_{i,j} \ln P_{i,j}(\theta, \bgamma) + y_{j,i} \ln P_{j,i}(\theta, \bgamma) + (1 - y_{i,j} - y_{j,i}) \ln (1 - P_{i,j}(\theta, \bgamma) - P_{j,i}(\theta, \bgamma) )\right] \\
    & = \frac{1}{N} \sum_{i = 1}^n \sum_{j \neq i} \left[ 2 y_{i,j} \ln P_{i,j}(\theta, \bgamma)  + (1 - 2y_{i,j}) \ln (1 - P_{i,j}(\theta, \bgamma) - P_{j,i}(\theta, \bgamma) ) \right],
    \end{split}
\end{align}
where $N \equiv n(n-1)$.
As above, we can consider three equivalent representations for the log-likelihood function and switch between them according to analytical convenience.

\subsection{Discrete heterogeneity}

As mentioned in Introduction section, this study considers situations where the agents are grouped into several sub-samples, and the individual fixed effects are heterogeneous across these groups but are homogeneous within the groups in the following manner:
\begin{align}\label{eq:grouping}
    \begin{split}
    A_{0,i}
    = \sum_{k = 1}^{K^A} a_{0,k} \cdot \mbf{1}\{ i \in \mcl{C}^A_{0,k}\},
    \quad
    B_{0,i}
    = \sum_{k = 1}^{K^B} b_{0,k} \cdot \mbf{1}\{ i \in \mcl{C}^B_{0,k}\}.
    \end{split}
\end{align}
That is, the agents can be classified into $K^A$ groups $\mcl{C}^A_0 \equiv \{\mcl{C}^A_{0,1}, \ldots, \mcl{C}^A_{0,K^A}\}$ in terms of the sender effects $\{A_{0,i}\}$, where $K^A$ is the total number of groups, which form a partition of $\{1, \ldots, n\}$ into $K^A$ subsets.
Similarly, in terms of the receiver effects $\{B_{0,i}\}$, the agents can be grouped as $\mcl{C}^B_0 \equiv \{\mcl{C}^B_{0,1}, \ldots, \mcl{C}^B_{0,K^B}\}$.
When an individual is a member of the intersection $\mcl{C}_{0,k}^A \bigcap \mcl{C}_{0,l}^B$, his/her sender and receiver effects are equal to $a_{0,k}$ and $b_{0,l}$, respectively.
This intersection set would correspond to one ``community'' in the community detection literature.
The group where each individual belongs to remains unknown to us.
Meanwhile, it is often assumed in the literature that the number of groups is known to researchers (e.g., \citealp{bonhomme2015grouped}; \citealp{okui2020heterogeneous}).
Then, following these studies, we treat $K^A$ and $K^B$ as known values and assume that $K^A, K^B \ge 2$.\footnote{ 
    This assumption should be debatable.
    Rather than directly specifying the number of groups, in the literature on the BS algorithm for example, \cite{ke2016structure} and \cite{lian2019homogeneity} propose introducing an additional threshold parameter to detect the group structure.
    However, since the number of groups can vary only discretely with the threshold value, introducing the threshold parameter is essentially equivalent to selecting the number of groups.
    For the use of the BS algorithm, \cite{wang2020identifying} formally shows that Bayesian Information Criterion (BIC)-type criterion can consistently select the correct number of groups as the sample size increases.  
    }
We discuss how to choose $K^A$ and $K^B$ in practice in Section \ref{sec:empiric}.
Note that transforming $(\mbf{A}_0, \mbf{B}_0)$ to $(\mbf{A}_0 + c, \mbf{B}_0 - c)$ for any constant $c$ does not change the model \eqref{eq:model}.
Thus, without loss of generality, we assume that $a_{0,1} = 0$ for location normalization.

Under this setup, the full ML estimator solves
\begin{align}\label{eq:full_mle}
    \max_{\left(\theta, \mbf{a}, \mbf{b}, \mcl{C}^A, \mcl{C}^B \right) }\mcl{L}_n(\theta, \mbf{A}, \mbf{B}),
\end{align}
where $\mbf{a} \equiv (a_1, \ldots, a_{K^A})$ with $a_1 = 0$, $\mbf{b} \equiv (b_1, \ldots, b_{K^B})$, $\mcl{C}^A \equiv \{\mcl{C}^A_1, \ldots, \mcl{C}^A_{K^A}\}$, $\mcl{C}^B \equiv \{\mcl{C}^B_1, \ldots, \mcl{C}^B_{K^B}\}$, $A_i = \sum_{k = 1}^{K^A} a_k \cdot \mbf{1}\{ i \in \mcl{C}^A_k\}$, and $B_i = \sum_{k = 1}^{K^B} b_k \cdot \mbf{1}\{ i \in \mcl{C}^B_k\}$.
The maximization problem in \eqref{eq:full_mle} is clearly a combinatorial (NP-hard) optimization problem.
In the context of panel data models, several authors have proposed iterative k-means (like) algorithms to computationally obtain a (local) solution efficiently to the problems similar to \eqref{eq:full_mle} (e.g., \citealp{bonhomme2015grouped}; \citealp{liu2020identification}).
However, the iterative algorithm is still computationally demanding.
More importantly, it cannot be directly applied to our network model where each agent's heterogeneity parameters affect not only the value of his/her own likelihood function but also that of the others.

Hence, in this paper, we propose to decompose the maximization problem in \eqref{eq:full_mle} into three steps.
The first step is to estimate $\bgamma_0 = (\mbf{A}_0^\top, \mbf{B}_0^\top)^\top$ using the full ML estimator based on the log-likelihood function in \eqref{eq:llfunc} without explicitly considering the group structure.
Given the consistent estimates of these parameters, the second step is to estimate the group memberships $\mcl{C}_0^A$ and $\mcl{C}_0^B$ using the BS algorithm (e.g., \citealp{bai1997estimating}; \citealp{ke2016structure}; \citealp{lian2019homogeneity}; \citealp{wang2020identifying}).
The final step is to solve \eqref{eq:full_mle} with the group structure replaced by the estimated $\mcl{C}_0^A$ and $\mcl{C}_0^B$.

\subsection{Identification}

Before presenting the estimation procedure in detail, we discuss identification conditions for the parameters in model \eqref{eq:model}.
It is important to note that, even when individual heterogeneity parameters have only finite variations groupwisely, each individual's parameters must be point-identified separately to estimate the group structure consistently.
A practical reason for this is that our three-step estimator based on the BS algorithm requires preliminary consistent estimates of $\mbf{A}_0$ and $\mbf{B}_0$ in the estimation of the group structure.
Note that we can always estimate ``pseudo-true'' group memberships based on the maximum likelihood principle even when some elements of $\mbf{A}_0$ and $\mbf{B}_0$ are not point-identified.
However, they may not necessarily coincide with the true group memberships in general.\footnote{
    A more formal investigation on this issue is left as a future work.
    For a related discussion, see \cite{bonhomme2015grouped}.
    }

Here, again, we need some location normalization to identify $\mbf{A}_0$ and $\mbf{B}_0$.
In the following, similarly as above, we set $A_{0,1} = 0$ without loss of generality.
To facilitate the discussion, we also introduce several simplifying assumptions, some of which are mentioned previously.

\begin{assumption}\label{as:error}
    (i)  The payoff disturbances $\{\epsilon_{i,j}\}$ are identically distributed on the whole $\mbb{R}$ with a known strictly increasing  marginal CDF $F(\cdot)$.
    (ii) The pairs $\{(\epsilon_{i,j}, \epsilon_{j,i})\}$ are i.i.d. across dyads with joint CDF $H(\cdot, \cdot; \rho_0)$, and $H(\cdot, \cdot; \rho)$ is strictly increasing in each argument for all $\rho \in \mcl{R}$.
    (iii) $F(\cdot)$ is three times continuously differentiable, and $H(\cdot, \cdot; \cdot)$ is three times continuously differentiable with respect to all arguments.
\end{assumption}

Let $\Theta \equiv \mcl{B} \times \mcl{A} \times \mcl{R}$, $\mbb{A}_n \equiv \{0\} \times \mbb{A}^{n-1}$, $\mbb{B}_n \equiv \mbb{B}^n$, and $\mbb{C}_n \equiv \mbb{A}_n \times \mbb{B}_n$, where $\mcl{B} \subset \mbb{R}^{d_z}$, $\mcl{A} \subset \mbb{R}_{++}$, $\mcl{R} \subset \mbb{R}$, $\mbb{A} \subset \mbb{R}$, and $\mbb{B} \subset \mbb{R}$ are parameter spaces for $\beta$, $\alpha$, $\rho$, $A_i$'s, and $B_i$'s, respectively.

\begin{assumption}\label{as:para_space}
    (i) $\theta_0 \in \Theta^\mathrm{int}$, where $\Theta$ is compact.
    (ii) For all $i = 2, 3, \ldots$, $A_{0,i} \in \mbb{A}^{\mathrm{int}}$, and, for all $i = 1, 2, \ldots$, $B_{0,i} \in \mbb{B}^{\mathrm{int}}$, where $\mbb{A}$ and $\mbb{B}$ are compact.
\end{assumption}

\begin{assumption}\label{as:Z_bound}
    The covariates $\{Z_{i,j}\}$ are uniformly bounded.
\end{assumption}

In Assumption \ref{as:error}(i), we assume that the marginal CDF of the error term is known.
This assumption is typically adopted in the estimation of complete information games.
As shown by \cite{khan2018information}, when the marginal CDFs of $(\epsilon_{i,j}, \epsilon_{j,i})$ are unknown, it is generally impossible to estimate the interaction effect $\alpha_0$ at the parametric rate.
Assumption \ref{as:error}(ii) requires that the error terms are independent across dyads.
Note that the parameters $(A_{0,i}, B_{0,i})$ have the role of accommodating all unobserved payoff components in link formation behavior specific to $i$.
Therefore, assuming the independence within the remainders $\{\epsilon_{i,1} ,\ldots, \epsilon_{i,n} \}$ should not be too restrictive.
The other requirements in Assumption \ref{as:error} are standard in that they are satisfied in most of commonly used parametric models (such as logit and probit).
In Assumption \ref{as:para_space}(ii), we assume that the fixed-effect parameters $\{(A_{0,i}, B_{0,i})\}$ are bounded.
Although imposing boundedness on the degree heterogeneity parameters is commonly accepted in the literature on network formation models, some studies consider a more general framework where $||\mbf{A}_0||_\infty$ and $||\mbf{B}_0||_\infty$ can grow slowly (e.g., \citealp{yan2016asymptotics}; \citealp{yan2019statistical}).
The admissible parameter space for the correlation parameter $\rho$, $\mcl{R}$, depends on the choice of the functional form of $H$, which is typically $\mcl{R} = [-1,1]$.
Assumption \ref{as:Z_bound} should not be restrictive in practice.
Hereinafter, we fix the values of $\{Z_{i,j}\}$; that is, we interpret the following analysis as being conditional on the realization of $\{Z_{i,j}\}$.
Thus, any randomness in the model is considered to be due to the randomness of $\{\epsilon_{i,j}\}$. 

An important implication from Assumptions \ref{as:error}--\ref{as:Z_bound} is that the one-way link-formation probabilities $\{P_{i,j}(\theta, \bgamma)\}$ are uniformly bounded away from $0$ and $1$ for all possible parameter values.
In other words, we are assuming that our networks are dense such that the number of one-way links per agent will be about proportional to the number of sampled agents.
The plausibility of this assumption depends on the context of application.
For example, international trade networks and the visa-free travel networks given in Example \ref{ex:visa} and Section \ref{sec:empiric} may be regarded as dense networks.\footnote{
    \cite{graham2016homophily} and \cite{jochmans2018semiparametric} developed conditional likelihood methods that can be used to estimate the homophily parameters (i.e., $\beta_0$ in our context), even when the networks are sparse.
    However, their approach cannot be directly applied to our case because of the interdependence between $g_{i,j}$ and $g_{j,i}$. 
    }

\begin{assumption}\label{as:support}
    For any $(\beta_1, \alpha_1, \bgamma_1), (\beta_2, \alpha_2, \bgamma_2) \in \mcl{B} \times \mcl{A} \times \mbb{C}_n$ such that $(\beta_1, \alpha_1, \bgamma_1) \neq (\beta_2, \alpha_2, \bgamma_2)$, either or both (a) and (b) hold: 
    \begin{align*}
    (a) \;\; & \liminf_{n \to \infty} \frac{1}{N} \sum_{i = 1}^n \sum_{j \neq i} \mbf{1}\left\{ W_{i,j}^\top (\Pi_1 - \Pi_2) > 0, \;\;  W_{j,i}^\top (\Pi_1 - \Pi_2) + \alpha_1 - \alpha_2 < 0 \right\} > 0\\
    (b) \;\; & \liminf_{n \to \infty} \frac{1}{N} \sum_{i = 1}^n \sum_{j \neq i} \mbf{1}\left\{ W_{i,j}^\top (\Pi_1 - \Pi_2) < 0, \;\;  W_{j,i}^\top (\Pi_1 - \Pi_2) + \alpha_1 - \alpha_2 > 0 \right\} > 0,
    \end{align*}
    where $\Pi_1 = (\beta_1^\top, \bgamma_1^\top)^\top$, and $\Pi_2 = (\beta_2^\top, \bgamma_2^\top)^\top$.
\end{assumption}

Assumption \ref{as:support} is our main identification condition, which basically requires the following two conditions.
The first condition is a standard full-rank condition for $\{W_{i,j}\}$ and $\{W_{j,i}\}$.
The second condition is that at least either $Z_{i,j}$ or $Z_{j,i}$ should contain agent-specific covariates that have large enough supports and also have variations across all potential partners.
If no such variables exist, since the signs of $Z_{i,j}^\top (\beta_1 - \beta_2)$ and $Z_{j,i}^\top (\beta_1 - \beta_2)$ cannot differ for some parameter values for all $(i,j)$'s, Assumption \ref{as:support} does not hold.
It should be noted that this assumption is not inconsistent with Assumption \ref{as:Z_bound} under the compactness of the parameter space (i.e., Assumption \ref{as:para_space}).
While the existence of player-specific continuous variables with unbounded supports is typically required in the identification of non/semiparametric game models -- the so-called \textit{identification-at-infinity} argument (see, e.g., \citealp{tamer2003incomplete}; \citealp{kline2015identification}), we can develop our identification result under less stringent conditions owing to the full-parametric model specification.

\begin{theorem}[Identification]\label{thm:identification}
\begin{enumerate}[(i)]
    \item Suppose that Assumptions \ref{as:error}(i)--(ii) and \ref{as:para_space}--\ref{as:support} hold.
    Then, if $\rho_0$ is known, $(\beta_0, \alpha_0, \bgamma_0)$ can be point-identified on $\mcl{B} \times \mcl{A} \times \mbb{C}_n$.
    \item If $\rho_0$ is unknown, but it is a unique maximizer of $\mcl{L}^*_n(\rho)$, then $(\theta_0, \bgamma_0)$ can be point-identified on $\Theta \times \mbb{C}_n$, where
    \begin{align}\label{eq:rho_concentrate}
    \mcl{L}_n^*(\rho) \equiv  \E \mcl{L}_n((\tilde \beta_0(\rho)^\top, \tilde \alpha_0(\rho),\rho)^\top, \tilde \bgamma_0(\rho)),
\end{align}
and $(\tilde \beta_0(\rho), \tilde \alpha_0(\rho), \tilde \bgamma_0(\rho)) \equiv \argmax_{(\beta, \alpha,  \bgamma) \in \mcl{B} \times \mcl{A} \times \mbb{C}_n} \E \mcl{L}_n((\beta^\top, \alpha, \rho)^\top, \bgamma)$.
\end{enumerate}
\end{theorem}

These identification results are similar to those in Theorem 2 of \cite{aradillas2019inference}.
That is, the model parameters can be point-identified either (i) if the distribution of unobserved payoff disturbances is fully known or (ii) if $\rho_0$ uniquely maximizes the concentrated log-likelihood function \eqref{eq:rho_concentrate}.
In the literature on network formation models, it is often assumed that $\epsilon_{i,j}$ and $\epsilon_{j,i}$ are independent (e.g., \citealp{hoff2002latent}; \citealp{jochmans2018semiparametric}; \citealp{yan2019statistical}).
If they are independent, since $\rho_0 = 0$ is known, condition (i) is satisfied.
Condition (ii) clearly depends on the choice of $H$ function and is difficult to verify in general; however, this is directly empirically testable.
A more primitive sufficient condition for this is that $\mcl{L}_n^*(\rho)$ is strictly concave, which is satisfied when  $\partial^2 \mcl{L}_n^*(\rho)/(\partial \rho)^2$ is strictly negative under Assumption \ref{as:error}(iii).
We provide an explicit form of $\partial^2 \mcl{L}_n^*(\rho)/(\partial \rho)^2$ in Appendix \ref{subsec:rho_unique}.
Even when neither of condition (i) nor (ii) is satisfied, it remains possible to partially identify the parameters, as in \cite{aradillas2019inference}.
For example, if $\mcl{L}^*_n(\rho)$ has two peaks at $\rho_1$ and $\rho_2$, the resulting identified set is directly given by $\{(\tilde \beta_0(\rho), \tilde \alpha_0(\rho), \rho, \tilde \bgamma_0(\rho)) : \rho \in\{\rho_1, \rho_2\}\}$.

\begin{remark}[Other identification strategies]\upshape
    There are several other routes for identification of our model than the one in Theorem \ref{thm:identification}.
    For example, as our model is fully parametric, a classical parametric identification approach may be used based on the properties of the information matrix (e.g., \citealp{rothenberg1971identification}; \citealp{bjorn1984simultaneous}), although it is not easy to verify in practice.
    If one admits the existence of a player-specific continuous variable that has a positive density on the whole $\mbb{R}$, then the model can be easily identified by the identification-at-infinity approach in the same manner as in \cite{tamer2003incomplete}.
    Since assuming the existence of such unbounded variables is restrictive in practice, several authors have proposed other approaches based on some shape restrictions on the distribution of unobservables, without relying on the identification-at-infinity argument (e.g., \citealp{kline2016empirical}; \citealp{zhou2019identification}).
    Investigating whether their approaches can be applied to our model is an interesting topic, but it is left for a future work.
\end{remark}


\section{Three-Step ML Estimation}\label{sec:3stepML}

\subsection{First-step ML estimator}

The first step of the ML estimation aims to obtain consistent estimates of $\bgamma_0 = (\mbf{A}_0^\top, \mbf{B}_0^\top)^\top$.
Let
\begin{align}\label{eq:mle}
\begin{split}
      (\hat \theta_n, \hat \bgamma_n) = \argmax_{(\theta, \bgamma) \: \in \: \Theta \times \mbb{C}_n} \mcl{L}_n(\theta, \bgamma),
\end{split}
\end{align}
where $\hat \theta_n = (\hat \beta_n^\top, \hat \alpha_n, \hat \rho_n)^\top$, and $\hat \bgamma_n = (\hat{\mbf{A}}_n^\top, \hat{\mbf{B}}_n^\top)^\top$.
Below, we present the asymptotic properties of the initial full ML estimator in \eqref{eq:mle} with the main focus on $\hat \bgamma_n$.
Instead of introducing particular identification conditions, for generality, we directly assume that the true parameter $(\theta_0, \bgamma_0)$ is a unique maximizer of $\E \mcl{L}_n(\theta, \bgamma)$.

\begin{assumption}\label{as:identification}
    $\E \mcl{L}_n(\theta, \bgamma)$ is uniquely maximized at $(\theta_0, \bgamma_0) \in \Theta \times \mbb{C}_n$ for all sufficiently large $n$.
\end{assumption}

We first establish several consistency results in the next theorem.
In particular, we show that the individual specific effects can be uniformly consistently estimated.

\begin{theorem}\label{thm:consistency}
Suppose that Assumptions \ref{as:error}--\ref{as:Z_bound} and \ref{as:identification} hold.
Then, we have (i) $\hat \theta_n \overset{p}{\to} \theta_0$, (ii) $\frac{1}{n}\sum_{i=1}^n | \hat A_{n,i} - A_{0,i} | \overset{p}{\to} 0$, (iii) $\frac{1}{n}\sum_{i=1}^n | \hat B_{n,i} - B_{0,i} | \overset{p}{\to} 0$, and (iv) $||\hat \bgamma_n - \bgamma_0 ||_\infty \overset{p}{\to} 0$.
\end{theorem}

Next, we derive the convergence rate of $\hat \bgamma_n$.
To this end, it is convenient to re-define $\hat \theta_n$ and $\theta_0$ as
\begin{align*}
    \hat \theta_n = \argmax_{\theta \in \Theta} \mcl{L}_n(\theta, \tilde{\bgamma}_n(\theta)) \;\; \text{and} \;\;
    \theta_0 = \argmax_{\theta \in \Theta} \E \mcl{L}_n(\theta, \tilde{\bgamma}_0(\theta)),
\end{align*}
respectively, where $\tilde \bgamma_n(\theta) \equiv \argmax_{\bgamma \in \mbb{C}_n} \mcl{L}_n(\theta, \bgamma)$ and $\tilde \bgamma_0(\theta) \equiv \argmax_{\bgamma \in \mbb{C}_n} \E \mcl{L}_n(\theta, \bgamma)$ for any given $\theta \in \Theta$, assuming that they are well-defined.
Further, we define $\bgamma_{-1} = (A_2, \ldots, A_n, B_1, \ldots, B_n)^\top$, 
\begin{align*}
    \underset{(2n-1) \times (2n-1)}{\mcl{H}_{n, \bgamma\bgamma} (\theta, \bgamma)}
    \equiv  \frac{\partial^2 \mcl{L}_n(\theta, \bgamma)}{ \partial \bgamma_{-1} \partial \bgamma_{-1}^\top }, \qquad
    \underset{(d_z + 2) \times (d_z + 2)}{\mcl{H}_{n, \theta\theta} (\theta, \bgamma)}
    \equiv  \frac{\partial^2 \mcl{L}_n(\theta, \bgamma)}{ \partial \theta \partial \theta^\top}
    , \qquad
    \underset{(2n-1) \times (d_z + 2)}{\mcl{H}_{n, \bgamma\theta} (\theta, \bgamma)}
    \equiv  \frac{\partial^2 \mcl{L}_n(\theta, \bgamma)}{ \partial \bgamma \partial \theta^\top},
\end{align*}
and $\mcl{H}_{n, \theta\bgamma} (\theta, \bgamma) \equiv \mcl{H}_{n, \bgamma\theta} (\theta, \bgamma)^\top$.
The exact form of $\mcl{H}_{n, \bgamma\bgamma} (\theta, \bgamma)$ can be found in \eqref{eq:hessian} in Appendix \ref{sec:notations}.
Finally, let
\begin{align*}
    \underset{(d_z + 2) \times (d_z + 2)}{\mcl{I}_{n,\theta\theta}(\theta, \bgamma)} \equiv  \mcl{H}_{n, \theta \theta}(\theta, \bgamma) - \mcl{H}_{n, \theta \bgamma} (\theta, \bgamma) \left[ \mcl{H}_{n,\bgamma\bgamma} (\theta, \bgamma) \right]^{-1} \mcl{H}_{n, \bgamma \theta } (\theta, \bgamma).
\end{align*}
This $\mcl{I}_{n,\theta\theta}(\theta, \bgamma)$ matrix serves as the Hessian matrix for the concentrated ML estimator $\hat \theta_n$ (see, e.g., \citealp{amemiya1985advanced}).
Now, we introduce the following assumptions.

\begin{assumption}\label{as:unique}
    $\tilde \bgamma_0(\theta)$ uniquely exists uniformly on $\{\Theta: ||\theta - \theta_0|| \le \varepsilon \}$, where $\varepsilon > 0$ is an arbitrary small constant.
\end{assumption}

\begin{assumption}\label{as:hessian}
    For an arbitrary small constant $\varepsilon > 0$, there exist constants $c_{\bgamma}, c_\theta > 0$ that may depend on $\varepsilon$ such that (i) $\lambda_\mathrm{min}\left(- n \cdot \mcl{H}_{n, \bgamma\bgamma} (\theta, \bgamma)\right) > c_{\bgamma}$ and (ii) $\lambda_\mathrm{min} \left(- \mcl{I}_{n, \theta \theta}(\theta, \bgamma)  \right) > c_\theta$ w.p.a.1 uniformly on $\{\Theta \times \mbb{C}_n : ||\theta - \theta_0|| \le \varepsilon, \: ||\bgamma - \bgamma_0||_\infty \le \varepsilon \}$.
\end{assumption}

Assumptions \ref{as:unique} and \ref{as:hessian} should be fairly reasonable in practice.
Then, under these additional assumptions, we can derive the $\ell_1$-norm and max-norm convergence rates for $\hat \bgamma_n$, as shown in the next theorem.

\begin{theorem}\label{thm:conv_rate}
    Suppose that Assumptions \ref{as:error}--\ref{as:Z_bound} and \ref{as:identification}--\ref{as:hessian} hold.
    Then, we have (i) $\frac{1}{n}\sum_{i=1}^n | \hat A_{n,i} - A_{0,i} | = O_P(n^{-1/2})$, (ii) $\frac{1}{n}\sum_{i=1}^n | \hat B_{n,i} - B_{0,i} | = O_P(n^{-1/2})$, and (iii) $|| \hat \bgamma_n - \bgamma_0 ||_\infty = O_P(\sqrt{\ln n / n})$.
\end{theorem}

The max-norm convergence rate obtained in Theorem \ref{thm:conv_rate} is consistent with the result of Theorem 3 in \cite{graham2017econometric} and that of Theorem 3.1 in \cite{yan2019statistical}.

\subsection{Binary segmentation algorithm}\label{subsec:bs}

Given the consistent estimates of $\mbf{A}_0$ and $\mbf{B}_0$, we use the BS algorithm to estimate the group structure.
We first sort $\hat{\mbf{A}}_n$ and $\hat{\mbf{B}}_n$ in ascending order and write the order statistics as
\begin{align*}
   \hat A_{n,(1)} \le \hat A_{n,(2)} \le \cdots \le \hat A_{n,(n)}, 
   \;\; \text{and} \;\;
   \hat B_{n,(1)} \le \hat B_{n,(2)} \le \cdots \le \hat B_{n,(n)}.
\end{align*}
In the following, we mainly describe the estimation of the group membership for the sender effects, $\mcl{C}^A_0$.
The exactly the same procedure described below can be used to estimate $\mcl{C}^B_0$.

An concept behind the BS algorithm is quite simple.
If $\{A_{0,i}\}$ are heterogeneous across $K^A$ latent groups but are homogeneous within the groups, there should exist $K^A - 1$ ``break points'' in the sorted $\{A_{0,i}\}$.
Since $\hat{\mbf{A}}_n$ is uniformly consistent for $\mbf{A}_0$, these break points appear also in the following sequence: $\hat A_{n,(1)}, \ldots , \hat A_{n,(n)}$ w.p.a.1.
For $1 \le i < j \le n$, we define $\hat \Delta^A(i,j)$ as the sum of squared variations over $\{\hat A_{n,(i)}, \ldots, \hat A_{n,(j)}\}$; namely,
\begin{align*}
    \hat \Delta^A(i,j)
    & \equiv \sum_{l = i}^j (\hat A_{n,(l)} - \bar A_{n,i,j})^2,
    \;\; \text{where} \;\; \bar A_{n,i,j} \equiv \frac{1}{j - i + 1}\sum_{l = i}^j \hat A_{n,(l)}.
\end{align*}
Further, we define
\begin{align*}
    \hat S^A_{i,j}(\kappa) 
    & \equiv \left\{ \begin{array}{ll}
        \frac{1}{j - i + 1} \left( \hat \Delta^A(i,\kappa) + \hat \Delta^A(\kappa + 1, j) \right) & \text{if} \;\; j < \kappa \\
        \frac{1}{j - i + 1} \hat \Delta^A(i,j) & \text{if} \;\; j = \kappa.
    \end{array}\right.
\end{align*}
That is, $\hat S^A_{i,j}(\kappa)$ provides the total variance of $\{\hat A_{n,(i)}, \ldots, \hat A_{n,(j)}\}$ when a break point is placed at $\kappa$.
Assuming that $K^A \ge 2$, the BS algorithm proceeds as follows:

\begin{flushleft}
\textbf{Step 1 $(K^A = 2)$:}
\end{flushleft}

    We find the first break point, say $\hat t_1$, by
    \begin{align*}
        \hat t_1 = \argmin_{1 \le \kappa < n} \hat S^A_{1,n}(\kappa).
    \end{align*}
    Then, we can partition $\{\hat A_{n,(i)}\}$ into the following two subsets: $\{\hat A_{n,(i)}\} = \{\hat A_{n,(i)}\}_{i=1}^{\hat t_1} \bigcup \{\hat A_{n,(i)}\}_{i = \hat t_1 + 1}^n$.
    If $K^A = 2$, assuming that $a_{0,1} < a_{0,2}$ without loss of generality, we obtain $\hat{\mcl{C}}_{n,1}^A \equiv \{i : \hat A_{n,(1)} \le \hat A_{n,i} \le \hat A_{n,(\hat t_1)}\}$ and $\hat{\mcl{C}}_{n,2}^A \equiv \{i : \hat A_{n,(\hat t_1 + 1)} \le \hat A_{n,i} \le \hat A_{n,(n)}\}$ as the estimators of $\mcl{C}_{0,1}^A$ and $\mcl{C}_{0,2}^A$, respectively, and the algorithm stops.
    If $K^A > 2$, we proceed to the next step.

\begin{flushleft}
\textbf{Step 2 $(K^A = 3)$:}
\end{flushleft}

    Now, if $K^A = 3$, there exists one more break point either in $\{\hat A_{n,(i)}\}_{i=1}^{\hat t_1}$ or in $\{\hat A_{n,(i)}\}_{i = \hat t_1 + 1}^n$.
    In other words, either one of the two converges to a sequence of constants as the sample size increases.
    Then, when we compute $\hat S^A_{1,\hat t_1}(\hat t_1)$ and $\hat S^A_{\hat t_1 + 1, n}(n)$, and if $\hat S^A_{1,\hat t_1}(\hat t_1) > \hat S^A_{\hat t_1 + 1, n}(n)$ for example, the break point is likely to lie in the former subset.
    Thus, the second break point, say $\hat t_2$, can be found by
    \begin{align*}
        \hat t_2 = \begin{cases}
            \argmin_{1 \le \kappa < \hat t_1} \hat S^A_{1,\hat t_1}(\kappa) 
                & \text{if} \;\; \hat S^A_{1,\hat t_1}(\hat t_1) > \hat S^A_{\hat t_1 + 1, n}(n) \\
            \argmin_{\hat t_1 + 1 \le \kappa < n} \hat S^A_{\hat t_1 + 1, n}(\kappa) 
                & \text{if} \;\; \hat S^A_{1,\hat t_1}(\hat t_1) \le \hat S^A_{\hat t_1 + 1, n}(n).
        \end{cases}
    \end{align*}
    Then, we add $\hat t_2$ to the set of break points, and (with a little abuse of notation) we sort and re-label the points to ensure that $\hat t_1 < \hat t_2$.
    The resulting partition is given by $\{\hat A_{n,(i)}\} = \{\hat A_{n,(i)}\}_{i=1}^{\hat t_1} \bigcup \{\hat A_{n,(i)}\}_{i = \hat t_1 + 1}^{\hat t_2}\bigcup \{\hat A_{n,(i)}\}_{i = \hat t_2 + 1}^n$.
    Setting $a_{0,1} < a_{0,2} < a_{0,3}$ without loss of generality, we can define $\hat{\mcl{C}}_{n,1}^A \equiv \{i : \hat A_{n,(1)} \le \hat A_{n,i} \le \hat A_{n,(\hat t_1)}\}$, $\hat{\mcl{C}}_{n,2}^A \equiv \{i : \hat A_{n,(\hat t_1 + 1)} \le \hat A_{n,i} \le \hat A_{n,(\hat t_2)}\}$, and $\hat{\mcl{C}}_{n,3}^A \equiv \{i : \hat A_{n,(\hat t_2 + 1)} \le \hat A_{n,i} \le \hat A_{n,(n)}\}$ as the estimators of $\mcl{C}_{0,1}^A$, $\mcl{C}_{0,2}^A$, and $\mcl{C}_{0,3}^A$, respectively.

\begin{flushleft}
\textbf{Step 3 $(K^A \ge 4)$:}
\end{flushleft}

    When $K^A = 4$, if $\max\{\hat S^A_{1,\hat t_1}(\hat t_1), \hat S^A_{\hat t_1 + 1, \hat t_2}(\hat t_2), \hat S^A_{\hat t_2 + 1, n}(n)\} = \hat S^A_{1,\hat t_1}(\hat t_1)$ for example, the third break point should exist in $\{\hat A_{n,(i)}\}_{i=1}^{\hat t_1}$.
    Then, following the same procedure as above, we can obtain $\hat t_3 = \argmin_{1 \le \kappa < \hat t_1} \hat S^A_{1,\hat t_1}(\kappa) $.
    We repeat these steps until $K^A$ groups are detected with $K^A - 1$ break points.
    Finally, letting $\hat t_0 = 0$ and $\hat t_{K^A} = n$ so that $\hat t_0 < \hat t_1 < \cdots < \hat t_{K^A - 1} < \hat t_{K^A}$, each $k$-th group can be estimated by 
    \begin{align*}
        \hat{\mcl{C}}^A_{n,k} \equiv \{i : \hat A_{n,(\hat t_{k - 1} + 1)} \le \hat A_{n,i} \le \hat A_{n,(\hat t_k)}\}
    \end{align*}
    for $k = 1, \ldots, K^A$.
    Then, $\hat{\mcl{C}}^A_n \equiv \{\hat{\mcl{C}}^A_{n,1}, \ldots, \hat{\mcl{C}}^A_{n, K^A}\}$ is our estimator for the true group structure $\mcl{C}_0^A$.

\bigskip

In the same manner as above, we define $\hat{\mcl{C}}^B_n \equiv \{\hat{\mcl{C}}^B_{n,1}, \ldots, \hat{\mcl{C}}^B_{n, K^B}\}$ for the estimation of $\mcl{C}_0^B$.
Note that the identification of group membership can be achieved only up to ``label swapping''.
Thus, without loss of generality, we can label the groups according to the values of the group effects such that $a_{0,1} < a_{0,2} <\cdots < a_{0,K^A}$; that is, we define the $k$-th group as the group with the $k$-th smallest sender effect.
Similarly, we set $b_{0,1} < b_{0,2} <\cdots < b_{0,K^B}$.
To investigate the asymptotic properties of the BS algorithm, we introduce the following assumption.

\begin{assumption}\label{as:bs}
    (i) There exist constants $c_A, c_B > 0$ such that $\min_{k \neq k'}|a_{0,k} - a_{0,k'}| > c_A$ and $\min_{k \neq k'}|b_{0,k} - b_{0,k'}| > c_B$. (ii) $|\mcl{C}_{0,k}^A|/n \to \tau_k^A \in (0, 1)$ for all $k = 1, \ldots , K^A$ and $|\mcl{C}_{0,k}^B|/n \to \tau_k^B \in (0, 1)$ for all $k = 1, \ldots , K^B$.
\end{assumption}

Assumption \ref{as:bs} is parallel to Assumption A2 in \cite{wang2020identifying}.
Similar assumptions are commonly used in the literature of panel data models with latent group structure.
The following theorem provides the consistency result of the BS algorithm:

\begin{theorem}\label{thm:grouping}
    Suppose that Assumptions \ref{as:error}--\ref{as:Z_bound} and \ref{as:identification}--\ref{as:bs} hold.
    Then, we have $\Pr(\hat{\mcl{C}}^A_n = \mcl{C}_0^A) \to 1$ and $\Pr(\hat{\mcl{C}}^B_n = \mcl{C}_0^B) \to 1$.
\end{theorem}

Although the proof of Theorem \ref{thm:grouping} is almost analogous to \cite{ke2016structure} and \cite{wang2020identifying}, for completeness, we provide it in Appendix \ref{app:proofs}.

\begin{remark}[Fine-tuning]\upshape\label{rem:fine}
    \cite{bai1997estimating} shows that the BS method tends to over/underestimate the location of break points depending on the share of each group and the gaps between the values of the group effects.
    To account for this problem, he proposes the following repartitioning method.
    Let $\{\hat t_0, \ldots , \hat t_{K^A}\}$ be the estimated break points obtained by the standard BS method.
    Then, we replace the initial estimate $\hat t_k$ with $\hat t_k^\mathrm{repart} \equiv \argmin_{\hat t_{k-1} + 1 \le \kappa < \hat t_{k + 1}} \hat S^A_{\hat t_{k-1} + 1, \hat t_{k + 1}}(\kappa)$ for all $k = 1, \ldots, K^A - 1$.
    Letting $\hat{\mcl{C}}^{A,\mathrm{repart}}_n$ be the resulting repartitioned estimator of $\mcl{C}_0^A$, given the result of Theorem \ref{thm:grouping}, it is straightforward to see that  $\hat{\mcl{C}}^{A,\mathrm{repart}}_n$ is also consistent for $\mcl{C}_0^A$.
    Note that the above repartitioning procedure can be implemented recursively until convergence.
    The numerical simulations in Section \ref{sec:MC} indicate that using this repartitioning method remarkably improves the probability of correctly predicting the group memberships.
\end{remark}

\subsection{Post-grouping ML estimator}

In the final step, we solve \eqref{eq:full_mle} approximately by using $(\hat{\mcl{C}}^A_n, \hat{\mcl{C}}^B_n)$ in the place of $(\mcl{C}_0^A, \mcl{C}_0^B)$.
Recalling that $a_1$ is pinned at $a_1 = 0$, let $\delta \equiv (\theta^\top, a_2, \ldots, a_{K^A}, b_1, \ldots, b_{K^B})^\top$, and $\mbb{D}\equiv \Theta \times \mbb{A}^{K^A - 1} \times \mbb{B}^{K^B}$ for the parameter space of $\delta$.
We denote $\delta_0$ as the true value of $\delta$.
Then, our final ML estimator for $\delta_0$ is defined as
\begin{align*}
    \hat \delta_n = \argmax_{\delta \in \mbb{D}} \hat{\mcl{L}}_n(\delta),
\end{align*}
where $\hat{\mcl{L}}_n(\delta) \equiv \mcl{L}_n\left(\theta, \left\{ \sum_{k = 1}^{K^A} a_k \cdot \mbf{1}\{ i \in \hat{\mcl{C}}^A_{n,k}\} \right\}, \left\{ \sum_{k = 1}^{K^B} b_k \cdot \mbf{1}\{ i \in \hat{\mcl{C}}^B_{n,k}\} \right\} \right)$.
Similarly, we define
\begin{align*}
    \hat \delta_n^{\text{oracle}} = \argmax_{\delta \in \mbb{D}} \mcl{L}_n(\delta),
\end{align*}
where $\mcl{L}_n(\delta) \equiv \mcl{L}_n\left(\theta, \left\{ \sum_{k = 1}^{K^A} a_k \cdot \mbf{1}\{ i \in \mcl{C}^A_{0,k}\} \right\}, \left\{ \sum_{k = 1}^{K^B} b_k \cdot \mbf{1}\{ i \in \mcl{C}^B_{0,k}\} \right\} \right)$; that is, $\hat \delta_n^{\text{oracle}}$ is the ``oracle'' estimator that is computed based on the true $\mcl{C}_0^A$ and $\mcl{C}_0^B$.
Since $\hat \delta_n^{\text{oracle}}$ is the standard parametric ML estimator, the estimator follows a normal distribution asymptotically at the parametric rate, and its asymptotic covariance matrix is given by the inverse Fisher Information matrix.
Meanwhile, we have shown in Theorem \ref{thm:grouping} that the estimated group memberships $(\hat{\mcl{C}}^A_n, \hat{\mcl{C}}^B_n)$ are equal to $(\mcl{C}_0^A, \mcl{C}_0^B)$ w.p.a.1.
Therefore, we can claim that the final ML estimator $\hat \delta_n$ has asymptotically the same statistical performance as the oracle estimator $\hat \delta_n^{\text{oracle}}$, and, thus, it is asymptotically fully efficient.
We formally state this result in the next theorem.

\begin{theorem}\label{thm:normality}
    Suppose that Assumptions \ref{as:error}--\ref{as:Z_bound} and \ref{as:identification}--\ref{as:bs} hold.
    In addition, we assume that $\mcl{I}_{\delta\delta} \equiv - \lim_{n \to \infty} \E \left[ \partial^2 \mcl{L}_n(\delta_0)/(\partial \delta \partial \delta^\top) \right]$ exists and is positive definite.
    Then, $\hat \delta_n$ and $\hat \delta_n^{\mathrm{oracle}}$ have the same asymptotic distribution: $\sqrt{\frac{N}{2}} (\hat \delta_n^{\mathrm{oracle}} - \delta_0) \overset{d}{\to} N(\mbf{0}_{d_z + K^A + K^B + 1}, \mcl{I}_{\delta\delta}^{-1})$.
\end{theorem}

Recall that the asymptotic equivalence between $\hat \delta_n$ and $\hat \delta_n^{\text{oracle}}$ relies on the dense network structure where each agent's specific effects can be point-identified.
When the networks are not dense,  $\hat \delta_n$ is generally inconsistent, while $\hat \delta_n^{\text{oracle}}$ may be still consistent (potentially with a slower convergence rate).
Finally, note that the above discussions hold true if the repartitioned estimator $(\hat{\mcl{C}}^{A, \mathrm{repart}}_n, \hat{\mcl{C}}^{B, \mathrm{repart}}_n)$ is used instead of $(\hat{\mcl{C}}^A_n, \hat{\mcl{C}}^B_n)$.


\section{Monte Carlo Experiments}\label{sec:MC}

In this section, we examine the finite sample performance of the three-step ML estimator. 
We consider the following data-generating process for the Monte Carlo experiments:
\begin{align*}
    \begin{split}
       u_{i,j}(q) & = Z_{i,j,1} \beta_{0,1} + Z_{i,j,2} \beta_{0,2} + \alpha_0 q + \sum_{k = 1}^{K^A} a_{0,k} \cdot \mbf{1}\{ i \in \mcl{C}^A_{0,k}\} + \sum_{k = 1}^{K^B} b_{0,k} \cdot \mbf{1}\{ i \in \mcl{C}^B_{0,k}\} - \epsilon_{i,j} , \;\; \text{for} \;\; i \neq j,
    \end{split}
\end{align*}
where $Z_{i,j,1} = |X_i - X_j|$ with $X_i \overset{i.i.d.}{\sim} \mathrm{Uniform}[-1,1]$, $Z_{i,j,2} \overset{i.i.d.}{\sim} N(0, 1)$, $(\epsilon_{i,j}, \epsilon_{j,i})$ is i.i.d. across dyads as the standard bivariate normal with correlation coefficient $\rho_0 = 0.6$, $(\beta_{0,1}, \beta_{0,2}, \alpha_0) = (-1.2, 1.6, 0.6)$, and $K^A = K^B = 3$.
For the groupwise heterogeneity parameters, we consider $(a_{0,1},a_{0,2},a_{0,3}) = (0, r, 2r)$ and $ (b_{0,1},b_{0,2},b_{0,3}) = (-0.4 - r, -0.4, -0.4 + r)$ for $r \in \{0.4, 0.7, 1.0\}$.
The smaller (larger) $r$ becomes, the more difficult (easier) the identification of the group structure.
The group memberships are determined randomly while maintaining the equal size of each group.
Exceptionally for observation $1$, $1 \in \mcl{C}_{0,1}^A$ (so that $A_{0,1} = 0$) is fixed throughout the experiments.
For each model setup, we consider two sample sizes: $n \in \{54,75\}$; thus, the size of each group is 18 in the former case and is 25 in the latter.
The number of Monte Carlo repetitions is set to 500 for each single experiment.
For the estimation of the group memberships, for comparison, we use both the standard BS method without repartitions and the repartitioned BS method.

We first report the simulation results of estimating the common parameters $(\alpha_0, \beta_{0,1}, \beta_{0,2}, \rho_0)$.
Table \ref{table:common} presents the bias and RMSE (root mean squared error) for the following four estimators: the initial ML estimator given in \eqref{eq:mle} (1st-step ML), the three-step ML estimator based on the BS method with no iterations (BS0) and that with two iterations (BS2), and the oracle estimator based on the true group membership (Oracle). 
For estimating $(\beta_{0,1}, \beta_{0,2})$, as expected, the 1st-step ML estimator is largely biased for all scenarios due to the incidental parameter problem.
Although the three-step estimators (i.e., BS0 and BS2) also have some biases when $r = 0.4$ and $n = 54$, the biases disappear as either $r$ or $n$ increases.
Thus, these biases are probably due to frequent misclassification of group memberships under small $r$ and $n$.
In terms of RMSE, although we can observe a certain gap between the oracle estimator and the three-step estimators, the gaps can be reduced by increasing $r$ and $n$, which is consistent with our theory. 
Interestingly, even when using the 1st-step ML estimator, the strategic interaction effect and the error correlation parameter can be estimated with almost no bias.

\begin{table}[!h]
    \begin{center}
    \caption{Simulation Results: Estimation of Common Parameters}
    \begin{small}
    \begin{tabular}{cclcccccccc}
            \hline \hline
            &  &  & \multicolumn{2}{c}{$\alpha_0$} & \multicolumn{2}{c}{$\beta_{0,1}$} & \multicolumn{2}{c}{$\beta_{0,2}$} & \multicolumn{2}{c}{$\rho_0$} \\
            $r$ & $n$ & Estimator & Bias & RMSE & Bias & RMSE & Bias & RMSE & Bias & RMSE \\
            \hline
            0.4 & 54 & 1st-step ML & 0.027  & 0.194  & -0.192  & 0.358  & 0.222  & 0.257  & 0.047  & 0.196  \\
            &  & BS0 & 0.056  & 0.174  & -0.113  & 0.237  & 0.120  & 0.162  & -0.059  & 0.176  \\
            &  & BS2 & 0.057  & 0.176  & -0.128  & 0.254  & 0.134  & 0.175  & -0.050  & 0.176  \\
            &  & Oracle & 0.006  & 0.131  & -0.008  & 0.082  & 0.009  & 0.092  & 0.001  & 0.126  \\
            & 75 & 1st-step ML & 0.032  & 0.122  & -0.102  & 0.261  & 0.152  & 0.172  & 0.022  & 0.118  \\
            &  & BS0 & 0.043  & 0.116  & -0.056  & 0.176  & 0.080  & 0.108  & -0.054  & 0.125  \\
            &  & BS2 & 0.043  & 0.117  & -0.063  & 0.190  & 0.091  & 0.117  & -0.043  & 0.120  \\
            &  & Oracle & -0.001  & 0.093  & -0.002  & 0.061  & 0.004  & 0.066  & 0.004  & 0.091  \\
            \hline
            0.7 & 54 & 1st-step ML & 0.030  & 0.221  & -0.175  & 0.351  & 0.221  & 0.263  & 0.061  & 0.234  \\
            &  & BS0 & 0.020  & 0.176  & -0.056  & 0.220  & 0.056  & 0.125  & -0.072  & 0.193  \\
            &  & BS2 & 0.035  & 0.184  & -0.074  & 0.235  & 0.084  & 0.143  & -0.069  & 0.197  \\
            &  & Oracle & 0.007  & 0.132  & -0.010  & 0.079  & 0.010  & 0.094  & 0.002  & 0.128  \\
            & 75 & 1st-step ML & 0.031  & 0.127  & -0.104  & 0.267  & 0.152  & 0.173  & 0.031  & 0.126  \\
            &  & BS0 & 0.017  & 0.114  & -0.017  & 0.169  & 0.025  & 0.074  & -0.074  & 0.138  \\
            &  & BS2 & 0.025  & 0.118  & -0.038  & 0.183  & 0.048  & 0.086  & -0.058  & 0.132  \\
            &  & Oracle & -0.004  & 0.091  & -0.001  & 0.058  & 0.004  & 0.065  & 0.006  & 0.091  \\
            \hline
            1.0 & 54 & 1st-step ML & 0.027  & 0.234  & -0.178  & 0.356  & 0.227  & 0.267  & 0.082  & 0.268  \\
            &  & BS0 & -0.007  & 0.182  & 0.003  & 0.209  & -0.018  & 0.108  & -0.120  & 0.233  \\
            &  & BS2 & 0.013  & 0.186  & -0.036  & 0.214  & 0.029  & 0.114  & -0.092  & 0.224  \\
            &  & Oracle & 0.006  & 0.136  & -0.011  & 0.080  & 0.014  & 0.094  & 0.000  & 0.139  \\
            & 75 & 1st-step ML & 0.033  & 0.145  & -0.111  & 0.274  & 0.155  & 0.179  & 0.045  & 0.148  \\
            &  & BS0 & -0.009  & 0.130  & 0.029  & 0.173  & -0.039  & 0.085  & -0.103  & 0.172  \\
            &  & BS2 & 0.010  & 0.120  & -0.012  & 0.158  & 0.012  & 0.079  & -0.063  & 0.137  \\
            &  & Oracle & -0.001  & 0.094  & -0.002  & 0.059  & 0.006  & 0.066  & 0.005  & 0.097  \\
            \hline \hline
    \end{tabular}
    \end{small}
    \label{table:common}
    \end{center}
    \begin{small}
        Note.
        1st-step ML: the initial ML estimator,
        BS0: the three-step ML estimator based on the BS method with no repartitions,
        BS2: the three-step ML estimator based on the repartitioned BS method with two iterations, 
        Oracle: the oracle estimator based on the true group membership.
    \end{small}
\end{table}

The simulation results of estimating the group memberships are summarized in Table \ref{table:groupMC}.
Here, we compare the performance of BS0 and BS2 in terms of the ratio of correct group classification.
First of all, the results indicate that the repartitioned BS method (i.e., BS2) clearly outperforms the standard BS method without repartitions (i.e., BS0).  
As expected, as $r$ gets smaller, correctly predicting the group membership becomes significantly more difficult.
If the gaps between the values of the group effects are sufficiently large and the sample size is not small, BS2 can attain almost 90\% of correct classification.\footnote{
    One might view that the results reported in Table \ref{table:groupMC} are not particularly good for the BS algorithm.
    A main reason for that would be that our model is a bivariate binary response model, whereas most of the previous studies using the BS algorithm has focused on models with a continuous outcome.
}
We cannot observe any clear difference between the estimation of $\mcl{C}^A_0$ and that of $\mcl{C}^B_0$.

\begin{table}[!h]
    \begin{center}
    \caption{Simulation Results: BS Algorithm}
    \begin{small}
    \begin{tabular}{ccccc}
            \hline \hline
            &  &  & \multicolumn{2}{c}{Correct classification ratio}  \\
            $r$ & $n$ & Estimator & $\mcl{C}^A_0$ & $\mcl{C}^B_0$ \\
            \hline
            0.4 & 54 & BS0 & 0.533  & 0.533  \\
             &  & BS2 & 0.564  & 0.558  \\
             & 75 & BS0 & 0.569  & 0.571  \\
             &  & BS2 & 0.612  & 0.608  \\
            0.7 & 54 & BS0 & 0.638  & 0.645  \\
             &  & BS2 & 0.699  & 0.696  \\
             & 75 & BS0 & 0.706  & 0.699  \\
             &  & BS2 & 0.798  & 0.787  \\
            1.0 & 54 & BS0 & 0.720  & 0.712  \\
             &  & BS2 & 0.801  & 0.797  \\
             & 75 & BS0 & 0.782  & 0.783  \\
             &  & BS2 & 0.900  & 0.898  \\
            \hline \hline
    \end{tabular}
    \end{small}
    \label{table:groupMC}
    \end{center}
    \begin{small}
        Note.
        BS0: the three-step ML estimator based on the BS method with no repartitions,
        BS2: the three-step ML estimator based on the repartitioned BS method with two iterations.
    \end{small}
\end{table}


\section{ Empirical Application to International Visa-Free Travels}\label{sec:empiric}

As an empirical application of our model and method, we analyze the network of international visa-free travels.
The dependent variable of interest is $G_n = (g_{i,j})_{1 \le i,j \le n}$, where $g_{i,j} = 1$ if country $i$ allows the citizens in country $j$ to visit $i$ without visas, and $g_{i,j} = 0$ otherwise.
Since the bilateral relationship about visa-free policy is expected to be complementary, this would fit into our model framework.

In this empirical study, we consider 57 countries selected mainly from Asia, the Middle East, the former USSR, and Oceania.\footnote{
    The list of countries used in this empirical study is as follows: Armenia, Australia, Azerbaijan, Bahrain, Bangladesh, Belarus, Bhutan, Brunei, Cambodia, China, Cyprus, Estonia, Fiji, Georgia, Hong Kong, India, Indonesia, Iran, Iraq, Israel, Japan, Jordan, Kazakhstan, Kiribati, Kuwait, Kyrgyzstan, Laos, Latvia, Lebanon, Lithuania, Malaysia, Moldova, Mongolia, Myanmar, Nauru, Nepal, New Zealand, Oman, Pakistan, Papua New Guinea, Philippines, Qatar, Russia, Saudi Arabia, Singapore, South Korea, Sri Lanka, Tajikistan, Thailand, Tonga, Turkey, UAE, Ukraine, Uzbekistan, Vanuatu, Viet Nam, and Yemen.
    These countries are selected based on geographical proximity and ease of data collection.
    }
The information about the visa policy of each country is taken from \textit{Henly and Partners: Passport Index 2020} (\url{https://www.henleypassportindex.com/passport}).\footnote{
    Based on their definition, we categorize electronic travel authorization (eTA) and on-arrival visa as visa-free access.
    }
The total number of dyads in this network is $57(57 - 1)/2 = 1596$.
From Table \ref{table:dist_g}, which summarizes the distribution of the link connections, we can observe that the number of country pairs with one-way links is smaller than that with mutual links or no links.
This would suggest the presence of complementarity in the network formation process.  
According to the above-mentioned passport index, Japan, Singapore, and South Korea are the top three countries among the 57 countries in terms of the number of all countries with visa-free access.
For our restricted sample network, South Korea has the largest in-degree $\sum_i g_{i,\text{South Korea}} = 49$.
For the out-degree, Nepal has the largest value $\sum_j g_{\text{Nepal},j} = 55$; that is, Nepal allows 55 countries (out of 57) to visit Nepal only with on-arrival visas.
More detailed information can be found in Table \ref{table:degree} in Appendix \ref{subsec:empiric}.

\begin{table}[!h]
    \begin{center}
    \caption{Distribution of $\{(g_{i,j}, g_{j,i}): 1 \le i < j \le n \}$}
    \begin{tabular}{c|cc}
            \hline \hline
                          & $g_{j,i} = 0$ &  $g_{j,i} = 1$ \\ \hline
            $g_{i,j} = 0$ &	 495          &        343     \\
            $g_{i,j} = 1$ &  349          &        409     \\
        \hline \hline
    \end{tabular}
    \label{table:dist_g}
    \end{center}
\end{table}

The network for all the 57 countries is quite complicated and difficult to grasp the entire picture.
As one illustration of our data, Figure \ref{fig:asia} presents the sub-network obtained by restricting the vertices to the Eastern and Southeastern Asian countries.
The left panel in the figure shows the whole shape of this sub-network.
(Note that the direction of the arrows in the figure is ``not'' the direction of visa-free access, but it represents that the target country is allowed to visit the country at the arrow's origin without visas.)
The right panel shows the network created by leaving only one-way links from the left one; in other words, this is $(g_{i,j}\cdot (1 - g_{j,i}))_{i,j \in \text{Brunei}, \ldots, \text{Viet Nam}}$.
From this figure, we can expect the existence of a certain level of degree heterogeneity.
More specifically, Cambodia, for example, has five outgoing one-way links in this sub-network, suggesting that this country would have a larger sender effect $A$.
In contrast, countries such as Japan and South Korea would exhibit a larger receiver effect $B$.

\begin{figure}[!h]
	\centering
	\includegraphics[width = 18cm]{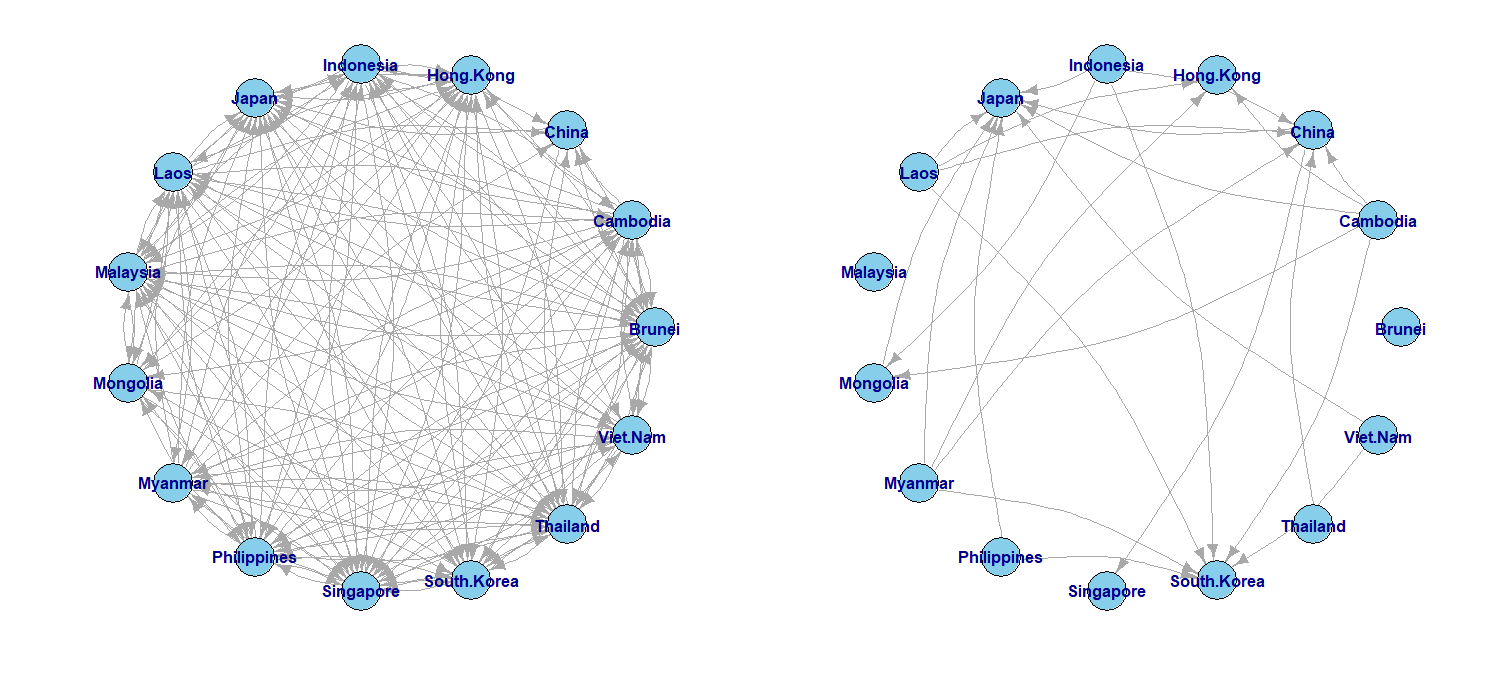}
    \caption{Eastern and Southeastern Asian sub-network}
    (Left panel: the whole sub-network, right panel: only one-way links.)
	\label{fig:asia}
\end{figure}

For estimating the network formation model, we consider five covariates; for their definitions, see Table \ref{table:def}.
The summary statistics of the covariates are provided in Table \ref{table:summary} in Appendix \ref{subsec:empiric}.
With these variables, we consider the following payoff function:
\begin{align*}
    u_{i,j}(q) 
    & = (\ln \textit{gdp\_pc}_i)(\ln \textit{gdp\_pc}_j) \beta_{0,1} + |\textit{free}_i - \textit{free}_j| \beta_{0,2} + \mbf{1}\{\textit{region}_i = \textit{region}_j\} \beta_{0,3} + \ln (\textit{export}_{ij} + 1)\beta_{0,4} \\
    & \quad + \ln (\textit{import}_{ij} + 1) \beta_{0,5} + \alpha_0 q + A_{0,i} + B_{0,j} - \epsilon_{i,j}, \;\; \text{for} \;\; i \neq j,
\end{align*}
where we assume that $(\epsilon_{i,j}, \epsilon_{j,i})$ have the standard bivariate normal distribution with correlation coefficient $\rho_0$.
To estimate our network formation model with grouped degree heterogeneity, we first need to determine the number of groups for the sender effects $\{A_{0,i}\}$ and that for the receiver effects $\{B_{0,i}\}$, $K^A$ and $K^B$, respectively.
Then, following \cite{ke2016structure} and \cite{wang2020identifying}, the optimal $(K^A, K^B)$ is selected as the minimizer of the BIC criterion: $-2\hat{\mcl{L}}_n(\hat \delta_n) + (6 + K^A + K^B)\ln (1596)$.
Then, as a result of searching over the models with $(K^A, K^B) \in \{2, \ldots, 7\}^2$, we find that the model with $(K^A, K^B) = (7, 6)$ achieves the smallest BIC (see Table \ref{table:BIC} in Appendix \ref{subsec:empiric} for more detailed information), and this is the model reported here.
For comparison, we estimate not only our proposed model, which we call the grouped heterogeneity model, but also a dyadic bivariate probit model without strategic interaction and degree heterogeneities as a benchmark.
For the estimation of the group memberships, we employ the repartitioning method with two iterations (i.e., BS2 in the previous section).

\begin{table}[!h]
    \begin{center}
    \caption{Definitions of Variables}
    \begin{tabular}{ll}
            \hline \hline
      \multicolumn{1}{c}{Variables} & \multicolumn{1}{c}{Definitions} \\ \hline
      $\textit{gdp\_pc}_i$ &	GDP per capita in 2018 (1,000 USD) \\
      $\textit{free}_i$	& Freedom rating (1 = Most Free, 7 = Least Free)$^a$ \\
      $\textit{region}_i$ & Categorical variable: East Asia, Southeast Asia, Central Asia, Europe, Middle East, or Oceania. \\
      $\textit{export}_{ij}$ & Total export value from country $i$ to $j$ in 2018 (million USD)$^b$ \\
      $\textit{import}_{ij}$ & Total import value from country $i$ to $j$ in 2018 (million USD)$^b$ \\
        \hline \hline
    \end{tabular}
    \label{table:def}
    \end{center}
    Sources: (a) \textit{Freedom in the World 2018}, Freedom House (\url{https://freedomhouse.org/}); (b) IMF DATA (\url{https://data.imf.org/}). 
\end{table}

The estimation results are summarized in Table \ref{table:result}.
First of all, as expected, our proposed model suggests that there is a significant strategic complementarity in the network formation behavior.
We can also find a certain level of degree heterogeneity in terms of both the sender and the receiver effects.
Comparing the grouped heterogeneity model and the benchmark model, the log-likelihood value for the former is apparently significantly larger than that for the latter.
This large difference in the degree of model fitting also demonstrates the significance of strategic effect and unobserved heterogeneity (note however that these models are not nested).
For specific parameter estimates, we can observe several non-negligible differences between the two models.
For example, the effect of the export amount is predicted to be positive in the grouped heterogeneity model, whereas the benchmark model predicts a significantly negative impact.
The error correlation parameter is not significantly different from zero in our model but is weakly positively significant in the benchmark model.
This result would be understandable since the benchmark model can account for the interdependence of the links only through the error correlation.
Additionally, there are several interesting findings.
For both models, if countries $i$ and $j$ are located in the same region, they become more likely to allow visa-free access, as expected.
Not only in terms of geographical proximity, but we can also observe significant homophily in terms of the political system. 

\begin{table}[!h]
    \begin{center}
    \caption{Estimation Results}
    \begin{tabular}{lcccc}
        \hline \hline
        & \multicolumn{2}{c}{\underline{Grouped Heterogeneity Model}} & \multicolumn{2}{c}{\underline{Benchmark Model}}  \\
        & Estimate & $t$-value & Estimate & $t$-value \\
        \hline
        Intercept &  &  & -0.208  & -1.196  \\
        Strategic effect: $\alpha$ & 0.482  & 3.842  &  &  \\
        $(\ln \textit{gdp\_pc}_j) (\ln \textit{gdp\_pc}_i)$: $\beta_1$ & 0.063  & 4.234  & 0.047  & 2.366  \\
        $|\textit{free}_i - \textit{free}_j|$: $\beta_2$ & -0.229  & -7.044  & -0.075  & -1.692  \\
        $\mbf{1}\{ \textit{region}_i = \textit{region}_j\} $: $\beta_3$ & 1.402  & 9.496  & 0.691  & 4.400  \\
        $\ln (\textit{export}_{ij} + 1)$: $\beta_4$ & 0.046  & 1.829  & -0.073  & -3.604  \\
        $\ln (\textit{import}_{ij} + 1)$: $\beta_5$ & 0.056  & 2.365  & 0.137  & 6.535  \\
        Error correlation: $\rho$ & -0.080  & -0.450  & 0.088  & 1.656  \\
        &  &  &  &  \\
        \multicolumn{1}{c}{\underline{Sender effects: $A$}} &  &  &  &  \\
        Group 1: $a_1$ & 0  & - &  &  \\
        Group 2: $a_2$ & 1.060  & 4.552  &  &  \\
        Group 3: $a_3$ & 2.051  & 10.199  &  &  \\
        Group 4: $a_4$ & 2.846  & 13.572  &  &  \\
        Group 5: $a_5$ & 3.813  & 16.302  &  &  \\
        Group 6: $a_6$ & 4.790  & 16.642  &  &  \\
        Group 7: $a_7$ & 6.105  & 15.387  &  &  \\
        \multicolumn{1}{c}{\underline{Receiver effects: $B$}} &  &  &  &  \\
        Group 1: $b_1$ & -5.670  & -14.388  &  &  \\
        Group 2: $b_2$ & -4.590  & -16.159  &  &  \\
        Group 3: $b_3$ & -3.902  & -14.559  &  &  \\
        Group 4: $b_4$ & -3.446  & -14.240  &  &  \\
        Group 5: $b_5$ & -2.561  & -10.952  &  &  \\
        Group 6: $b_6$ & -1.558  & -6.722  &  &  \\
        &  &  &  &  \\
        Log-likelihood & \multicolumn{2}{c}{-810.798} & \multicolumn{2}{c}{-1531.496} \\
        \# Observations & \multicolumn{2}{c}{1,596} & \multicolumn{2}{c}{1,596} \\
    \hline \hline
    \end{tabular}
    \label{table:result}
    \end{center}
\end{table}

For the estimation results of country-specific effects, we report the estimated group memberships in Table \ref{table:group}.
As expected from the above discussion, countries such as Cambodia are indeed classified into the highest group (i.e., Group 7) in terms of the sender effect.
The other two countries that have Group-7 sender effect are Nepal and Sri Lanka.
For the receiver effect, these two countries are classified as Group 2, and Cambodia is in Group 3.
Overall, interestingly, there seems to be a weak negative correlation between the sender effects and the receiver effects.
As expected from the above discussion, Japan and South Korea indeed belong to the group with the highest receiver effect (i.e., Group 6).
The magnitudes of the receiver effects seem to roughly correlate with the size of the countries' economies (with some exceptions, such as China, India, and Russia).

\begin{table}[!h]
    \begin{center}
    \caption{Estimated Group Memberships}
    \begin{small}
    \begin{tabular}{l|p{7cm}|p{7cm}}
        \hline\hline
     & \multicolumn{1}{|c|}{\textbf{Sender Effect}: $A$} & \multicolumn{1}{|c}{\textbf{Receiver Effect}: $B$}\\
        \hline
        Group 1	
        & Australia,
        China,
        Iraq,
        Nauru,
        Oman,
        Russia
        & Iraq,
        Pakistan
         \\
        \hline
        Group 2	
        & Bhutan,
        India,
        Japan,
        New Zealand,
        UAE        
        & Bangladesh,
        Iran,
        Jordan,
        Lebanon,
        Nepal,
        Sri Lanka,
        Yemen
         \\
        \hline
        Group 3	
        & Bahrain,
        Cyprus,
        Estonia,
        Georgia,
        Kiribati,
        Kuwait,
        Latvia,
        Lithuania,
        Mongolia,
        Myanmar,
        Papua New Guinea,
        Qatar,
        Saudi Arabia,
        South Korea,
        Tonga,
        Turkey,
        Ukraine,
        Viet Nam,
        Yemen
        &
        Armenia,
        Cambodia,
        Laos,
        Myanmar,
        Qatar,
        Russia,
        Viet Nam
        \\
        \hline
        Group 4 &
        Azerbaijan,
        Belarus,
        Brunei,
        Fiji,
        Hong Kong,
        Israel,
        Kazakhstan,
        Kyrgyzstan,
        Moldova,
        Pakistan,
        Singapore,
        Thailand,
        Uzbekistan        
        & Belarus,
        Bhutan,
        China,
        Fiji,
        Georgia,
        India,
        Indonesia,
        Kazakhstan,
        Kiribati,
        Kyrgyzstan,
        Moldova,
        Mongolia,
        Nauru,
        Papua New Guinea,
        Philippines,
        Tajikistan,
        Thailand,
        Tonga,
        Ukraine,
        Uzbekistan
        \\
        \hline
        Group 5	&
        Armenia,
        Bangladesh,
        Jordan,
        Lebanon,
        Malaysia,
        Philippines,
        Tajikistan,
        Vanuatu
        & Azerbaijan,
        Bahrain,
        Kuwait,
        Latvia,
        Lithuania,
        Oman,
        Saudi Arabia,
        Turkey,
        UAE,
        Vanuatu
        \\
        \hline
        Group 6	&
        Indonesia,
        Iran,
        Laos
        & Australia,
        Brunei,
        Cyprus,
        Estonia,
        Hong Kong,
        Israel,
        Japan,
        Malaysia,
        New Zealand,
        Singapore,
        South Korea
        \\
        \hline
        Group 7	&
        Cambodia,
        Nepal,
        Sri Lanka
        &        \\
        \hline\hline
    \end{tabular}
    \end{small}
    \label{table:group}
    \end{center}
\end{table} 


\section{Conclusion}\label{sec:conclusion}

This paper proposed a network formation model with pairwise strategic interaction and grouped degree heterogeneity.
Assuming some parametric form for the error distribution, we proved that the model parameters can be identified under the availability of agent-specific covariates that have large supports and also have variations across all potential partners.
For estimating the model, based on the same idea as in \cite{bresnahan1990entry} and \cite{berry1992estimation}, we proposed the three-step ML procedure: 
in the first-step, the model is estimated without considering the group structure; subsequently, we estimate the group memberships using the BS algorithm given the estimates for the heterogeneity parameters obtained in the first step; and, finally, based on the estimated group memberships, we re-estimate the model. 
Under certain regularity conditions, we showed that the proposed estimator is asymptotically unbiased and distributed as normal at the parametric rate.
The results of the Monte Carlo simulations show that our estimator performs reasonably well in finite samples. 
An empirical application to international visa-free travel networks indicates the usefulness of the proposed
model.

Several limitations and extensions are as follows.
First, our approach can be used only in pairwise network formation games with no network externalities to/from the rest of the links, and this limits the empirical applicability.
Therefore, it would be worthwhile to extend our results to network formation models with general network externalities involving more than two agents.
However, we conjecture that we would resort to partial identification to achieve this.
Second, our approach requires that the degree heterogeneity parameters have discrete support, although, in reality, it is possible that they are continuous.
To address this issue, it is of interest to modify our model in a similar manner to \cite{bonhomme2017discretizing} and investigate the three-step ML estimator in which $K^A$ and $K^B$ grows slowly to infinity.
Third, as our model is a dyadic binary game model, where a pairwise network formation model is its special case, we can consider its ordered-response game version as a natural extension.
For example, we might be interested in analyzing bilateral military relations: non-alliance, quasi-alliance, or alliance.
We expect that such extension can be relatively easily achieved by adopting the ML estimator discussed in \cite{aradillas2019inference}.
Finally, related to the empirical application in this study, we might be interested in investigating the causal effect of visa policies between two countries on the flows of tourists between them; this is a dyadic treatment evaluation problem when the treatment variable is determined strategically.
To deal with such situations, combining the results of this study and the marginal treatment effect framework developed in \cite{hoshino2020treatment} would be beneficial. 
We leave these topics for future research.


\clearpage
\small
\appendix

\begin{center}
	\LARGE Appendix
\end{center}

\section{Notations}\label{sec:notations}

\paragraph{Variables and parameters}
\begin{align*}
    W_{i,j} 
    & \equiv (Z_{i,j}^\top, \chi_{n,i}^\top, \chi_{n,j}^\top)^\top \\
    \theta
    & \equiv (\beta^\top, \alpha, \rho)^\top \\
    \mbf{A}
    & \equiv (A_1, \ldots, A_n)^\top, \quad \mbf{A}_{-1} \equiv (A_2, \ldots, A_n)^\top, \quad  \mbf{B}
    \equiv (B_1, \ldots, B_n)^\top \\
    \bgamma 
    & \equiv (\mbf{A}^\top, \mbf{B}^\top)^\top, \quad \bgamma_{-1} \equiv (\mbf{A}_{-1}^\top \mbf{B}^\top)^\top \\
    \Pi 
    & \equiv (\beta^\top, \bgamma^\top)^\top \\
    \pi_{i,j} 
    & \equiv Z_{i,j}^\top\beta + A_i + B_j = W_{i,j}^\top \Pi \\
    \delta 
    & \equiv (\theta^\top, a_2, \ldots, a_{K^A}, b_1, \ldots , b_{K^B})^\top.
\end{align*}
\paragraph{Functions and derivatives}
Throughout this appendix, for notational simplicity, we denote $\partial_a g(a) = \partial g(a) / (\partial a)$, $\partial^2_{ab} g(a, b) = \partial^2 g(a,b) / (\partial a \partial b)$, and so fourth.
\begin{align*}
    P_{i,j}(\theta, \bgamma) 
    & \equiv F(W_{i,j}^\top \Pi) - H(W_{i,j}^\top \Pi, W_{j,i}^\top \Pi + \alpha; \rho) \\
    \ell_{i,j}(\theta, \bgamma)
    & \equiv y_{i,j} \ln P_{i,j}(\theta, \bgamma) + y_{j,i} \ln P_{j,i}(\theta, \bgamma) + (1 - y_{i,j} - y_{j,i}) \ln (1 - P_{i,j}(\theta, \bgamma) - P_{j,i}(\theta, \bgamma) ) \\
    \mcl{L}_n(\theta, \bgamma)
    & \equiv \frac{2}{N} \sum_{i = 1}^n \sum_{j > i} \ell_{i,j}(\theta, \bgamma)\\
    p_{1,i,j}(\theta, \bgamma) 
    & \equiv \partial_{\pi_{i,j}} P_{i,j}(\theta, \bgamma) = f(W_{i,j}^\top \Pi) - H_1(W_{i,j}^\top \Pi, W_{j,i}^\top \Pi + \alpha; \rho) \\
    p_{2,i,j}(\theta, \bgamma)
    & \equiv \partial_{\pi_{j,i}} P_{i,j}(\theta, \bgamma) = - H_2(W_{i,j}^\top \Pi, W_{j,i}^\top \Pi + \alpha; \rho) \\
    H_\rho(\cdot, \cdot, ; \rho) 
    & \equiv \partial_\rho H(\cdot, \cdot; \rho) ,
\end{align*}
where $f$ is the derivative of $F$, and $H_l(\cdot, \cdot ; \rho)$ is the derivative of $H(\cdot, \cdot ; \rho)$ with respect to the $l$-th argument.
\begin{align*}
    \partial_{A_k} P_{i,j}(\theta, \bgamma)
    & = p_{1,i,j}(\theta, \bgamma)\mbf{1}\{ i = k \} + p_{2,i,j}(\theta, \bgamma)\mbf{1}\{ j = k \} \\
    \partial_{B_k} P_{i,j}(\theta, \bgamma)
    & = p_{1,i,j}(\theta, \bgamma)\mbf{1}\{ j = k \} + p_{2,i,j}(\theta, \bgamma)\mbf{1}\{ i = k \} \\
    s_{1,i,j}(\theta, \bgamma) 
    & \equiv \partial_{\pi_{i,j}} \ell_{i,j}(\theta, \bgamma) = \frac{y_{i,j} p_{1,i,j}(\theta, \bgamma)}{P_{i,j}(\theta, \bgamma)} + \frac{y_{j,i} p_{2,j,i}(\theta, \bgamma)}{P_{j,i}(\theta, \bgamma)} - \frac{(1 - y_{i,j} - y_{j,i})[p_{1,i,j}(\theta, \bgamma) + p_{2,j,i}(\theta, \bgamma)]}{1 - P_{i,j}(\theta, \bgamma) - P_{j,i}(\theta, \bgamma)} \\
    s_{2,i,j}(\theta, \bgamma) 
    & \equiv \partial_{\pi_{j,i}} \ell_{i,j}(\theta, \bgamma) = \frac{y_{i,j} p_{2,i,j}(\theta, \bgamma)}{P_{i,j}(\theta, \bgamma)} + \frac{y_{j,i} p_{1,j,i}(\theta, \bgamma)}{P_{j,i}(\theta, \bgamma)} - \frac{(1 - y_{i,j} - y_{j,i})[p_{2,i,j}(\theta, \bgamma) + p_{1,j,i}(\theta, \bgamma)]}{1 - P_{i,j}(\theta, \bgamma) - P_{j,i}(\theta, \bgamma)} \\
    s^{A_k}_{i,j}(\theta, \bgamma) 
    & \equiv \partial_{A_k} \ell_{i,j}(\theta, \bgamma) = s_{1,i,j}(\theta, \bgamma)\mbf{1}\{ i = k \} + s_{2,i,j}(\theta, \bgamma)\mbf{1}\{ j = k \} \\
    s^{B_k}_{i,j}(\theta, \bgamma) 
    & \equiv \partial_{B_k} \ell_{i,j}(\theta, \bgamma) = s_{1,i,j}(\theta, \bgamma)\mbf{1}\{ j = k \} + s_{2,i,j}(\theta, \bgamma)\mbf{1}\{ i = k \} \\
    s^\mbf{A}_{i,j}(\theta, \bgamma) 
    & \equiv \partial_{\mbf{A}_{-1}} \ell_{i,j}(\theta, \bgamma) = s_{1,i,j}(\theta, \bgamma)\chi_{n,i,-1} + s_{2,i,j}(\theta, \bgamma)\chi_{n,j,-1} \\
    s^\mbf{B}_{i,j}(\theta, \bgamma) 
    & \equiv \partial_\mbf{B} \ell_{i,j}(\theta, \bgamma) = s_{1,i,j}(\theta, \bgamma)\chi_{n,j} + s_{2,i,j}(\theta, \bgamma)\chi_{n,i} \\
    \xi_{i,j}(\theta, \bgamma) 
    & \equiv \partial_\theta P_{i,j}(\theta, \bgamma) = \left[ p_{1,i,j}(\theta, \bgamma)Z_{i,j}^\top + p_{2,i,j}(\theta, \bgamma)Z_{j,i}^\top, \; p_{2,i,j}(\theta, \bgamma), \; - H_\rho(W_{i,j}^\top \Pi, W_{j,i}^\top \Pi + \alpha; \rho)\right]^\top \\
    s^\theta_{i,j}(\theta, \bgamma) 
    & \equiv  \partial_\theta \ell_{i,j}(\theta, \bgamma) = \frac{y_{i,j} \xi_{i,j}(\theta, \bgamma)}{P_{i,j}(\theta, \bgamma)} + \frac{y_{j,i} \xi_{j,i}(\theta, \bgamma)}{P_{j,i}(\theta, \bgamma)} - \frac{(1 - y_{i,j} - y_{j,i})[\xi_{i,j}(\theta, \bgamma) + \xi_{j,i}(\theta, \bgamma)]}{1 - P_{i,j}(\theta, \bgamma) - P_{j,i}(\theta, \bgamma)},
\end{align*}
where $\chi_{n,i,-1}$ and $\chi_{n,j,-1}$ are $(n - 1) \times 1$ vectors defined by removing the first element of $\chi_{n,i}$ and $\chi_{n,j}$, respectively.
Using these notations, we can write
\begin{align*}
    \mcl{S}_{n, \theta}(\theta, \bgamma) 
    & \equiv \partial_\theta \mcl{L}_n(\theta, \bgamma) = \frac{2}{N}\sum_{i = 1}^n \sum_{j > i}  s^\theta_{i,j}(\theta, \bgamma) \\
    \mcl{S}_{n,\bgamma}(\theta, \bgamma) 
    & \equiv \partial_{\bgamma_{-1}} \mcl{L}_n(\theta, \bgamma) = \frac{2}{N} \sum_{i = 1}^n \sum_{j > i} \left[ s_{i,j}^\mbf{A} (\theta, \bgamma)^\top,  \: s_{i,j}^\mbf{B} (\theta, \bgamma)^\top \right]^\top.
\end{align*}
Further, writing $\ell_{i,j}(\delta) \equiv \ell_{i,j}\left(\theta, \left\{ \sum_{k = 1}^{K^A} a_k \cdot \mbf{1}\{ i \in \mcl{C}^A_{0,k}\} \right\}, \left\{ \sum_{k = 1}^{K^B} b_k \cdot \mbf{1}\{ i \in \mcl{C}^B_{0,k}\} \right\} \right)$ so that $\mcl{L}_n(\delta) = \frac{2}{N}\sum_{i = 1}^n \sum_{j > i} \ell_{i,j}(\delta) $, we define
\begin{align*}
    s^{a_k}_{i,j}(\delta) 
    & \equiv \partial_{a_k} \ell_{i,j}(\delta) = s_{1,i,j}(\delta)\mbf{1}\{ i \in \mcl{C}_{0,k}^A \} + s_{2,i,j}(\delta)\mbf{1}\{ j \in \mcl{C}_{0,k}^A \} \\
    s^{b_k}_{i,j}(\delta) 
    & \equiv \partial_{b_k} \ell_{i,j}(\delta) = s_{1,i,j}(\delta)\mbf{1}\{ j \in \mcl{C}_{0,k}^B \} + s_{2,i,j}(\delta)\mbf{1}\{ i \in \mcl{C}_{0,k}^B \} \\
    s^\delta_{i,j}(\delta) 
    & \equiv \partial_{\delta} \ell_{i,j}(\delta) =  \left[ s^\theta_{i,j}(\delta)^\top, \; s^{a_2}_{i,j}(\delta), \ldots, s^{a_{K^A}}_{i,j}(\delta), s^{b_1}_{i,j}(\delta), \ldots, s^{b_{K^B}}_{i,j}(\delta) \right]^\top,
\end{align*}
where the definitions of $s_{1,i,j}(\delta)$, $s_{2,i,j}(\delta)$, and $s^\theta_{i,j}(\delta)$ should be clear from the context.

\paragraph{Hessian matrix}
Define
\begin{align*}
    h_{11,i,j}(\theta, \bgamma) 
    & \equiv \partial_{\pi_{i,j}} s_{1,i,j}(\theta, \bgamma)\\
    h_{12,i,j}(\theta, \bgamma) 
    & \equiv \partial_{\pi_{j,i}} s_{1,i,j}(\theta, \bgamma) = \partial_{\pi_{i,j}} s_{2,i,j}(\theta, \bgamma)  \\
    h_{22,i,j}(\theta, \bgamma)
    & \equiv \partial_{\pi_{j,i}} s_{2,i,j}(\theta, \bgamma).
\end{align*}
It is easy to see that
\begin{align}\label{eq:Eh}
    \begin{split}
    \E h_{11,i,j}(\theta_0,\bgamma_0)
    & = - \frac{p^2_{1,i,j}}{P_{i,j}} - \frac{p^2_{2,j,i}}{P_{j,i}} - \frac{[p_{1,i,j} + p_{2,j,i}]^2}{1 - P_{i,j} - P_{j,i}}\\
    \E h_{12,i,j}(\theta_0,\bgamma_0) 
    & = - \frac{p_{1,i,j} p_{2,i,j}}{P_{i,j}} - \frac{p_{1,j,i}p_{2,j,i}}{P_{j,i}} - \frac{[p_{1,i,j} + p_{2,j,i}][p_{2,i,j} + p_{1,j,i}]}{1 - P_{i,j} - P_{j,i}},
    \end{split}
\end{align}
where we have used $p_{1,i,j}$, $p_{2,i,j}$, and $P_{i,j}$ to denote $p_{1,i,j}(\theta_0,\bgamma_0)$, $p_{2,i,j}(\theta_0,\bgamma_0)$, and $P_{i,j}(\theta_0,\bgamma_0)$, respectively, for simplicity.
Hereinafter, when the dependence on the parameters $(\theta, \bgamma)$ is suppressed, it means that the functions are evaluated at the true value $(\theta_0,\bgamma_0)$.

Note that, since $\ell_{i,j}(\theta, \bgamma) = \ell_{j,i}(\theta, \bgamma)$, we have $h_{11,i,j}(\theta, \bgamma) = \partial^2_{\pi_{i,j} \pi_{i,j}} \ell_{i,j}(\theta, \bgamma) = \partial^2_{\pi_{i,j} \pi_{i,j}} \ell_{j,i}(\theta, \bgamma) = h_{22,j,i}(\theta, \bgamma)$ and $h_{12,i,j}(\theta, \bgamma) = h_{12,j,i}(\theta, \bgamma)$.
By tedious calculations, we have
\begin{align*}
    \partial_{A_l A_k}^2 \mcl{L}_n(\theta, \bgamma)
    & = \frac{2}{N}\sum_{j \neq k} h_{11,k,j}(\theta, \bgamma)\mbf{1}\{l = k\} + \frac{2}{N} h_{12,l,k}(\theta, \bgamma)\mbf{1}\{l \neq k\} \;\; (\text{for} \; l,k \ge 2)\\
    \partial_{B_l B_k}^2 \mcl{L}_n(\theta, \bgamma)
    & = \frac{2}{N}\sum_{j \neq k} h_{11,j,k}(\theta, \bgamma)\mbf{1}\{l = k\} + \frac{2}{N} h_{12,l,k}(\theta, \bgamma)\mbf{1}\{l \neq k\} \;\; (\text{for} \; l,k \ge 1)\\
    \partial_{A_l B_k}^2 \mcl{L}_n(\theta, \bgamma)
    & = \frac{2}{N}\sum_{j \neq k} h_{12,k,j}(\theta, \bgamma)\mbf{1}\{l = k\} + \frac{2}{N} h_{11,l,k}(\theta, \bgamma)\mbf{1}\{l \neq k\} \;\; (\text{for} \; l \ge 2, \: k \ge 1).
\end{align*}
Hence,
\begin{align}\label{eq:hessian}
    \begin{split}
    \mcl{H}_{n, \mbf{A}\mbf{A}} (\theta, \bgamma) 
    & \equiv \partial_{\mbf{A}_{-1} \mbf{A}_{-1}^\top}^2 \mcl{L}_n(\theta, \bgamma) 
    = \frac{2}{N} \left( 
        \begin{array}{ccc}    
            \sum_{j \neq 2} h_{11,2,j}(\theta, \bgamma) & \cdots & h_{12,2,n}(\theta, \bgamma) \\
            \vdots & \ddots & \vdots \\
            h_{12,n,2}(\theta, \bgamma) & \cdots &  \sum_{j \neq n} h_{11,n,j}(\theta, \bgamma) 
        \end{array}
        \right) \\
    \mcl{H}_{n, \mbf{B}\mbf{B}} (\theta, \bgamma)
    & \equiv \partial_{\mbf{B} \mbf{B}^\top}^2 \mcl{L}_n(\theta, \bgamma)
    = \frac{2}{N} \left( 
        \begin{array}{ccc}    
            \sum_{j \neq 1} h_{11,j,1}(\theta, \bgamma) & \cdots & h_{12,1,n}(\theta, \bgamma) \\
            \vdots & \ddots & \vdots \\
            h_{12,n,1}(\theta, \bgamma) & \cdots &  \sum_{j \neq n} h_{11,j,n}(\theta, \bgamma) 
        \end{array}
        \right) \\
    \mcl{H}_{n, \mbf{A}\mbf{B}} (\theta, \bgamma)
    & \equiv \partial_{\mbf{A}_{-1} \mbf{B}^\top}^2 \mcl{L}_n(\theta, \bgamma)
    = \frac{2}{N} \left( 
        \begin{array}{cccc}    
            h_{11,2,1}(\theta, \bgamma) & \sum_{j \neq 2} h_{12,2,j}(\theta, \bgamma) & \cdots & h_{11,2,n}(\theta, \bgamma) \\
            \vdots & \vdots & \ddots & \vdots \\
            h_{11,n,1}(\theta, \bgamma) & h_{11,n,2}(\theta, \bgamma) & \cdots &  \sum_{j \neq n} h_{12,n,j}(\theta, \bgamma) 
        \end{array}
        \right) \\
    \mcl{H}_{n, \bgamma\bgamma} (\theta, \bgamma)
    & = \left( 
        \begin{array}{cc}    
            \mcl{H}_{n, \mbf{A}\mbf{A}} (\theta, \bgamma) & \mcl{H}_{n, \mbf{A}\mbf{B}} (\theta, \bgamma) \\
            \mcl{H}_{n, \mbf{B}\mbf{A}} (\theta, \bgamma) & \mcl{H}_{n, \mbf{B}\mbf{B}} (\theta, \bgamma)
        \end{array}
        \right)
    \end{split}
\end{align}


\section{Proofs of Theorems}\label{app:proofs}

\textbf{Proof of Theorem \ref{thm:identification}}

(i) We first confirm that the true parameter vector $(\theta_0, \bgamma_0)$ is a maximizer of $\E\mcl{L}_n(\theta, \bgamma)$.
We can observe that
\begin{align}\label{eq:ll_diff}
    \begin{split}
        \E\mcl{L}_n(\theta, \bgamma) - \E \mcl{L}_n(\theta_0, \bgamma_0) 
        & = \frac{1}{N}\sum_{i = 1}^n \sum_{j \neq i} \E\left\{ \ln \left[ P_{i,j}(\theta, \bgamma)^{y_{i,j}} P_{j,i}(\theta, \bgamma)^{y_{j,i}} [1 - P_{i,j}(\theta, \bgamma) - P_{j,i}(\theta, \bgamma)]^{1 - y_{i,j} - y_{j,i}} \right] \right. \\
        & \quad \left. - \ln \left[ P_{i,j}^{y_{i,j}} P_{j,i}^{y_{j,i}} [1 - P_{i,j} - P_{j,i}]^{1 - y_{i,j} - y_{j,i}} \right] \right\} \\
        &  = \frac{1}{N}\sum_{i = 1}^n \sum_{j \neq i} \E\left\{ \ln  \frac{ P_{i,j}(\theta, \bgamma)^{y_{i,j}} P_{j,i}(\theta, \bgamma)^{y_{j,i}} [1 - P_{i,j}(\theta, \bgamma) - P_{j,i}(\theta, \bgamma)]^{1 - y_{i,j} - y_{j,i}} }{P_{i,j}^{y_{i,j}} P_{j,i}^{y_{j,i}} [1 - P_{i,j} - P_{j,i}]^{1 - y_{i,j} - y_{j,i}} } \right\} \\
        & \le  \frac{1}{N}\sum_{i = 1}^n \sum_{j \neq i} \ln \E\left\{ \frac{ P_{i,j}(\theta, \bgamma)^{y_{i,j}} P_{j,i}(\theta, \bgamma)^{y_{j,i}} [1 - P_{i,j}(\theta, \bgamma) - P_{j,i}(\theta, \bgamma)]^{1 - y_{i,j} - y_{j,i}} }{P_{i,j}^{y_{i,j}} P_{j,i}^{y_{j,i}} [1 - P_{i,j} - P_{j,i}]^{1 - y_{i,j} - y_{j,i}} } \right\},
    \end{split}
    \end{align}
    where the last inequality follows from Jensen's inequality.
    Further,
    \begin{align*}
       & \E\left\{ \frac{ P_{i,j}(\theta, \bgamma)^{y_{i,j}} P_{j,i}(\theta, \bgamma)^{y_{j,i}} [1 - P_{i,j}(\theta, \bgamma) - P_{j,i}(\theta, \bgamma)]^{1 - y_{i,j} - y_{j,i}} }{P_{i,j}^{y_{i,j}} P_{j,i}^{y_{j,i}} [1 - P_{i,j} - P_{j,i}]^{1 - y_{i,j} - y_{j,i}} } \right\} \\
       & = \E[y_{i,j}] \frac{ P_{i,j}(\theta, \bgamma)}{P_{i,j}} + \E [y_{j,i}] \frac{ P_{j,i}(\theta, \bgamma)}{P_{j,i}} + \E[1 - y_{i,j} - y_{j,i}] \frac{ 1 - P_{i,j}(\theta, \bgamma) - P_{j,i}(\theta, \bgamma) }{1 - P_{i,j} - P_{j,i}} = 1,
    \end{align*}
implying that the left-hand side term of \eqref{eq:ll_diff} is less than or at most equal to zero for any given $(\theta, \bgamma)$.
Then, since $\rho_0$ is known, it is sufficient to show that 
\begin{align*}
    \frac{1}{N}\sum_{i = 1}^n \sum_{j \neq i}\mbf{1}\left\{ P_{i,j}((\beta^\top, \alpha, \rho_0)^\top, \bgamma) \neq  P_{i,j}(\theta_0, \bgamma_0) \right\} > 0
\end{align*} 
for sufficiently large $n$ and for all $(\beta, \alpha, \bgamma) \in \mcl{B} \times \mcl{A} \times \mbb{C}_n$ such that $(\beta, \alpha, \bgamma) \neq (\beta_0, \alpha_0, \bgamma_0)$.
The existence of pairs satisfying $P_{i,j}((\beta^\top, \alpha, \rho_0)^\top, \bgamma) \neq  P_{i,j}(\theta_0, \bgamma_0)$ contributes to a non-negligible difference between $\E\mcl{L}_n(\theta, \bgamma)$ and $\E \mcl{L}_n(\theta_0, \bgamma_0)$, allowing us to distinguish $(\theta, \bgamma)$ and $(\theta_0, \bgamma_0)$.
Here, by Assumptions \ref{as:error}(i) and (ii), $F(a) - H(a, b; \rho_0)$ is strictly increasing in $a$ and decreasing in $b$, respectively.
Therefore, we have
\begin{align*}
    W_{i,j}^\top \Pi >  W_{i,j}^\top \Pi_0, \;\;  W_{j,i}^\top \Pi + \alpha < W_{j,i}^\top \Pi_0 + \alpha_0 
    & \Longrightarrow P_{i,j}((\beta^\top, \alpha, \rho_0)^\top, \bgamma) > P_{i,j}(\theta_0, \bgamma_0)\\
    W_{i,j}^\top \Pi <  W_{i,j}^\top \Pi_0, \;\;  W_{j,i}^\top \Pi + \alpha > W_{j,i}^\top \Pi_0 + \alpha_0 
    & \Longrightarrow P_{i,j}((\beta^\top, \alpha, \rho_0)^\top, \bgamma) < P_{i,j}(\theta_0, \bgamma_0).
\end{align*}
Then, Assumption \ref{as:support} gives the desired result.

\bigskip
(ii) Since $(\theta_0, \bgamma_0)$ is a maximizer of $\E \mcl{L}_n(\theta, \bgamma)$ as confirmed above, $\rho_0$ must be a maximizer of $\mcl{L}_n^*(\rho)$, where the definition of $\mcl{L}_n^*(\rho)$ can be found in \eqref{eq:rho_concentrate},
and $(\tilde \beta_0(\rho_0), \tilde \alpha_0(\rho_0), \tilde \bgamma_0(\rho_0)) = (\beta_0, \alpha_0, \bgamma_0)$ holds.
For all $\rho \in \mcl{R}$, $\E \mcl{L}_n((\beta^\top, \alpha,\rho)^\top, \bgamma) - \E \mcl{L}_n((\tilde \beta_0(\rho)^\top, \tilde \alpha_0(\rho),\rho)^\top, \tilde \bgamma_0(\rho)) \le 0$ holds by definition.
Then, by the same argument as in the proof of (i), we can identify $(\tilde \beta_0(\rho), \tilde \alpha_0(\rho), \tilde \bgamma_0(\rho))$ uniquely for all $\rho \in \mcl{R}$ as Assumption \ref{as:support} is independent of the value of $\rho$.
Thus, if $\rho_0$ is identified as a unique maximizer of $\mcl{L}_n^*(\rho)$, all the parameters of the model are identified.\qed

\bigskip
\textbf{Proof of Theorem \ref{thm:consistency}}

(i) First, note that Assumptions \ref{as:error}--\ref{as:Z_bound} imply that there exist constants $\kappa_1, \kappa_2 \in (0, 1)$ such that $P_{i,j}(\theta, \bgamma) \in (\kappa_1, 1 - \kappa_1)$ and $1 - P_{i,j}(\theta, \bgamma) - P_{j,i}(\theta, \bgamma) \in (\kappa_2, 1 - \kappa_2)$
for all possible parameter values.
Observe that
\begin{align*}
\begin{split}
    & \mcl{L}_n(\theta, \bgamma) - \E \mcl{L}_n(\theta, \bgamma) \\
    & = \frac{1}{N}\sum_{i = 1}^n \sum_{j \neq i} \left[ 2 (y_{i,j} - \E y_{i,j}) \ln P_{i,j}(\theta, \bgamma) - 2(y_{i,j}  - \E y_{i,j}) \ln \left(1 - P_{i,j}(\theta, \bgamma) - P_{j,i}(\theta, \bgamma) \right) \right] \\
    & = \frac{2}{N}\sum_{i = 1}^n \sum_{j \neq i} (y_{i,j} - \E y_{i,j}) \psi_{i,j}(\theta, \bgamma),
\end{split}    
\end{align*}
where $\psi_{i,j}(\theta, \bgamma) \equiv \ln\left[ P_{i,j}(\theta, \bgamma) / \left(1 - P_{i,j}(\theta, \bgamma) - P_{j,i}(\theta, \bgamma)\right) \right]$.
Let $\bar \psi \equiv \ln ( (1 - \kappa_1)/\kappa_2)$, so that
\begin{align*}
 - (1 - \kappa_1) \bar \psi  < (y_{i,j} - \E y_{i,j}) \psi_{i,j}(\theta, \bgamma) < ( 1 - \kappa_1 )\bar \psi,
\end{align*}
where the inequalities are uniform in $(\theta, \bgamma) \in \Theta \times \mbb{C}_n$.
By the triangle inequality,
\begin{align*}
    \left| \mcl{L}_n(\theta, \bgamma) - \E \mcl{L}_n(\theta, \bgamma) \right| \le \frac{2}{n}\sum_{i=1}^n\left | \frac{1}{n-1} \sum_{j \neq i}  (y_{i,j} - \E y_{i,j}) \psi_{i,j}(\theta, \bgamma) \right|.
\end{align*}
Further, by Hoeffding's inequality,
\begin{align*}
    \Pr\left( \left | \frac{1}{n-1} \sum_{j \neq i}  (y_{i,j} - \E y_{i,j}) \psi_{i,j}(\theta, \bgamma) \right| > t \right)
    & \le 2 \exp\left( -\frac{2 (n - 1)^2 t^2}{\sum_{j \neq i} (2 (1 - \kappa_1 )\bar \psi)^2}\right) \\
    & = 2 \exp\left( -\frac{ (n - 1) t^2}{ 2 ( 1 - \kappa_1 )^2 \bar \psi^2}\right).
\end{align*}
Hence, Boole's inequality gives
\begin{align*}
    \Pr\left( \max_{1 \le i \le n} \left | \frac{1}{n-1} \sum_{j \neq i}  (y_{i,j} - \E y_{i,j}) \psi_{i,j}(\theta, \bgamma) \right| > t \right)
    & \le 2 n \exp\left( -\frac{ (n - 1) t^2}{ 2 ( 1 - \kappa_1 )^2 \bar \psi^2}\right).
\end{align*}
Setting $t = C \sqrt{\ln n/n}$ for a sufficiently large constant $C > 0$, we have
\begin{align*}
    2 n \exp\left( -\frac{ (n - 1) t^2}{ 2 (1 - \kappa_1 )^2 \bar \psi^2}\right)
    & =  2 n \exp\left( -\frac{n - 1}{ 2 (1 - \kappa_1 )^2 \bar \psi^2}\frac{C^2 \ln n}{ n}\right) \\
    & =  2 \exp\left( \ln n - \left(\frac{C^2 (n - 1)/n }{  2 ( 1 - \kappa_1 )^2 \bar \psi^2 }\right) \ln n \right) \to 0 \;\; \text{as} \;\; n \to \infty.
\end{align*}
This implies that
\begin{align}\label{eq:unif_conv}
   \sup_{(\theta, \bgamma) \: \in \: \Theta \times \mbb{C}_n} \left| \mcl{L}_n(\theta, \bgamma) - \E \mcl{L}_n(\theta, \bgamma) \right| = O_P\left( \sqrt{ \frac{\ln n}{n} }\right).
\end{align}
Then, with Assumption \ref{as:identification}, the rest of the proof follows from the same argument as in the proof of Theorem 2 in \cite{graham2017econometric}.
\qed

\bigskip

(ii) (iii) We prove the result by contradiction. 
Suppose that there exists a positive constant $c$ such that 
\begin{align*}
    \max\left\{ \frac{1}{n}\sum_{i=1}^n \left|\hat A_{n,i} - A_{0,i} \right|, \; \frac{1}{n}\sum_{i=1}^n \left|\hat B_{n,i} - B_{0,i} \right| \right\} \ge c > 0
\end{align*}
w.p.a.1.
Under Assumption \ref{as:para_space}(ii), this implies that there is a non-vanishing potion of observations with either or both $\hat A_{n,i}$ and $\hat B_{n,i}$ being not in the neighborhood of $A_{0,i}$ and $B_{0,i}$, respectively.
Therefore, by Assumption \ref{as:identification}, there exist a constant $\eta(c) > 0$ and $n(c) < \infty$ such that 
\begin{align}\label{eq:ineq_Aave1}
   4\eta(c) < \E \mcl{L}_n(\theta_0, \mbf{A}_0, \mbf{B}_0) - \E \mcl{L}_n(\theta_0, \hat{\mbf{A}}_n, \hat{\mbf{B}}_n)
\end{align}
for all $n \ge n(c)$.
Note that \eqref{eq:unif_conv} implies that 
\begin{align} \label{eq:ineq_Aave2}
    \E \mcl{L}_n(\theta_0, \mbf{A}_0, \mbf{B}_0) < \mcl{L}_n(\theta_0, \mbf{A}_0, \mbf{B}_0) + \eta(c)
\end{align}
w.p.a.1.
By the definition of the ML estimator,
\begin{align} \label{eq:ineq_Aave3}
    \mcl{L}_n(\theta_0, \mbf{A}_0, \mbf{B}_0) < \mcl{L}_n(\hat \theta_n, \hat{\mbf{A}}_n, \hat{\mbf{B}}_n) + \eta(c).
\end{align}
In addition, by the continuous mapping theorem and result (i), we have
\begin{align} \label{eq:ineq_Aave4}
    \mcl{L}_n(\hat \theta_n, \hat{\mbf{A}}_n,  \hat{\mbf{B}}_n) < \mcl{L}_n(\theta_0, \hat{\mbf{A}}_n,  \hat{\mbf{B}}_n) + \eta(c)
\end{align}
w.p.a.1.
Now, combining the inequalities \eqref{eq:ineq_Aave1}--\eqref{eq:ineq_Aave4} gives
\begin{align*}
    \E \mcl{L}_n(\theta_0, \hat{\mbf{A}}_n, \hat{\mbf{B}}_n) 
    & < \E \mcl{L}_n(\theta_0, \mbf{A}_0, \mbf{B}_0) - 4\eta(c) \\
    & < \mcl{L}_n(\theta_0, \mbf{A}_0, \mbf{B}_0) - 3\eta(c) \\
    & < \mcl{L}_n(\hat \theta_n, \hat{\mbf{A}}_n, \hat{\mbf{B}}_n) - 2\eta(c) \\
    & < \mcl{L}_n(\theta_0, \hat{\mbf{A}}_n, \hat{\mbf{B}}_n) - \eta(c)
\end{align*}
w.p.a.1.
The last line implies that $\eta(c) <  \mcl{L}_n(\theta_0, \hat{\mbf{A}}_n, \hat{\mbf{B}}_n) - \E \mcl{L}_n(\theta_0, \hat{\mbf{A}}_n, \hat{\mbf{B}}_n)$ w.p.a.1; however, this contradicts with \eqref{eq:unif_conv}.
Hence, as the choice of $c$ is arbitrary, we obtain the desired result.
\qed

\bigskip

(iv) Note that, for each $i$ ($i \neq 1$), it holds that
\begin{align*}
    \hat A_{n,i} = \argmax_{A_i \in \mbb{A}} \mcl{L}_n(\hat \theta_n, A_i, \hat{\mbf{A}}_{n,-i}, \hat{\mbf{B}}_n),
    \qquad A_{0,i} = \argmax_{A_i \in \mbb{A}} \E \mcl{L}_n(\theta_0, A_i, \mbf{A}_{0,-i}, \mbf{B}_0),
\end{align*}
where $\mbf{A}_{-i} \equiv (A_1, \ldots, A_{i-1}, A_{i+1}, \ldots, A_n)^\top$, and $\mcl{L}_n( \theta, A_i, \mbf{A}_{-i}, \mbf{B}) = \mcl{L}_n(\theta, \mbf{A}, \mbf{B})$.
Pick any $c > 0$, and let $\mbb{A}^c_i \equiv \{A \in \mbb{A} : |A - A_{0,i}| \ge c\}$.
Define $\varepsilon_n(c)$ as follows: 
\[
    \varepsilon_n(c) \equiv \min_{2 \le i \le n}\left[ \E \mcl{L}_n(\theta_0, A_{0,i}, \mbf{A}_{0,-i},\mbf{B}_0) -  \max_{A_i \in \mbb{A}^c_i} \E \mcl{L}_n(\theta_0, A_i, \mbf{A}_{0,-i},\mbf{B}_0)\right].
\]
By Assumption \ref{as:identification}, there exists $n(c) < \infty$ such that $\varepsilon_n(c)$ is strictly larger than zero for all $n \ge n(c)$.
By the definition of $\hat A_{n,i}$, we have
\begin{align}\label{eq:ineq_Ainf1}
    \mcl{L}_n(\hat \theta_n, \hat A_{n,i}, \hat{\mbf{A}}_{n,-i}, \hat{\mbf{B}}_n) > \mcl{L}_n(\hat \theta_n, A_{0,i}, \hat{\mbf{A}}_{n,-i}, \hat{\mbf{B}}_n) -  \varepsilon_n(c) /5.
\end{align}
By the triangle inequality, 
\begin{align*}
    & \left| \mcl{L}_n(\hat \theta_n, A_i, \hat{\mbf{A}}_{n,-i}, \hat{\mbf{B}}_n) - \mcl{L}_n(\theta_0, A_i, \mbf{A}_{0,-i}, \mbf{B}_0)  \right| \\
    & \le \left| \mcl{L}_n(\hat \theta_n, A_i, \hat{\mbf{A}}_{n,-i}, \hat{\mbf{B}}_n) - \mcl{L}_n(\theta_0, A_i, \hat{\mbf{A}}_{n,-i}, \hat{\mbf{B}}_n)  \right|  + \left| \mcl{L}_n(\theta_0, A_i, \hat{\mbf{A}}_{n,-i}, \hat{\mbf{B}}_n) - \mcl{L}_n(\theta_0, A_i, \mbf{A}_{0,-i}, \hat{\mbf{B}}_n) \right|\\
    & \quad + \left| \mcl{L}_n(\theta_0, A_i, \mbf{A}_{0,-i}, \hat{\mbf{B}}_n) - \mcl{L}_n(\theta_0, A_i, \mbf{A}_{0,-i}, \mbf{B}_0) \right|\\
    & \le \left| \partial_{\mbf{A}_{-i}^\top } \mcl{L}_n(\theta_0, A_i, \bar{\mbf{A}}_{n,-i}, \hat{\mbf{B}}_n)[\hat{\mbf{A}}_{n,-i} - \mbf{A}_{0,-i}] \right| + \left| \partial_{\mbf{B}^\top } \mcl{L}_n(\theta_0, A_i, \mbf{A}_{0,-i}, \bar{\mbf{B}}_n)[\hat{\mbf{B}}_n - \mbf{B}_0] \right| + o_P(1),
\end{align*}
where the second inequality follows from the mean value expansion with result (i), $\bar{\mbf{A}}_{n,-i} \in [\hat{\mbf{A}}_{n,-i}, \mbf{A}_{0,-i}]$, and $\bar{\mbf{B}}_n \in [\hat{\mbf{B}}_n, \mbf{B}_0]$. 
Here, the first term on the right-hand side has the following form: 
\begin{align*}
\begin{split}
    & \partial_{\mbf{A}_{-k}^\top} \mcl{L}_n(\theta_0, A_k, \bar{\mbf{A}}_{n,-k}, \hat{\mbf{B}}_n) [\hat{\mbf{A}}_{n,-k} - \mbf{A}_{0,-k}] \\
    & = \frac{2}{N}\sum_{i = 1}^n \sum_{j > i} \partial_{\mbf{A}_{-k}^\top} \ell_{i,j}(\theta_0, A_k, \bar{\mbf{A}}_{n,-k}, \hat{\mbf{B}}_n) [\hat{\mbf{A}}_{n,-k} - \mbf{A}_{0,-k}] \\
    & = \frac{2}{N}\sum_{i = 1}^n \sum_{j > i} s_{1,i,j}(\theta_0, A_k, \bar{\mbf{A}}_{n,-k}, \hat{\mbf{B}}_n) \chi_{n,i,-k}^\top [\hat{\mbf{A}}_{n,-k} - \mbf{A}_{0,-k}] + \frac{2}{N}\sum_{i = 1}^n \sum_{j > i} s_{2,i,j}(\theta_0, A_k, \bar{\mbf{A}}_{n,-k}, \hat{\mbf{B}}_n) \chi_{n,j,-k}^\top [\hat{\mbf{A}}_{n,-k} - \mbf{A}_{0,-k}] \\
    & = \frac{2}{N}\sum_{i = 1}^n \sum_{j > i} s_{1,i,j}(\theta_0, A_k, \bar{\mbf{A}}_{n,-k}, \hat{\mbf{B}}_n) [\hat{A}_{n,i} - A_{0,i}]\mbf{1}\{i \neq k\} + \frac{2}{N}\sum_{i = 1}^n \sum_{j > i} s_{2,i,j}(\theta_0, A_k, \bar{\mbf{A}}_{n,-k}, \hat{\mbf{B}}_n) [\hat{A}_{n,j} - A_{0,j}]\mbf{1}\{j \neq k\} \\
    & = \frac{2}{N}\sum_{i = 1}^n \sum_{j \neq i} s_{1,i,j}(\theta_0, A_k, \bar{\mbf{A}}_{n,-k}, \hat{\mbf{B}}_n) [\hat{A}_{n,i} - A_{0,i}]\mbf{1}\{i \neq k\},
\end{split}
\end{align*}
where $\chi_{n,i,-k}$ and $\chi_{n,j,-k}$ are $(n - 1) \times 1$ vectors defined by removing the $k$-th element of $\chi_{n,i}$ and $\chi_{n,j}$, respectively, and the last equality holds because $s_{2,i,j}(\theta, \mbf{A}, \mbf{B}) = \partial_{\pi_{j,i}} \ell_{i,j}(\theta, \mbf{A}, \mbf{B}) = \partial_{\pi_{j,i}} \ell_{j,i}(\theta, \mbf{A}, \mbf{B}) = s_{1,j,i}(\theta, \mbf{A}, \mbf{B})$.
Then, for a constant $c > 0$ independent of $A_k$ and $k$, we have
\begin{align*}
    \left| \partial_{\mbf{A}_{-k}^\top} \mcl{L}_n(\theta_0, A_k, \bar{\mbf{A}}_{n,-k}, \hat{\mbf{B}}_n) [\hat{\mbf{A}}_{n,-k} - \mbf{A}_{0,-k}] \right|
    & \le \frac{c}{n}\sum_{i = 1}^n \left| \hat A_{n,i} - A_{0,i} \right| = o_P(1)
\end{align*}
by result (ii).
Based on the same argument, we can also show that $\left| \partial_{\mbf{B}^\top } \mcl{L}_n(\theta_0, A_i, \mbf{A}_{0,-i}, \bar{\mbf{B}}_n)[\hat{\mbf{B}}_n - \mbf{B}_0] \right| = o_P(1)$, implying that 
\begin{align*}
    \left| \mcl{L}_n(\hat \theta_n, A_i, \hat{\mbf{A}}_{n,-i}, \hat{\mbf{B}}_n) - \mcl{L}_n(\theta_0, A_i, \mbf{A}_{0,-i}, \mbf{B}_0)  \right| = o_P(1)
\end{align*}
uniformly in $A_i \in \mbb{A}$ and $i$.
Similarly, we can show that $\left| \E \mcl{L}_n(\hat \theta_n, A_i, \hat{\mbf{A}}_{n,-i}, \hat{\mbf{B}}_n) - \E \mcl{L}_n(\theta_0, A_i, \mbf{A}_{0,-i}, \mbf{B}_0) \right| = o_P(1)$.
Hence, the following inequalities hold w.p.a.1:
\begin{align}
    \mcl{L}_n(\theta_0, A_i, \mbf{A}_{0,-i}, \mbf{B}_0)
    & > \mcl{L}_n(\hat \theta_n, A_i, \hat{\mbf{A}}_{n,-i},\hat{\mbf{B}}_n) -  \varepsilon_n(c) /5 \label{eq:ineq_Ainf2}\\
    \E \mcl{L}_n(\hat \theta_n, A_i, \hat{\mbf{A}}_{n,-i},\hat{\mbf{B}}_n)
    & > \E \mcl{L}_n(\theta_0, A_i, \mbf{A}_{0,-i},\mbf{B}_0) -  \varepsilon_n(c) /5 \label{eq:ineq_Ainf3}
\end{align}
uniformly in $A_i \in \mbb{A}$ and $i$.
In addition, \eqref{eq:unif_conv} implies that 
\begin{align}
    \E \mcl{L}_n(\theta_0, \hat A_{n,i}, \mbf{A}_{0,-i}, \mbf{B}_0) 
    & > \mcl{L}_n(\theta_0, \hat A_{n,i}, \mbf{A}_{0,-i}, \mbf{B}_0)  -  \varepsilon_n(c) /5 \label{eq:ineq_Ainf4}\\
    \mcl{L}_n(\hat \theta_n, A_{0,i}, \hat{\mbf{A}}_{0,-i},\hat{\mbf{B}}_n)
    & > \E \mcl{L}_n(\hat \theta_n, A_{0,i}, \hat{\mbf{A}}_{0,-i},\hat{\mbf{B}}_n) -  \varepsilon_n(c) /5 \label{eq:ineq_Ainf5}
\end{align}
w.p.a.1.
Then, combining the inequalities \eqref{eq:ineq_Ainf1} and \eqref{eq:ineq_Ainf2}--\eqref{eq:ineq_Ainf5} yields
\begin{align*}
    \E \mcl{L}_n(\theta_0, \hat A_{n,i}, \mbf{A}_{0,-i}, \mbf{B}_0) 
    & > \mcl{L}_n(\theta_0, \hat A_{n,i}, \mbf{A}_{0,-i}, \mbf{B}_0)  -  \varepsilon_n(c) /5 \\
    & > \mcl{L}_n(\hat \theta_n, \hat A_{n,i}, \hat{\mbf{A}}_{n,-i}, \hat{\mbf{B}}_n) -  2\varepsilon_n(c) /5\\
    & > \mcl{L}_n(\hat \theta_n, A_{0,i}, \hat{\mbf{A}}_{n,-i}, \hat{\mbf{B}}_n) -  3\varepsilon_n(c) /5 \\
    & > \E \mcl{L}_n(\hat \theta_n, A_{0,i}, \hat{\mbf{A}}_{n,-i}, \hat{\mbf{B}}_n)  -  4\varepsilon_n(c) /5\\
    & > \E \mcl{L}_n(\theta_0, A_{0,i}, \mbf{A}_{0,-i}, \mbf{B}_0) -  \varepsilon_n(c) \\
    & = \max_{A_i \in \mbb{A}^c_i} \E \mcl{L}_n(\theta_0, A_i, \mbf{A}_{0,-i}, \mbf{B}_0) \\
    & \quad + \underbrace{\left[ \E \mcl{L}_n(\theta_0, A_{0,i}, \mbf{A}_{0,-i}, \mbf{B}_0) - \max_{A_i \in \mbb{A}^c_i} \E \mcl{L}_n(\theta_0, A_i, \mbf{A}_{0,-i}, \mbf{B}_0) \right] -  \varepsilon_n(c)}_{\ge \: 0} \\
   & \ge  \max_{A_i \in \mbb{A}^c_i} \E \mcl{L}_n(\theta_0, A_i, \mbf{A}_{0,-i}, \mbf{B}_0)
\end{align*}
w.p.a.1 for all $i$.
The last line implies that $\hat A_{n,i} \notin \mbb{A}_i^c$.
As the choice of $c$ is arbitrary, this further implies that $\max_{1 \le i \le n}|\hat A_{n,i} - A_{0,i}| \overset{p}{\to}0$.
Analogously, we can also show that $\max_{1 \le i \le n}|\hat B_{n,i} - B_{0,i}| \overset{p}{\to}0$.
\qed

\bigskip

\begin{lemma}\label{lem:h_LLN}
    For any $(\theta, \bgamma) \in \Theta \times \mbb{C}_n$ such that $||\theta - \theta_0|| = o(1)$ and $||\bgamma - \bgamma_0||_\infty = o(1)$, 
    \begin{enumerate}
        \item[(i)] $\max_{1 \le k \le n}\left| \frac{1}{n - 1} \sum_{j \neq k} \left( h_{11,k,j}(\theta, \bgamma) - \E h_{11,k,j}(\theta_0, \bgamma_0) \right) \right| = o_P(1)$,
        \item[(ii)] $\max_{1 \le k \le n}\left| \frac{1}{n - 1} \sum_{j \neq k} \left( h_{11,j,k}(\theta, \bgamma) - \E h_{11,j,k}(\theta_0, \bgamma_0) \right) \right| = o_P(1)$.
    \end{enumerate} 
\end{lemma}

\begin{proof}
We only prove (i) since (ii) is completely analogous.
By the triangle inequality,
\begin{align}\label{eq:triangle1}
    \begin{split}
    \left| \frac{1}{n - 1} \sum_{j \neq k} \left( h_{11,k,j}(\theta, \bgamma) - \E h_{11,k,j}(\theta_0, \bgamma_0) \right) \right| 
    & \le \left| \frac{1}{n - 1} \sum_{j \neq k} \left( h_{11,k,j}(\theta, \bgamma) - h_{11,k,j}(\theta_0, \bgamma_0) \right) \right| \\
    & \quad + \left| \frac{1}{n - 1} \sum_{j \neq k} \left( h_{11,k,j}(\theta_0, \bgamma_0) - \E h_{11,k,j}(\theta_0, \bgamma_0) \right) \right|.
    \end{split}
\end{align}
With Assumption \ref{as:error}(iii), the mean value expansion gives
\begin{align*}
    h_{11,k,j}(\theta, \bgamma) - h_{11,k,j}(\theta_0, \bgamma_0) 
    & = h_{11,k,j}(\theta, \bgamma) - h_{11,k,j}(\theta_0, \bgamma) + h_{11,k,j}(\theta_0, \bgamma) -  h_{11,k,j}(\theta_0, \bgamma_0) \\
    & = \partial_{\theta^\top} h_{11,k,j}(\bar \theta, \bgamma) [\theta - \theta_0] + \partial_{\mbf{A}_{-1}^\top} h_{11,k,j}(\theta, \bar \bgamma) [\mbf{A}_{-1} - \mbf{A}_{0,-1}] + \partial_{\mbf{B}^\top} h_{11,k,j}(\theta, \bar \bgamma) [\mbf{B} - \mbf{B}_0],
\end{align*}
where $\bar \theta \in [\theta, \theta_0]$, and $\bar \bgamma \in [\bgamma, \bgamma_0]$.
Further, letting $h_{111,k,j}(\theta, \bgamma) \equiv \partial_{\pi_{k,j}} h_{11,k,j}(\theta, \bgamma)$ and $h_{112,k,j}(\theta, \mbf{A}) \equiv \partial_{\pi_{j,k}} h_{11,k,j}(\theta, \bgamma)$, we have
\begin{align*}
    \partial_{\mbf{A}_{-1}^\top} h_{11,k,j}(\theta, \bar \bgamma) [\mbf{A}_{-1} - \mbf{A}_{0,-1}]
    & = h_{111,k,j}(\theta, \bar \bgamma)\chi_{n,k,-1}^\top [\mbf{A}_{-1} - \mbf{A}_{0,-1}] + h_{112,k,j}(\theta, \bar\gamma)\chi_{n,j,-1}^\top [\mbf{A}_{-1} - \mbf{A}_{0,-1}]\\
    \partial_{\mbf{B}^\top} h_{11,k,j}(\theta, \bar \bgamma) [\mbf{B} - \mbf{B}_0]
    & = h_{111,k,j}(\theta, \bar \bgamma)\chi_{n,j}^\top [\mbf{B} - \mbf{B}_0] + h_{112,k,j}(\theta, \bar\gamma)\chi_{n,k}^\top[\mbf{B} - \mbf{B}_0].
\end{align*}
Then, for some large constant $c > 0$,
\begin{align*}
    & \left| \frac{1}{n - 1} \sum_{j \neq k} \left( h_{11,k,j}(\theta, \bgamma) - h_{11,k,j}(\theta_0, \bgamma_0) \right) \right| \\
    & \le \left| \frac{1}{n - 1} \sum_{j \neq k} \partial_{\theta^\top} h_{11,k,j}(\bar \theta, \bgamma) [\theta - \theta_0] \right| \\
    & \quad + \left| \frac{1}{n - 1} \sum_{j \neq k}  h_{111,k,j}(\theta, \bar \bgamma)\chi_{n,k,-1}^\top [\mbf{A}_{-1} - \mbf{A}_{0,-1}] \right| + \left| \frac{1}{n - 1} \sum_{j \neq k}  h_{112,k,j}(\theta, \bar\gamma)\chi_{n,j,-1}^\top [\mbf{A}_{-1} - \mbf{A}_{0,-1}] \right| \\
    & \quad + \left| \frac{1}{n - 1} \sum_{j \neq k}  h_{111,k,j}(\theta, \bar \bgamma)\chi_{n,j}^\top [\mbf{B} - \mbf{B}_0] \right| + \left| \frac{1}{n - 1} \sum_{j \neq k}  h_{112,k,j}(\theta, \bar\gamma)\chi_{n,k}^\top[\mbf{B} - \mbf{B}_0] \right| \\
    & \le c \left\| \theta - \theta_0 \right\| + 2c \left\| \mbf{A} - \mbf{A}_0 \right\|_\infty + 2c \left\| \mbf{B} - \mbf{B}_0 \right\|_\infty.
\end{align*}
As the right-hand side term in the last line is independent of $k$, we have $\left| \frac{1}{n - 1} \sum_{j \neq k} \left( h_{11,k,j}(\theta, \bgamma) - h_{11,k,j}(\theta_0, \bgamma_0) \right) \right| = o(1)$ for all $k$ for any $(\theta, \bgamma)$ such that $||\theta - \theta_0|| = o(1)$ and $||\bgamma - \bgamma_0||_\infty = o(1)$.

For the second term of \eqref{eq:triangle1}, note that the random components involved in $h_{11,k,j}(\theta_0, \bgamma_0)$ are only $(y_{k,j}, y_{j,k})$ and, thus, that $\{h_{11,k,j}(\theta_0, \bgamma_0)\}_{j \neq k}$ are independent by Assumption \ref{as:error}(ii).
Further, as $h_{11,k,j}(\theta_0, \bgamma_0)$ is uniformly bounded, using Hoeffding's and Boole's inequalities similarly as above, we can show that
\begin{align*}
    \max_{1 \le k \le n} \left| \frac{1}{n - 1} \sum_{j \neq k} \left( h_{11,k,j}(\theta_0, \bgamma_0) - \E h_{11,k,j}(\theta_0, \bgamma_0) \right) \right| = O_P\left(\sqrt{\frac{\ln n}{n}} \right).
\end{align*}
This completes the proof.
\end{proof}

\begin{lemma}\label{lem:A_rate1}
    \begin{enumerate}
    \item[(i)] $\left\| \tilde \bgamma_n(\theta) - \tilde \bgamma_0(\theta_0) \right\|_\infty = o_P(1)$ for any $\theta \in \Theta$ such that $||\theta - \theta_0|| = o(1)$,
    \item[(ii)] $\frac{1}{n}\sum_{i = 1}^n \left| \tilde{A}_{n,i}(\theta_0) - \tilde{A}_{0,i}(\theta_0) \right| = O_P(n^{-1/2})$,
    \item[(iii)] $\frac{1}{n}\sum_{i = 1}^n \left| \tilde{B}_{n,i}(\theta_0) - \tilde{B}_{0,i}(\theta_0) \right| = O_P(n^{-1/2})$,
    \item[(iv)]  $\left\| \tilde \bgamma_n(\theta_0) - \tilde \bgamma_0(\theta_0) \right\|_\infty = O_P(\sqrt{\ln n/n})$.
    \end{enumerate} 
\end{lemma}

\begin{proof}
    (i) By the triangle inequality,
    \begin{align*}
        \left\| \tilde \bgamma_n(\theta) - \tilde \bgamma_0(\theta_0) \right\|_\infty 
        & \le \left\| \tilde \bgamma_n(\theta) - \tilde \bgamma_0(\theta) \right\|_\infty  + \left\| \tilde \bgamma_0(\theta) - \tilde \bgamma_0(\theta_0) \right\|_\infty.
    \end{align*}
    For the first term on the right-hand side, the same argument as in the proof of Theorem \ref{thm:consistency}(iv) achieves $|| \tilde \bgamma_n(\theta) - \tilde \bgamma_0(\theta) ||_\infty = o_P(1)$ for any $\theta$ in the neighborhood of $\theta_0$ under Assumption \ref{as:unique}.
    For the second term, Assumption \ref{as:unique} and Berge's theorem implies that every element of  $\tilde \bgamma_0(\theta)$ is continuous in the neighborhood of $\theta_0$ (see, e.g., Corollary A4.8, \citealp{kreps2012microeconomic}). 
    Thus, $|| \tilde \bgamma_0(\theta) - \tilde \bgamma_0(\theta_0) ||_\infty = o(1)$ holds.
    \bigskip

    (ii) (iii) By the first-order condition and the mean value expansion,
    \begin{align*}
        \mbf{0}_{(2n-1) \times 1} 
        & = n \cdot \mcl{S}_{n,\bgamma}(\theta_0, \tilde \bgamma_n(\theta_0)) \\
        & = n \cdot \mcl{S}_{n,\bgamma} - \left(- n \cdot \mcl{H}_{n,\bgamma\bgamma}(\theta_0, \bar \bgamma_n)\right) [\tilde \bgamma_{n,-1}(\theta_0) - \tilde \bgamma_{0,-1}(\theta_0)],
    \end{align*}
    where $\bar \bgamma_n \in [\tilde \bgamma_n(\theta_0), \tilde \bgamma_0(\theta_0)]$.
    Then, by result (i) and Assumption \ref{as:hessian}(i), we have
    \begin{align}\label{eq:expansion_A}
        \tilde \bgamma_{n,-1}(\theta_0) - \tilde \bgamma_{0,-1}(\theta_0) = \left(- n \cdot \mcl{H}_{n,\bgamma\bgamma}(\theta_0, \bar \bgamma_n)\right)^{-1} n \cdot \mcl{S}_{n,\bgamma}
    \end{align}
    w.p.a.1; thus, $||\tilde \bgamma_n(\theta_0) - \tilde \bgamma_0(\theta_0)|| = || \tilde \bgamma_{n,-1}(\theta_0) - \tilde \bgamma_{0,-1}(\theta_0) || \le O_P(1) \cdot ||n \cdot \mcl{S}_{n,\bgamma}||$.
    
    Further, observe that
    \begin{align*}
        \E ||n \cdot \mcl{S}_{n,\bgamma} ||^2 
        & = \frac{4}{(n-1)^2} \sum_{i = 1}^n \sum_{j > i}\sum_{k = 1}^n \sum_{l > k} \E\left[ s_{i,j}^\mbf{A}(\theta_0, \bgamma_0)^\top  s_{k,l}^\mbf{A}(\theta_0, \bgamma_0) + s_{i,j}^\mbf{B}(\theta_0, \bgamma_0)^\top  s_{k,l}^\mbf{B}(\theta_0, \bgamma_0) \right] \\
        & = \frac{4}{(n-1)^2} \sum_{i = 1}^n \sum_{j > i}\sum_{k = 1}^n \sum_{l > k} \E \left[ (s_{1,i,j} \chi_{n,i,-1}^\top + s_{2,i,j} \chi_{n,j,-1}^\top) (s_{1,k,l} \chi_{n,k,-1} + s_{2,k,l} \chi_{n,l,-1}) \right] \\
        & \quad + \frac{4}{(n-1)^2} \sum_{i = 1}^n \sum_{j > i}\sum_{k = 1}^n \sum_{l > k} \E \left[ (s_{1,i,j} \chi_{n,j}^\top + s_{2,i,j} \chi_{n,i}^\top) (s_{1,k,l} \chi_{n,l} + s_{2,k,l} \chi_{n,k}) \right] \\
        & = \frac{4}{(n-1)^2} \sum_{i = 1}^n \sum_{j > i} \E \left[ (s_{1,i,j} \chi_{n,i,-1}^\top + s_{2,i,j} \chi_{n,j,-1}^\top) (s_{1,i,j} \chi_{n,i,-1} + s_{2,i,j} \chi_{n,j,-1}) \right] \\
        & \quad + \frac{4}{(n-1)^2} \sum_{i = 1}^n \sum_{j > i} \E \left[ (s_{1,i,j} \chi_{n,i,-1}^\top + s_{2,i,j} \chi_{n,j,-1}^\top) (s_{1,j,i} \chi_{n,j,-1} + s_{2,j,i} \chi_{n,i,-1}) \right] \\
        & \quad + \frac{4}{(n-1)^2} \sum_{i = 1}^n \sum_{j > i} \E \left[ (s_{1,i,j} \chi_{n,j}^\top + s_{2,i,j} \chi_{n,i}^\top) (s_{1,i,j} \chi_{n,j} + s_{2,i,j} \chi_{n,i}) \right] \\
        & \quad + \frac{4}{(n-1)^2} \sum_{i = 1}^n \sum_{j > i} \E \left[ (s_{1,i,j} \chi_{n,j}^\top + s_{2,i,j} \chi_{n,i}^\top) (s_{1,j,i} \chi_{n,i} + s_{2,j,i} \chi_{n,j}) \right] = O(1)
    \end{align*}
    by Assumption \ref{as:error}(ii).
    Then, it holds that $||n \cdot \mcl{S}_{n,\bgamma}|| = O_P(1)$ by Markov's inequality; thus, we have $||\tilde \bgamma_n(\theta_0) - \tilde \bgamma_0(\theta_0)|| = O_P(1)$.
    Finally, by the basic norm inequality, it holds that 
    \begin{align*}
        \sum_{i = 1}^n|\tilde{A}_{n,i}(\theta_0) - \tilde{A}_{0,i}(\theta_0)| + \sum_{i = 1}^n|\tilde{B}_{n,i}(\theta_0) - \tilde{B}_{0,i}(\theta_0)| \le \sqrt{2n-1} ||\tilde \bgamma_n(\theta_0) - \tilde \bgamma_0(\theta_0)|| = O_P(\sqrt{n}),
    \end{align*}
    which gives the desired result.
    \bigskip

    (iv) Let $\bgamma_{-k} = (\mbf{A}_{-k}^\top, \mbf{B}^\top)^\top$ for $k \neq 1$ and write  $\mcl{L}_n(\theta, A_k, \bgamma_{-k})$ as $\mcl{L}_n(\theta, \bgamma)$.
    By the first-order condition and mean value expansion,
    \begin{align*}
        0 
        & = n \cdot \partial_{A_k} \mcl{L}_n(\theta_0, \tilde A_{n,k}(\theta_0), \tilde \bgamma_{n,-k}(\theta_0)) \\
        & = \frac{2}{n-1}\sum_{i = 1}^n \sum_{j > i} s_{i,j}^{A_k} 
        + n \cdot \partial_{A_k} \mcl{L}_n(\theta_0, \tilde A_{n,k}(\theta_0), \tilde \bgamma_{n,-k}(\theta_0))  - n \cdot \partial_{A_k} \mcl{L}_n( \theta_0, \tilde A_{0,k}(\theta_0), \tilde \bgamma_{n,-k}(\theta_0)) \\
        & \quad + n \cdot \partial_{A_k} \mcl{L}_n(\theta_0, A_{0,k}, \tilde \bgamma_{n,-k}(\theta_0)) - n \cdot \partial_{A_k} \mcl{L}_n( \theta_0, A_{0,k}, \tilde \bgamma_{0,-k}(\theta_0))  \\   
        & = \frac{2}{n-1}\sum_{i = 1}^n \sum_{j > i} s_{i,j}^{A_k}
        + \frac{2}{n-1}\sum_{j \neq k} h_{11,k,j}(\theta_0, \bar A_{n,k}, \tilde \bgamma_{n,-k}(\theta_0)) [\tilde A_{n,k}(\theta_0) - \tilde A_{0,k}(\theta_0)] \\
        & \quad + \frac{2}{n-1}\sum_{i = 1}^n \sum_{j > i} \partial_{\bgamma_{-k}^\top} s_{i,j}^{A_k}(\theta_0, A_{0,k}, \bar \bgamma_{n,-k}) [\tilde \bgamma_{n,-k}(\theta_0) - \tilde \bgamma_{0,-k}(\theta_0)],
    \end{align*}
    where $\bar A_{n,k} \in [\tilde A_{n,k}(\theta_0), \tilde A_{0,k}(\theta_0)]$, and $\bar \bgamma_{n,-k}\in [\tilde \bgamma_{n,-k}(\theta_0), \tilde \bgamma_{0,-k}(\theta_0)]$.
    In view of \eqref{eq:Eh}, Lemma \ref{lem:h_LLN}(i), and result (i) imply that $\frac{2}{n-1}\sum_{j \neq k} h_{11,k,j}(\theta_0, \bar A_{n,k}, \tilde \bgamma_{n,-k}(\theta_0))$ is bounded and away from zero w.p.a.1 uniformly in $k$.
    Then,
    \begin{align*}
        \left| \tilde{A}_{n,k}(\theta_0) - \tilde{A}_{0,k}(\theta_0) \right| \le (c + o_p(1))\left\{ \left| T_{1,n,k} \right| + \left| T_{2,n,k} \right| \right\}
    \end{align*}
    for some $c > 0$, where
    \begin{align*}
        T_{1,n,k} 
        & \equiv \frac{2}{n-1}\sum_{i = 1}^n \sum_{j > i} s_{i,j}^{A_k}, \\
        T_{2,n,k}
        & \equiv \frac{2}{n-1}\sum_{i = 1}^n \sum_{j > i} \partial_{\bgamma_{-k}^\top} s_{i,j}^{A_k}(\theta_0, A_{0,k}, \bar \bgamma_{n,-k}) [\tilde \bgamma_{n,-k}(\theta_0) - \tilde \bgamma_{0,-k}(\theta_0)]\\
        & = \frac{2}{n-1}\sum_{i = 1}^n \sum_{j > i} \partial_{\mbf{A}_{-k}^\top} s_{i,j}^{A_k}(\theta_0, A_{0,k}, \bar \bgamma_{n,-k}) [\tilde{\mbf{A}}_{n,-k}(\theta_0) - \tilde{\mbf{A}}_{0,-k}(\theta_0)] \\
        & \quad + \frac{2}{n-1}\sum_{i = 1}^n \sum_{j > i} \partial_{\mbf{B}^\top} s_{i,j}^{A_k}(\theta_0, A_{0,k}, \bar \bgamma_{n,-k}) [\tilde{\mbf{B}}_n(\theta_0) - \tilde{\mbf{B}}_0(\theta_0)] \\
        &  \equiv T_{21,n,k} + T_{22,n,k}, \;\; \text{say.}
    \end{align*}

    First, observe that 
    \begin{align*}
        T_{1,n,k}
        & = \frac{2}{n-1}\sum_{i = 1}^n \sum_{j > i} s_{1,i,j}\mbf{1}\{ i = k \} + \frac{2}{n-1}\sum_{i = 1}^n \sum_{j > i}  s_{2,i,j}\mbf{1}\{ j = k \}  \\
        & = \frac{2}{n-1}\sum_{j > k} s_{1,k,j} + \frac{2}{n-1} \sum_{j < k} s_{2,j,k} \\
        & = \frac{2}{n-1} \sum_{j \neq k} s_{1,k,j},
    \end{align*}
    where the last equality holds because $s_{1,k,j} = s_{2,j,k}$.
    Clearly, $s_{1,k,j}$ is uniformly bounded and $\E s_{1,k,j} = 0$.
    Then, with Assumption \ref{as:error}(ii), 
    we can show that
    \begin{align}\label{eq:T_1}
        \max_{1 \le k \le n}|T_{1,n,k}| = O_P(\sqrt{\ln n / n})
    \end{align}
    by Hoeffding's and Boole's inequalities.

    Next, observe that there exist bounded functions $s_{11,i,j}(\theta, \bgamma) \equiv \partial_{\pi_{i,j}} s_{1,i,j}(\theta, \bgamma)$, $s_{12,i,j}(\theta, \bgamma) \equiv \partial_{\pi_{j,i}} s_{1,i,j}(\theta, \bgamma)$, $s_{21,i,j}(\theta, \bgamma) \equiv \partial_{\pi_{i,j}} s_{2,i,j}(\theta, \bgamma)$, and $s_{22,i,j}(\theta, \bgamma) \equiv \partial_{\pi_{j,i}} s_{2,i,j}(\theta, \bgamma)$, satisfying
    \begin{align*}
        \sum_{i = 1}^n \sum_{j > i} \partial_{\mbf{A}_{-k}^\top} s_{i,j}^{A_k}(\theta, \bgamma) 
        & = \sum_{i = 1}^n \sum_{j >i} s_{11,i,j}(\theta, \bgamma) \mbf{1}\{i = k\}\chi_{n,i,-k}^\top + \sum_{i = 1}^n \sum_{j > i} s_{12,i,j}(\theta, \bgamma) \mbf{1}\{i = k\}\chi_{n,j,-k}^\top \\
        & \quad + \sum_{i = 1}^n \sum_{j > i} s_{21,i,j}(\theta, \bgamma) \mbf{1}\{j = k\}\chi_{n,i,-k}^\top + \sum_{i = 1}^n \sum_{j > i} s_{22,i,j}(\theta, \bgamma) \mbf{1}\{j = k\}\chi_{n,j,-k}^\top \\
        & = \sum_{j > k} s_{12,k,j}(\theta, \bgamma) \chi_{n,j,-k}^\top  + \sum_{j < k} s_{21,j,k}(\theta, \bgamma) \chi_{n,j,-k}^\top \\
        & = \sum_{j \neq k} s_{12,k,j}(\theta, \bgamma) \chi_{n,j,-k}^\top,
    \end{align*}
    where the second equality follows since $ \mbf{1}\{i = k\}\chi_{n,i,-k}$ is only a vector of zeros.
    Hence, for some $c > 0$, the triangle inequality gives
    \begin{align*}
        \left|  T_{21,n,k} \right| 
        & = \left| \frac{2}{n-1}    \sum_{j \neq k} s_{12,k,j}(\theta_0, A_{0,k}, \bar \bgamma_{n,-k}) \chi_{n,j,-k}^\top [\tilde{\mbf{A}}_{n,-k}(\theta_0) - \tilde{\mbf{A}}_{0,-k}(\theta_0)] \right| \\
        & \le \frac{c}{n - 1}\sum_{j \neq k} \left| \chi_{n,j,-k}^\top [\tilde{\mbf{A}}_{n,-k}(\theta_0) - \tilde{\mbf{A}}_{0,-k}(\theta_0)]\right| \\
        & = \frac{c}{n - 1}\sum_{j \neq k} \left| \tilde A_{n,j}(\theta_0) - \tilde A_{0,j}(\theta_0) \right| \le \frac{c}{n-1}\sum_{j = 1}^n\left| \tilde A_{n,j}(\theta_0) - \tilde A_{0,j}(\theta_0)  \right| = O_P(n^{-1/2})
    \end{align*}
    by result (ii).
    Note that the last inequality is independent of $k$.
    Analogously, we can show that $| T_{22,n,k} | = O_P(n^{-1/2})$ uniformly in $k$ by result (iii).
    Thus, we have
    \begin{align}\label{eq:T_2}
        \max_{1 \le k \le n}|T_{2,n,k}| = O_P(n^{-1/2}).
    \end{align}
    Combining \eqref{eq:T_1} and \eqref{eq:T_2} yields $\max_{1 \le k \le n}|\tilde{A}_{n,k}(\theta_0) - \tilde{A}_{0,k}(\theta_0)| = O_P(\sqrt{\ln n / n}) +  O_P(n^{-1/2})$.

    Similarly as above, using Lemma \ref{lem:h_LLN}(ii), we can also show that $\max_{1 \le k \le n}|\tilde{B}_{n,k}(\theta_0) - \tilde{B}_{0,k}(\theta_0)| = O_P(\sqrt{\ln n / n}) +  O_P(n^{-1/2})$.
    This completes the proof.
\end{proof}


\begin{lemma}\label{lem:theta_rate}
    $|| \hat \theta_n - \theta_0 || = O_P(n^{-1/2})$. 
\end{lemma}

\begin{proof}
    Applying the implicit function theorem to $\mcl{S}_{n,\bgamma}(\theta, \tilde\bgamma_n(\theta)) = \partial_{\bgamma} \mcl{L}_n(\theta, \tilde\bgamma_n(\theta)) = \mbf{0}_{(2n-1) \times 1}$ for $\theta \in \Theta$ yields
    \begin{align*}
        \mbf{0}_{(2n - 1) \times (d_z + 2) } 
        & = \partial_{\theta^\top} \mcl{S}_{n,\bgamma}(\theta, \tilde\bgamma_n(\theta)) \\
        & = \mcl{H}_{n, \bgamma \theta } (\theta, \tilde \bgamma_n(\theta)) + \mcl{H}_{n,\bgamma\bgamma} (\theta, \tilde \bgamma_n(\theta))\partial_{\theta^\top} \tilde \bgamma_n(\theta) \\
         \Longrightarrow \partial_{\theta^\top} \tilde \bgamma_n(\theta)
         & = - \left[ n \cdot \mcl{H}_{n,\bgamma\bgamma} (\theta, \tilde \bgamma_n(\theta)) \right]^{-1} n \cdot \mcl{H}_{n, \bgamma \theta } (\theta, \tilde \bgamma_n(\theta)),
    \end{align*}
    where the right-hand side exists w.p.a.1 for $\theta$ in a neighborhood of $\theta_0$ by Lemma \ref{lem:A_rate1}(i) and Assumption \ref{as:hessian}(i).
    Then, by the second-order Taylor expansion,
    \begin{align*}
        0 
        & \le \mcl{L}_n(\hat \theta_n, \tilde \bgamma_n(\hat \theta_n)) - \mcl{L}_n(\theta_0, \tilde \bgamma_n(\theta_0)) \\
        & =  \mcl{S}_{n, \theta}(\theta_0, \tilde \bgamma_n(\theta_0))^\top (\hat \theta_n - \theta_0) \\
        & \quad + \frac{1}{2} (\hat \theta_n - \theta_0)^\top \left[ \mcl{H}_{n, \theta \theta}(\bar \theta_n, \tilde\bgamma_n(\bar \theta_n)) + \mcl{H}_{n, \theta \bgamma} (\bar \theta_n, \tilde\bgamma_n(\bar \theta_n)) \{ \partial_{\theta^\top} \tilde \bgamma_n(\bar \theta_n) \} \right] (\hat \theta_n - \theta_0) \\
        & \le \left\| \mcl{S}_{n, \theta}(\theta_0, \tilde \bgamma_n(\theta_0)) \right\| \cdot \left\| \hat \theta_n - \theta_0 \right\| - \frac{1}{2}\lambda_\text{min} \left( - \mcl{I}_{n, \theta \theta}(\bar \theta_n, \tilde \bgamma_n(\bar \theta_n)) \right)  \cdot \left\| \hat \theta_n - \theta_0 \right\|^2 \\
        & \Longrightarrow \left\| \hat \theta_n - \theta_0 \right\| \le \frac{2 \left\| \mcl{S}_{n, \theta}(\theta_0, \tilde \bgamma_n(\theta_0))  \right\|}{\lambda_\text{min} \left( - \mcl{I}_{n, \theta \theta}(\bar \theta_n, \tilde \bgamma_n(\bar \theta_n)) \right)},
    \end{align*}
    where $\bar \theta_n \in [\hat \theta_n, \theta_0]$.
    Here, since Theorem \ref{thm:consistency}(i) and Lemma \ref{lem:A_rate1}(i) imply that $\tilde \bgamma_n(\bar \theta_n)$ is uniformly consistent for $\bgamma_0$, $\lambda_\text{min} \left(- \mcl{I}_{n, \theta \theta}(\bar \theta_n, \tilde \bgamma_n(\bar \theta_n))  \right) > c_\theta$ w.p.a.1 by Assumption \ref{as:hessian}(ii).
    Thus, it suffices to show that $||\mcl{S}_{n, \theta}(\theta_0, \tilde \bgamma_n(\theta_0))|| = O_P(n^{-1/2})$.
    
    We decompose $\mcl{S}_{n, \theta}(\theta_0, \tilde \bgamma_n(\theta_0))$ into the following two terms:
    \begin{align*}
        \mcl{S}_{n, \theta}(\theta_0, \tilde \bgamma_n(\theta_0))
        & = \frac{2}{N} \sum_{i = 1}^n \sum_{j > i} s^\theta_{i,j}(\theta_0, \bgamma_0) + \frac{2}{N} \sum_{i = 1}^n \sum_{j > i} [s^\theta_{i,j}(\theta_0, \tilde \bgamma_n(\theta_0)) - s^\theta_{i,j}(\theta_0, \tilde \bgamma_0(\theta_0))] \\
        & \equiv  s^\theta_{1,n} + s^\theta_{2, n}, \;\; \text{say.}
    \end{align*}
    Since $\E s^\theta_{1,n} = \mbf{0}_{(d_z + 2) \times 1}$, by Assumption \ref{as:error}(ii),
    \begin{align*}
        \text{Var} \left[ s^\theta_{1,n} \right] 
        & = \frac{4}{N^2} \sum_{i = 1}^n \sum_{j > i}\sum_{k = 1}^n \sum_{l > k}\E[ s^\theta_{i,j} s^{\theta\top}_{k,l}] \\
        & = \frac{4}{N^2} \sum_{i = 1}^n \sum_{j > i}\E[ s^\theta_{i,j} s^{\theta\top}_{i,j} + s^\theta_{i,j} s^{\theta\top}_{j,i}] = O(n^{-2}).
    \end{align*}
    Thus, by Chebyshev's inequality, $||s^\theta_{1,n}|| = O_P(n^{-1})$.
    For $s^\theta_{2,n}$, observe that there exist bounded functions $s_{1,i,j}^\theta (\theta, \bgamma) \equiv \partial_{\pi_{i,j}} s^\theta_{i,j}(\theta, \bgamma)$ and $s_{2,i,j}^\theta(\theta, \bgamma) \equiv \partial_{\pi_{j,i}} s^\theta_{i,j}(\theta, \bgamma)$, such that 
    \begin{align}\label{eq:s_theta_A}
        \begin{split}
        & s^\theta_{i,j}(\theta_0, \tilde \bgamma_n(\theta_0)) - s^\theta_{i,j}(\theta_0, \tilde \bgamma_0(\theta_0)) \\
        & = \partial_{\mbf{A}_{-1}^\top} s^\theta_{i,j}(\theta_0, \bar \bgamma_n) [\tilde{\mbf{A}}_{n,-1}(\theta_0) - \tilde{\mbf{A}}_{0,-1}(\theta_0)] + \partial_{\mbf{B}^\top} s^\theta_{i,j}(\theta_0, \bar \bgamma_n) [\tilde{\mbf{B}}_n(\theta_0) - \tilde{\mbf{B}}_0(\theta_0)] \\
        & = s_{1,i,j}^\theta(\theta_0, \bar \bgamma_n)\chi_{n,i,-1}^\top [\tilde{\mbf{A}}_{n,-1}(\theta_0) - \tilde{\mbf{A}}_{0,-1}(\theta_0)] + s_{2,i,j}^\theta(\theta_0, \bar \bgamma_n)\chi_{n,j,-1}^\top [\tilde{\mbf{A}}_{n,-1}(\theta_0) - \tilde{\mbf{A}}_{0,-1}(\theta_0)] \\
        & \quad + s_{1,i,j}^\theta(\theta_0, \bar \bgamma_n)\chi_{n,j}^\top [\tilde{\mbf{B}}_n(\theta_0) - \tilde{\mbf{B}}_0(\theta_0)] + s_{2,i,j}^\theta(\theta_0, \bar \bgamma_n)\chi_{n,i}^\top [\tilde{\mbf{B}}_n(\theta_0) - \tilde{\mbf{B}}_0(\theta_0)] \\ 
        & \equiv t_{1,i,j} + t_{2,i,j} + t_{3,i,j} + t_{4,i,j}, \;\; \text{say,}
        \end{split}
    \end{align}
    where $\bar \bgamma_n \in [\tilde \bgamma_n(\theta_0), \tilde \bgamma_0(\theta_0)]$.
    Thus, for some $c > 0$, 
    \begin{align*}
        \left\| \frac{2}{N} \sum_{i = 1}^n \sum_{j > i}  t_{1,i,j} \right\| 
        & \le \frac{c}{N}\sum_{i=1}^n \sum_{j > i} \left| \chi_{n,i,-1}^\top [\tilde{\mbf{A}}_{n,-1}(\theta_0) - \tilde{\mbf{A}}_{0,-1}(\theta_0)] \right| \\
        & = \frac{c}{n}\sum_{i=2}^n \left| \tilde{A}_{n,i}(\theta_0) - \tilde{A}_{0,i}(\theta_0) \right| = O_P(n^{-1/2})
    \end{align*}
    by Lemma \ref{lem:A_rate1}(ii).
    Similarly, it is straightforward to see that $|| \frac{2}{N} \sum_{i = 1}^n \sum_{j > i}  t_{2,i,j} || = O_P(n^{-1/2})$ and that $|| \frac{2}{N} \sum_{i = 1}^n \sum_{j > i}  t_{3,i,j} || = O_P(n^{-1/2})$ and $|| \frac{2}{N} \sum_{i = 1}^n \sum_{j > i}  t_{4,i,j} || = O_P(n^{-1/2})$ by Lemma \ref{lem:A_rate1}(iii).
    Hence, we have $||s_{2,n}^\theta|| = O_P(n^{-1/2})$, and this completes the proof.
\end{proof}

\bigskip
\textbf{Proof of Theorem \ref{thm:conv_rate}}

(i) (ii) By the first-order condition and the mean value expansion,
\begin{align*}
    \mbf{0}_{(2n -1) \times 1} 
    & = n \cdot \mcl{S}_{n,\bgamma}(\hat \theta_n, \hat \bgamma_n) \\
    & = n \cdot \mcl{S}_{n,\bgamma} + n \cdot \mcl{S}_{n,\bgamma} (\hat \theta_n, \hat \bgamma_n)  - n \cdot \mcl{S}_{n, \bgamma} ( \theta_0, \hat \bgamma_n) + n \cdot \mcl{S}_{n,\bgamma} ( \theta_0, \hat \bgamma_n) - n \cdot \mcl{S}_{n, \bgamma} ( \theta_0, \bgamma_0) \\   
    & = n \cdot \mcl{S}_{n,\bgamma} + n \cdot \partial_{ \theta^\top } \mcl{S}_{n, \bgamma}(\bar \theta_n, \hat \bgamma_n) [\hat \theta_n - \theta_0] - \left( - n \cdot \mcl{H}_{n, \bgamma \bgamma } (\theta_0, \bar \bgamma_n) \right) [\hat \bgamma_{n,-1} - \bgamma_{0,-1}],
\end{align*}
where $\bar \theta_n \in [\hat \theta_n, \theta_0]$, and $\bar \bgamma_n \in [\hat \bgamma_n, \bgamma_0]$.
Thus, under Assumption \ref{as:hessian}(i),
\begin{align*}
    \hat \bgamma_{n,-1} - \bgamma_{0,-1}
    & = \left( - n \cdot \mcl{H}_{n, \bgamma \bgamma } (\theta_0, \bar \bgamma_n) \right)^{-1} n \cdot \mcl{S}_{n,\bgamma} + \left( - n \cdot \mcl{H}_{n, \bgamma \bgamma } (\theta_0, \bar\bgamma_n) \right)^{-1} n \cdot \partial_{ \theta^\top } \mcl{S}_{n, \bgamma}(\bar \theta_n, \hat \bgamma_n) [\hat \theta_n - \theta_0].
\end{align*}
As shown in the proof of Lemma \ref{lem:A_rate1}(ii)--(iii), $||n \cdot \mcl{S}_{n,\bgamma}|| = O_P(1)$.
For the second term, noting that $\partial_{ \theta^\top } \mcl{S}_{n, \bgamma}(\theta, \bgamma) = \frac{2}{N} \sum_{i = 1}^n \sum_{j > i} [\partial_{\mbf{A}_{-1}^\top} s_{i,j}^\theta (\theta, \bgamma),  \: \partial_{\mbf{B}^\top} s_{i,j}^\theta (\theta, \bgamma) ]^\top$.
Hence,
\begin{align*}
    \left\| n \cdot \partial_{ \theta^\top } \mcl{S}_{n, \bgamma}(\theta, \bgamma)  \right\|^2 
    & = n^2 \cdot \text{tr} \left\{ \left[ \partial_{ \theta^\top } \mcl{S}_{n, \bgamma}(\theta, \bgamma) \right]^\top \partial_{ \theta^\top } \mcl{S}_{n, \bgamma}(\theta, \bgamma) \right\} \\
    & = \text{tr} \left\{ \frac{4}{(n-1)^2} \sum_{i = 1}^n \sum_{j > i} \sum_{k = 1}^n \sum_{l > k} \left[\partial_{\mbf{A}_{-1}^\top} s_{i,j}^\theta (\theta, \bgamma) \right] \left[\partial_{\mbf{A}_{-1}^\top} s_{k,l}^\theta (\theta, \bgamma) \right]^\top \right\} \\
    & \quad + \text{tr} \left\{ \frac{4}{(n-1)^2} \sum_{i = 1}^n \sum_{j > i} \sum_{k = 1}^n \sum_{l > k} \left[\partial_{\mbf{B}^\top} s_{i,j}^\theta (\theta, \bgamma) \right] \left[\partial_{\mbf{B}^\top} s_{k,l}^\theta (\theta, \bgamma) \right]^\top \right\} \\
    & \equiv u^\mbf{A}_n(\theta, \bgamma) + u^\mbf{B}_n(\theta, \bgamma), \;\; \text{say.}
\end{align*}
Further,
\begin{align*}
    u^\mbf{A}_n(\theta, \bgamma) 
    & = \frac{4}{(n-1)^2} \sum_{i = 1}^n \sum_{j > i} \sum_{k = 1}^n \sum_{l > k} \text{tr}\left\{ \left[ s^\theta_{1,i,j}(\theta, \bgamma)\chi_{n,i,-1}^\top + s^\theta_{2,i,j}(\theta, \bgamma)\chi_{n,j,-1}^\top \right] \left[ \chi_{n,k,-1} s^\theta_{1,k,l}(\theta, \bgamma)^\top + \chi_{n,l,-1} s^\theta_{2,k,l}(\theta, \bgamma)^\top \right]  \right\} \\
    & = \frac{4}{(n-1)^2} \sum_{i = 1}^n \sum_{j > i} \sum_{k = 1}^n \sum_{l > k} \left[ \chi_{n,i,-1}^\top \chi_{n,k,-1} \cdot s^\theta_{1,k,l}(\theta, \bgamma)^\top s^\theta_{1,i,j}(\theta, \bgamma) + \chi_{n,j,-1}^\top \chi_{n,k,-1} \cdot s^\theta_{1,k,l}(\theta, \bgamma)^\top s^\theta_{2,i,j}(\theta, \bgamma)  \right. \\
    & \hspace{130pt}  + \left.  \chi_{n,i,-1}^\top\chi_{n,l,-1} \cdot s^\theta_{2,k,l}(\theta, \bgamma)^\top s^\theta_{1,i,j}(\theta, \bgamma) + \chi_{n,j,-1}^\top \chi_{n,l,-1} \cdot s^\theta_{2,k,l}(\theta, \bgamma)^\top s^\theta_{2,i,j}(\theta, \bgamma)  \right].
\end{align*}
Noting that $\chi_{n,i,-1}^\top\chi_{n,k,-1} = \mbf{1}\{i = k > 1\}$, we have
\begin{align*}
    \frac{4}{(n-1)^2} \sum_{i = 1}^n \sum_{j > i} \sum_{k = 1}^n \sum_{l > k} \chi_{n,i,-1}^\top \chi_{n,k,-1} \cdot s^\theta_{1,k,l}(\theta, \bgamma)^\top s^\theta_{1,i,j}(\theta, \bgamma)
    & = \frac{4}{(n-1)^2} \sum_{i = 2}^n \sum_{j > i} \sum_{l > i} s^\theta_{1,i,l}(\theta, \bgamma)^\top s^\theta_{1,i,j}(\theta, \bgamma) \\
    & = O(n)
\end{align*}
for any $(\theta, \bgamma) \in \Theta \times \mbb{C}_n$.
Applying the same discussion to the other terms, we obtain $ u^\mbf{A}_n(\theta, \bgamma) = O(n)$.
By the same argument, we can easily show that $ u^\mbf{B}_n(\theta, \bgamma) = O(n)$ for any $(\theta, \bgamma) \in \Theta \times \mbb{C}_n$.
Then, combined with Lemma \ref{lem:theta_rate}, we obtain $||n \cdot \partial_{ \theta^\top } \mcl{S}_{n, \bgamma}(\bar \theta_n, \hat \bgamma_n) [\hat \theta_n - \theta_0]|| \le ||n \cdot \partial_{ \theta^\top } \mcl{S}_{n, \bgamma}(\bar \theta_n, \hat \bgamma_n)|| \cdot || \hat \theta_n - \theta_0|| = O_P(1)$.

From these results, under Assumption \ref{as:hessian}(i), it holds that $||\hat \bgamma_n - \bgamma_0 || = O_P(1)$.
Finally, we obtain the desired result by the basic norm inequality, as in the proof of Lemma \ref{lem:A_rate1}(ii)--(iii).
\bigskip

(iii) Recall that, as in the proof of Lemma \ref{lem:A_rate1}(iv), we write $\mcl{L}_n(\theta, A_k, \bgamma_{-k}) = \mcl{L}_n(\theta, \bgamma)$ for $k \neq 1$.
By the first-order condition and mean value expansion, 
\begin{align*}
    0 
    & = n \cdot \partial_{A_k} \mcl{L}_n(\hat \theta_n, \hat A_{n,k}, \hat \bgamma_{n,-k}) \\
    & = \frac{2}{n-1}\sum_{i = 1}^n \sum_{j > i} s_{i,j}^{A_k}+ n \cdot \partial_{A_k} \mcl{L}_n(\hat \theta_n, \hat A_{n,k}, \hat \bgamma_{n,-k}) - n \cdot \partial_{A_k} \mcl{L}_n(\hat \theta_n, A_{0,k}, \hat \bgamma_{n,-k}) \\
    & \quad + n \cdot \partial_{A_k} \mcl{L}_n(\hat \theta_n, A_{0,k}, \hat \bgamma_{n,-k}) - n \cdot \partial_{A_k} \mcl{L}_n(\theta_0, A_{0,k}, \hat \bgamma_{n,-k}) + n \cdot \partial_{A_k} \mcl{L}_n(\theta_0, A_{0,k}, \hat \bgamma_{n,-k}) - n \cdot \partial_{A_k} \mcl{L}_n( \theta_0, A_{0,k}, \bgamma_{0,-k})  \\   
    & = \frac{2}{n-1}\sum_{i = 1}^n \sum_{j > i} s_{i,j}^{A_k} + \frac{2}{n-1}\sum_{j \neq k} h_{11,k,j} (\hat \theta_n, \bar A_{n,k}, \hat \bgamma_{n,-k}) [\hat A_{n,k} - A_{0,k}] \\
    & \quad + \frac{2}{n-1}\sum_{i = 1}^n \sum_{j > i} \partial_{\theta^\top} s_{i,j}^{A_k}(\bar \theta_n, A_{0,k}, \hat \bgamma_{n,-k}) [\hat \theta_n - \theta_0] + \frac{2}{n-1}\sum_{i = 1}^n \sum_{j > i} \partial_{\bgamma_{-k}^\top} s_{i,j}^{A_k}(\theta_0, A_{0,k}, \bar \bgamma_{n,-k}) [\hat \bgamma_{n,-k} - \bgamma_{0,-k}],
\end{align*}
where $\bar \theta_n \in [\hat \theta_n, \theta_0]$, $\bar A_{n,k} \in [\hat A_{n,k}, A_{0,k}]$, and $\bar \bgamma_{n,-k} \in [\hat \bgamma_{n,-k}, \bgamma_{0,-k}]$.
Then, similar to the proof of Lemma \ref{lem:A_rate1}(iv), we have
\begin{align*}
    \left| \hat A_{n,k} - A_{0,k} \right| \le (c + o_p(1))\left\{ \left| T_{1,n,k} \right| + \left| T_{2,n,k} \right| + \left| T_{3,n,k} \right| \right\}
\end{align*}
for some $c > 0$, where
\begin{align*}
    T_{1,n,k} 
    & \equiv \frac{2}{n-1}\sum_{i = 1}^n \sum_{j > i} s_{i,j}^{A_k}, \quad
    T_{2,n,k} \equiv \frac{2}{n-1}\sum_{i = 1}^n \sum_{j > i} \partial_{\bgamma_{-k}^\top} s_{i,j}^{A_k}(\theta_0, A_{0,k}, \bar \bgamma_{n,-k}) [\hat \bgamma_{n,-k} - \bgamma_{0,-k}], \\
    T_{3,n,k} & \equiv \frac{2}{n-1}\sum_{i = 1}^n \sum_{j > i} \partial_{\theta^\top} s_{i,j}^{A_k}(\bar \theta_n, A_{0,k}, \hat \bgamma_{n,-k}) [\hat \theta_n - \theta_0].
\end{align*}
As shown in \eqref{eq:T_1} and \eqref{eq:T_2}, $\max_{1 \le k \le n} |T_{1,n,k}| = O_P(\sqrt{\ln n/n})$ and $\max_{1 \le k \le n} |T_{2,n,k}| = O_P(n^{-1/2})$.
For $T_{3,n,k}$, observe that
\begin{align*}
    \frac{2}{n-1}\sum_{i = 1}^n \sum_{j > i} \partial_{\theta^\top} s_{i,j}^{A_k}(\theta, \bgamma) 
    & = \frac{2}{n-1} \sum_{i=1}^n \sum_{j > i} \left[ \mbf{1}\{i = k\} s_{1,i,j}^\theta(\theta, \bgamma)^\top + \mbf{1}\{j = k\} s_{2,i,j}^\theta(\theta, \bgamma)^\top \right] \\
    & = \frac{2}{n-1} \sum_{j > k} s_{1,k,j}^\theta(\theta, \bgamma)^\top + \frac{2}{n-1} \sum_{i < k} s_{2,i,k}^\theta(\theta, \bgamma)^\top.
\end{align*}
Hence, clearly, $\max_{1\le k \le n}||\frac{2}{n-1}\sum_{i = 1}^n \sum_{j > i} \partial_{\theta^\top} s_{i,j}^{A_k}(\theta, \bgamma)|| = O(1)$ for any $(\theta, \bgamma) \in \Theta \times \mbb{C}_n$.
Then, with Lemma \ref{lem:theta_rate},
\begin{align*}
    \max_{1 \le k \le n}|T_{3,n,k}| = O_P(n^{-1/2}).
\end{align*}
Combining these results yields $\max_{1 \le k \le n}|\hat A_{n,k} - A_{0,k}| = O_P(\sqrt{\ln n / n}) + O_P(n^{-1/2})$.

Similarly as above, we can also show that $\max_{1 \le k \le n}|\hat B_{n,k} - B_{0,k}| = O_P(\sqrt{\ln n / n}) + O_P(n^{-1/2})$.
This completes the proof.
\qed

\bigskip
\textbf{Proof of Theorem \ref{thm:grouping}}

To simplify the discussion, we focus on the estimation of $\mcl{C}_0^A$ with $K^A = 3$ only.
The other cases can be proved analogously.
In addition, for notational simplicity, we omit the superscript $A$ in this proof.

Now, let
\begin{align*}
    u_{n,i} \equiv \hat A_{n,i} - A_{0,i} \;\; \text{for} \;\; i = 1, \ldots, n.
\end{align*}
In particular, $u_{n,1} = 0$ holds by the normalization.
In accordance with the ordering $\hat A_{n,(1)} \le \cdots \le \hat A_{n,(n)}$, we permutate $A_{0,i}$'s and obtain $\{A_{0,(i)}\}$.
By Theorem \ref{thm:conv_rate}(iii), we have $\max_{1 \le i \le n}|u_{n,i}| = O_P(\sqrt{\ln n / n})$.
Hence, w.p.a.1, the sequence $\{ A_{0,(i)} \}$ contains two ``true'' break points $(t^0_1, t^0_2) \equiv (t^0_{n,1}, t^0_{n,2})$ in the following manner:
\begin{align*}
    A_{0,(i)}
    = \begin{cases}
    a_{0,1} \;\; \text{if} \;\; 1 \le i \le t^0_1 \\
    a_{0,2} \;\; \text{if} \;\; t^0_1 + 1 \le i \le t^0_2  \\
    a_{0,3} \;\; \text{if} \;\; t^0_2 + 1 \le i \le n.
    \end{cases}
\end{align*}
We can assume, without loss of generality, that $\hat S_{1,n}(t_1^0) < \hat S_{1,n}(t_2^0)$.
If $\hat S_{1,n}(t_1^0) > \hat S_{1,n}(t_2^0)$, by reversing the order of $\{\hat A_{n,(i)}\}$ and re-labeling the break points appropriately, we can prove the theorem completely analogously.
Recall that $\hat t_1 = \argmin_{1 \le \kappa < n} \hat S_{1,n}(\kappa)$.
We first show that $\Pr(\hat t_1 = t^0_1) \to 1$ by demonstrating that (i) $\Pr(\hat t_1 < t^0_1) \to 0$, (ii) $\Pr(t^0_1 < \hat t_1 \le t^0_2) \to 0$, and (iii) $\Pr(t^0_2 < \hat t_1) \to 0$.
\bigskip

(i) For a given $m < t_1^0$, we have
\begin{align*}
    \bar A_{n,1,m} 
    & = \frac{1}{m}\sum_{l = 1}^m (A_{0,(l)} + u_{n,(l)}) = a_{0,1} + \bar u_{n,1,m} \\
    \bar A_{n,m + 1,n}
    & = \frac{1}{n - m}\sum_{l = m + 1}^n  (A_{0,(l)} + u_{n,(l)}) = \frac{(t_1^0 - m) a_{0,1}}{n - m} + \frac{(t^0_2 - t^0_1) a_{0,2}}{n - m} + \frac{(n - t^0_2) a_{0,3}}{n - m} + \bar u_{n,m + 1,n},
\end{align*}
where $\bar u_{n,1,m} = m^{-1}\sum_{l = 1}^m u_{n,(l)}$, and $\bar u_{n,m + 1,n} = (n - m)^{-1}\sum_{l = m + 1}^n u_{n,(l)}$.
Hence, since $\hat A_{n,(l)} - \bar A_{n,1,m} = u_{n,(l)} - \bar u_{n,1,m}$ for $l \le m$, we have
\begin{align*}
    \hat \Delta(1,m) = \sum_{l=1}^m (u_{n,(l)} - \bar u_{n,1,m})^2.
\end{align*}
Similarly, since
\begin{align*} 
    \hat A_{n,(l)} - \bar A_{n,m + 1,n} = \begin{cases}
        a_{1m} + u_{n,(l)} - \bar u_{n,m + 1,n} & \text{if} \;\; m + 1 \le l \le t_1^0 \\
        a_{2m} + u_{n,(l)} - \bar u_{n,m + 1,n} & \text{if} \;\; t_1^0 + 1 \le l \le t_2^0 \\
        a_{3m} + u_{n,(l)} - \bar u_{n,m + 1,n} & \text{if} \;\; t_2^0 + 1 \le l \le n,
    \end{cases}
\end{align*}
where $a_{1m} \equiv  \frac{(t^0_2 - t^0_1) (a_{0,1} - a_{0,2})}{n - m} + \frac{(n - t^0_2) (a_{0,1} - a_{0,3})}{n - m}$, $a_{2m} \equiv  \frac{(t^0_1 - m) (a_{0,2} - a_{0,1})}{n - m} + \frac{(n - t^0_2) (a_{0,2} - a_{0,3})}{n - m}$, and $a_{3m} \equiv \frac{(t^0_1 - m) (a_{0,3} - a_{0,1})}{n - m} + \frac{(t^0_2 - t^0_1) (a_{0,3} - a_{0,2})}{n - m}$, we have
\begin{align*}
    \hat \Delta(m + 1, n)
    & = (t_1^0 - m) a_{1m}^2 + (t_2^0 - t_1^0) a_{2m}^2 + (n - t_2^0) a_{3m}^2 + \sum_{l= m + 1}^n (u_{n,(l)} - \bar u_{n,m + 1,n})^2 \\
    & \quad + 2  a_{1m} \sum_{l= m + 1}^{t_1^0} (u_{n,(l)} - \bar u_{n,m + 1,n}) + 2 a_{2m} \sum_{l= t_1^0 + 1}^{t_2^0} (u_{n,(l)} - \bar u_{n,m + 1,n}) + 2 a_{3m} \sum_{l= t_2^0 + 1}^n (u_{n,(l)} - \bar u_{n,m + 1,n}).
\end{align*}
Then, it holds that $\hat S_{1,n}(m) = \frac{1}{n}\left( \hat \Delta(1,m) + \hat \Delta(m + 1, n) \right) = \mu_1(m) + r_1(m)$, where
\begin{align*}
    \mu_1(m) 
    & \equiv \frac{t_1^0 - m}{n} a_{1m}^2 + \frac{t_2^0 - t_1^0}{n} a_{2m}^2 + \frac{n - t_2^0}{n} a_{3m}^2 \\
    r_1(m)
    & \equiv \frac{1}{n}\left[ \sum_{l=1}^m (u_{n,(l)} - \bar u_{n,1,m})^2 + \sum_{l= m + 1}^n (u_{n,(l)} - \bar u_{n,m + 1,n})^2 \right] \\
    & \quad + \frac{2  a_{1m}}{n} \sum_{l= m + 1}^{t_1^0} (u_{n,(l)} - \bar u_{n,m + 1,n}) + \frac{2 a_{2m}}{n} \sum_{l= t_1^0 + 1}^{t_2^0} (u_{n,(l)} - \bar u_{n,m + 1,n}) + \frac{2 a_{3m}}{n} \sum_{l= t_2^0 + 1}^n (u_{n,(l)} - \bar u_{n,m + 1,n}).
\end{align*}
By similar calculation, we can observe $\hat S_{1,n}(t_1^0) = \frac{1}{n}\left( \hat \Delta(1,t_1^0) + \hat \Delta(t_1^0 + 1, n) \right) = \mu_1(t_1^0) + r_1(t_1^0)$, where
\begin{align*}
    \hat \Delta(1,t_1^0) 
    & = \sum_{l=1}^{t_1^0} (u_{n,(l)} - \bar u_{n,1,t_1^0})^2 \\
    \hat \Delta(t_1^0 + 1, n)
    & =  (t_2^0 - t_1^0) a_{2 t_1^0}^2 + (n - t_2^0) a_{3 t_1^0}^2 + \sum_{l= t_1^0 + 1}^n (u_{n,(l)} - \bar u_{n, t_1^0 + 1,n})^2 \\
    & \quad + 2 a_{2 t_1^0} \sum_{l= t_1^0 + 1}^{t_2^0} (u_{n,(l)} - \bar u_{n,t_1^0 + 1,n}) + 2 a_{3 t_1^0} \sum_{l= t_2^0 + 1}^n (u_{n,(l)} - \bar u_{n,t_1^0 + 1,n}),
\end{align*}
with $a_{2t_1^0} \equiv  \frac{(n - t^0_2) (a_{0,2} - a_{0,3})}{n - t_1^0}$, and $a_{3t_1^0} \equiv  \frac{(t^0_2 - t_1^0) (a_{0,3} - a_{0,2})}{n - t_1^0}$, and
\begin{align*}
    \mu_1(t_1^0) 
    & \equiv \frac{t_2^0 - t_1^0}{n} a_{2 t_1^0}^2  + \frac{n - t_2^0}{n} a_{3 t_1^0}^2 \\
    r_1(t_1^0)
    & \equiv \frac{1}{n}\left[ \sum_{l=1}^{t_1^0} (u_{n,(l)} - \bar u_{n,1,t_1^0})^2 + \sum_{l= t_1^0 + 1}^n (u_{n,(l)} - \bar u_{n,t_1^0 + 1,n})^2 \right] \\
    & \quad + \frac{2 a_{2 t_1^0}}{n} \sum_{l= t_1^0 + 1}^{t_2^0} (u_{n,(l)} - \bar u_{n,t_1^0 + 1,n}) + \frac{2 a_{3t_1^0}}{n} \sum_{l= t_2^0 + 1}^n (u_{n,(l)} - \bar u_{n, t_1^0 + 1,n}).
\end{align*}
By tedious calculation, we can find that
\begin{align*}
  \mu_1(m) - \mu_1(t_1^0) 
   = \frac{(t_1^0 - m)}{n(1 - m/n)(1 - t_1^0/n)} \left[ \left(1 - \frac{t_1^0}{n}\right) (a_{0,1} - a_{0,2}) + \left(1 - \frac{t_2^0}{n} \right) (a_{0,2} - a_{0,3}) \right]^2.
\end{align*}
Note that $\mu_1(m) - \mu_1(t_1^0)  > 0$ for any $m < t_1^0$ under Assumption \ref{as:bs}(i).
Meanwhile,
\begin{align*}
   r_1(m) - r_1(t_1^0) 
   & = \frac{1}{n}\left[ \sum_{l=1}^m (u_{n,(l)} - \bar u_{n,1,m})^2 - \sum_{l=1}^{t_1^0} (u_{n,(l)} - \bar u_{n,1,t_1^0})^2 \right] \\
   & \quad + \frac{1}{n}\left[ \sum_{l= m + 1}^n (u_{n,(l)} - \bar u_{n,m + 1,n})^2 - \sum_{l= t_1^0 + 1}^n (u_{n,(l)} - \bar u_{n,t_1^0 + 1,n})^2 \right] \\
   & \quad + \frac{2 a_{1m}}{n} \sum_{l= m + 1}^{t_1^0} (u_{n,(l)} - \bar u_{n,m + 1,n}) \\
   & \quad + \frac{2}{n} \left( a_{2m}\sum_{l= t_1^0 + 1}^{t_2^0} (u_{n,(l)} - \bar u_{n,m + 1,n}) - a_{2 t_1^0} \sum_{l= t_1^0 + 1}^{t_2^0} (u_{n,(l)} - \bar u_{n,t_1^0 + 1,n}) \right) \\
   & \quad + \frac{2}{n} \left( a_{3m}\sum_{l= t_2^0 + 1}^n (u_{n,(l)} - \bar u_{n,m + 1,n}) - a_{3t_1^0}{n} \sum_{l= t_2^0 + 1}^n (u_{n,(l)} - \bar u_{n, t_1^0 + 1,n}) \right).
\end{align*}
Then, by carefully examining each term of the right hand,\footnote{
        For this calculation, see pp. 8--9 in Supplementary Material for \cite{wang2020identifying}.
    }
we can find that $r_1(m) - r_1(t_1^0) = \frac{t_1^0 - m}{n} \cdot \{O_P(\ln n/ n) + O_P(\sqrt{\ln n/ n})\}$.
Hence, we have 
\begin{align*}
    \hat S_{1,n}(m) - \hat S_{1,n}(t_1^0) 
    & = \underbrace{\mu_1(m) - \mu_1(t_1^0)}_{= \: n^{-1}(t_1^0 - m) \cdot c_n \; (c_n > 0)} + \underbrace{r_1(m) - r_1(t_1^0)}_{n^{-1}(t_1^0 - m) \cdot o_P(1)} \\
    & \ge 0
\end{align*}
w.p.a.1 for any $m < t_1^0$.
Thus, since $\hat t_1$ is a minimizer of $\hat S_{1,n}(m)$,
\begin{align*}
    \Pr( \hat t_1 < t_1^0) 
    & = \Pr(\hat S_{1,n}(\hat t_1) - \hat S_{1,n}(t_1^0) \ge 0, \: \hat t_1 < t_1^0) + \Pr(\hat S_{1,n}(\hat t_1) - \hat S_{1,n}(t_1^0) < 0, \: \hat t_1 < t_1^0) \\
    & = \Pr(\hat S_{1,n}(\hat t_1) - \hat S_{1,n}(t_1^0) < 0, \: \hat t_1 < t_1^0) \\
    & \le \Pr(\hat S_{1,n}(m) - \hat S_{1,n}(t_1^0) < 0, \: m < t_1^0 \;\; \text{for some $m$}) \to 0
\end{align*}
as $n \to \infty$.
Thus, (i) follows.

(ii) For a given $t_1^0 < m \le t_2^0$, we have
\begin{align*}
    \bar A_{n,1,m} 
    & = \frac{1}{m}\sum_{l = 1}^m (A_{0,(l)} + u_{n,(l)}) = \frac{t_1^0}{m} a_{0,1} + \frac{m - t_1^0}{m} a_{0,2} + \bar u_{n,1,m} \\
    \bar A_{n,m + 1,n}
    & = \frac{1}{n - m}\sum_{l = m + 1}^n  (A_{0,(l)} + u_{n,(l)}) =  \frac{(t^0_2 - m) a_{0,2}}{n - m} + \frac{(n - t^0_2) a_{0,3}}{n - m} + \bar u_{n,m + 1,n}.
\end{align*}
Hence, since
\begin{align*} 
    \hat A_{n,(l)} - \bar A_{n,1,m} = \begin{cases}
        a_{1m} + u_{n,(l)} - \bar u_{n,1,m} & \text{if} \;\; 1 \le l \le t_1^0 \\
        a_{2m} + u_{n,(l)} - \bar u_{n,1,m} & \text{if} \;\; t_1^0 + 1 \le l \le m,
    \end{cases}
\end{align*}
where $a_{1m} \equiv  \frac{(m - t^0_1) (a_{0,1} - a_{0,2})}{m}$, and $a_{2m} \equiv \frac{t^0_1 (a_{0,2} - a_{0,1})}{m}$, we have
\begin{align*}
    \hat \Delta(1,m) 
    & = t_1^0 a_{1m}^2 + (m-t_1^0) a_{2m}^2 + \sum_{l=1}^m (u_{n,(l)} - \bar u_{n,1,m})^2\\
    & \quad + 2 a_{1m} \sum_{l=1}^{t_1^0} (u_{n,(l)} - \bar u_{n,1,m}) + 2 a_{2m} \sum_{l=t_1^0 + 1}^m (u_{n,(l)} - \bar u_{n,1,m}).
\end{align*}
Similarly, since
\begin{align*} 
    \hat A_{n,(l)} - \bar A_{n,m + 1,n} = \begin{cases}
        a_{3m} + u_{n,(l)} - \bar u_{n,m + 1,n} & \text{if} \;\; m + 1 \le l \le t_2^0 \\
        a_{4m} + u_{n,(l)} - \bar u_{n,m + 1,n} & \text{if} \;\; t_2^0 + 1 \le l \le n,
    \end{cases}
\end{align*}
where $a_{3m} \equiv  \frac{(n - t^0_2) (a_{0,2} - a_{0,3})}{n - m}$, and $a_{4m} \equiv  \frac{(t^0_2 - m) (a_{0,3} - a_{0,2})}{n - m}$, we have
\begin{align*}
    \hat \Delta(m + 1, n)
    & = (t_2^0 - m) a_{3m}^2 + (n - t_2^0) a_{4m}^2 + \sum_{l= m + 1}^n (u_{n,(l)} - \bar u_{n,m + 1,n})^2 \\
    & \quad + 2  a_{3m} \sum_{l= m + 1}^{t_2^0} (u_{n,(l)} - \bar u_{n,m + 1,n}) + 2 a_{4m} \sum_{l= t_2^0 + 1}^n (u_{n,(l)} - \bar u_{n,m + 1,n}).
\end{align*}
Then, it holds that $\hat S_{1,n}(m) = \frac{1}{n}\left( \hat \Delta(1,m) + \hat \Delta(m + 1, n) \right) = \mu_2(m) + r_2(m)$, where
\begin{align*}
    \mu_2(m) 
    & \equiv \frac{t_1^0}{n} a_{1m}^2 + \frac{m - t_1^0}{n} a_{2m}^2 + \frac{t_2^0 - m}{n} a_{3m}^2 + \frac{n - t_2^0}{n} a_{4m}^2 \\
    r_2(m)
    & \equiv \frac{1}{n}\left[ \sum_{l=1}^m (u_{n,(l)} - \bar u_{n,1,m})^2 + \sum_{l= m + 1}^n (u_{n,(l)} - \bar u_{n,m + 1,n})^2 \right] \\
    & \quad + \frac{2  a_{1m}}{n} \sum_{l=1}^{t_1^0} (u_{n,(l)} - \bar u_{n,1,m}) + \frac{2 a_{2m}}{n} \sum_{l=t_1^0 + 1}^m (u_{n,(l)} - \bar u_{n,1,m}) \\
    & \quad + \frac{2 a_{3m}}{n} \sum_{l= m + 1}^{t_2^0} (u_{n,(l)} - \bar u_{n,m + 1,n}) + \frac{2 a_{4m}}{n} \sum_{l= t_2^0 + 1}^n (u_{n,(l)} - \bar u_{n,m + 1,n}).
\end{align*}
To be consistent with the above notations, we re-write $\hat S_{1,n}(t_1^0) = \frac{1}{n}\left( \hat \Delta(1,t_1^0) + \hat \Delta(t_1^0 + 1, n) \right) = \mu_2(t_1^0) + r_2(t_1^0)$, where
\begin{align*}
    \mu_2(t_1^0) 
    & \equiv \frac{t_2^0 - t_1^0}{n} a_{3 t_1^0}^2  + \frac{n - t_2^0}{n} a_{4 t_1^0}^2 \\
    r_2(t_1^0)
    & \equiv \frac{1}{n}\left[ \sum_{l=1}^{t_1^0} (u_{n,(l)} - \bar u_{n,1,t_1^0})^2 + \sum_{l= t_1^0 + 1}^n (u_{n,(l)} - \bar u_{n,t_1^0 + 1,n})^2 \right] \\
    & \quad + \frac{2 a_{3 t_1^0}}{n} \sum_{l= t_1^0 + 1}^{t_2^0} (u_{n,(l)} - \bar u_{n,t_1^0 + 1,n}) + \frac{2 a_{4 t_1^0}}{n} \sum_{l= t_2^0 + 1}^n (u_{n,(l)} - \bar u_{n, t_1^0 + 1,n}),
\end{align*}
with $a_{3t_1^0} \equiv  \frac{(n - t^0_2) (a_{0,2} - a_{0,3})}{n - t_1^0}$, and $a_{4t_1^0} \equiv  \frac{(t^0_2 - t_1^0) (a_{0,3} - a_{0,2})}{n - t_1^0}$.
By tedious calculation, we can find that
\begin{align*}
    \mu_2(m) - \mu_2(t_1^0) 
    & = \frac{m - t_1^0}{n}\left[ \frac{t_1^0}{m}(a_{0,1} - a_{0,2})^2 - \frac{(n - t_2^0)^2}{(n - m)(n - t_1^0)}(a_{0,2} - a_{0,3})^2 \right]\\
    & \ge \frac{m - t_1^0}{n}\left[ \frac{t_1^0}{m}(a_{0,1} - a_{0,2})^2 - \frac{t_2^0 (n - t_2^0)}{m(n - t_1^0)}(a_{0,2} - a_{0,3})^2 \right]\\
    & = \frac{m - t_1^0}{n}\frac{t_2^0}{m}\left[ \frac{t_1^0}{t_2^0}(a_{0,1} - a_{0,2})^2 - \frac{n - t_2^0}{n - t_1^0}(a_{0,2} - a_{0,3})^2 \right],
\end{align*}
where the inequality follows because $(n-t_2^0)/(n - m) \le t_2^0/m$.
Further, it can be shown that $r_2(m) - r_2(t_1^0) = n^{-1}(m - t_1^0) \cdot o_P(1)$ uniformly in $t_1^0 < m \le t_2^0$.
Here, recall that we have assumed $\hat S_{1,n}(t_1^0) < \hat S_{1,n}(t_2^0)$.
Through straightforward calculation, we can find that 
\begin{align*}
    0 < \hat S_{1,n}(t_2^0) - \hat S_{1,n}(t_1^0)
    & = \frac{t_1^0}{n} a_{1 t_2^0}^2  + \frac{t_2^0 - t_1^0}{n} a_{2 t_2^0}^2 - \frac{t_2^0 - t_1^0}{n} a_{3 t_1^0}^2 - \frac{n - t_2^0}{n} a_{4 t_1^0}^2 + \frac{t_2^0 - t_1^0}{n} \cdot o_P(1)\\
    & = \frac{t^0_2 - t_1^0}{n} \left[ \frac{ t_1^0 }{t_2^0} (a_{0,1} - a_{0,2})^2 - \frac{n - t^0_2}{n - t_1^0} (a_{0,2} - a_{0,3})^2 \right] + \frac{t_2^0 - t_1^0}{n} \cdot o_P(1),
\end{align*}
where $a_{1t_2^0} \equiv  \frac{(t^0_2 - t_1^0) (a_{0,1} - a_{0,2})}{t_2^0}$, and $a_{2t_2^0} \equiv  \frac{t_1^0 (a_{0,2} - a_{0,1})}{t_2^0}$.
This implies that
\begin{align*}
    \frac{ t_1^0 }{t_2^0} (a_{0,1} - a_{0,2})^2 - \frac{n - t^0_2}{n - t_1^0} (a_{0,2} - a_{0,3})^2 
    \to \frac{ \tau_1 }{\tau_1 + \tau_2} (a_{0,1} - a_{0,2})^2 - \frac{\tau_3}{\tau_2 + \tau_3} (a_{0,2} - a_{0,3})^2 > 0
\end{align*}
by Assumption \ref{as:bs}(ii).
Hence, $\mu_2(m) - \mu_2(t_2^0) > 0$ holds for any $t_1^0 < m \le t_2^0$ for sufficiently large $n$.
Then, combining these results implies that $\hat S_{1,n}(m) - \hat S_{1,n}(t_1^0) > 0$ w.p.a.1 uniformly in $t_1^0 < m \le t_2^0$, leading to the desired result:
\begin{align*}
    \Pr( t_1^0 < \hat t_1 \le t_2^0) 
    & = \Pr(\hat S_{1,n}(\hat t_1) - \hat S_{1,n}(t_1^0) \ge 0, \: t_1^0 < \hat t_1 \le t_2^0) + \Pr(\hat S_{1,n}(\hat t_1) - \hat S_{1,n}(t_1^0) < 0, \: t_1^0 < \hat t_1 \le t_2^0) \\
    & = \Pr(\hat S_{1,n}(\hat t_1) - \hat S_{1,n}(t_1^0) < 0, \: t_1^0 < \hat t_1 \le t_2^0) \\
    & \le \Pr(\hat S_{1,n}(m) - \hat S_{1,n}(t_1^0) < 0, \: t_1^0 < m \le t_2^0 \;\; \text{for some $m$}) \to 0.
\end{align*}
Then, (ii) also follows. 
For case (iii), since this case is symmetric to (i), we can safely omit the proof.
\bigskip

Once the first break point is obtained, we can partition $\{\hat A_{n,(i)}\}$ into two subregions $\{\hat A_{n,(i)}\}_{i=1}^{\hat t_1}$ and $\{\hat A_{n,(i)}\}_{i = \hat t_1 + 1}^n$.
We next estimate in which group the second break point exists.
Given the above consistency result, w.p.a.1, we have $\hat S_{1,\hat t_1}(\hat t_1) = \hat S_{1, t_1^0}(t_1^0)$ and $\hat S_{\hat t_1 + 1, n}(n) = \hat S_{t_1^0 + 1, n}(n)$.
Since $t_1^0 < t_2^0$, it suffices to show that $\Pr(\hat S_{1, t_1^0}(t_1^0) < \hat S_{t_1^0 + 1, n}(n)) \to 1$.
We can observe that $\hat S_{1, t_1^0}(t_1^0) = \frac{1}{t_1^0} \hat \Delta(1, t_1^0) = r^*_{t_1^0}$ and $\hat S_{t_1^0 + 1, n}(n) = \frac{1}{n - t_1^0} \hat \Delta(t_1^0 + 1, n) = \mu^*_n + r^*_n$, where
\begin{align*}
    r_{t_1^0}^*
    & \equiv \frac{1}{t_1^0}\sum_{l = 1}^{t_1^0} (u_{n,(l)} - \bar u_{n,1,t_1^0})^2 \\
    \mu_n^*
    & \equiv \frac{t_2^0 - t_1^0}{n - t_1^0} (a_{1n}^*)^2 + \frac{n - t_2^0}{n - t_1^0} (a_{2n}^*)^2 \\
    r_n^*
    & \equiv \frac{1}{n - t_1^0}\sum_{l = t_1^0 + 1}^n (u_{n,(l)} - \bar u_{n,t_1^0 + 1, n})^2 + \frac{2 a_{1n}^*}{n - t_1^0} \sum_{l= t_1^0 + 1}^{t_2^0} (u_{n,(l)} - \bar u_{n,t_1^0 + 1,n}) + \frac{2 a_{2n}^*}{n - t_1^0} \sum_{l= t_2^0 + 1}^n (u_{n,(l)} - \bar u_{n,t_1^0 + 1,n}),   
\end{align*}
with $a^*_{1n} \equiv \frac{(n - t_2^0)(a_{0,2} - a_{0,3})}{n - t_1^0}$, and $a^*_{2n} \equiv \frac{(t_2^0 - t_1^0)(a_{0,3} - a_{0,2})}{n - t_1^0}$.
Through some calculation, we can find that
\begin{align*}
    \mu_n^* = \frac{(t_2^0 - t_1^0) (n - t_2^0) }{(n - t_1^0)^2} (a_{0,2} - a_{0,3})^2 \to \frac{\tau_2 \tau_3 }{(\tau_2 + \tau_3)^2} (a_{0,2} - a_{0,3})^2 > 0
\end{align*}
under Assumptions \ref{as:bs}(i) and (ii).
It can be easily seen that $r_n^* - r_{t_1^0}^* = o_P(1)$.
Hence, $\hat S_{t_1^0 + 1, n}(n) - \hat S_{1, t_1^0}(t_1^0) > 0$ holds w.p.a.1, as desired.
Then, once the subset $\{\hat A_{n,(i)}\}_{i = \hat t_1 + 1}^n$ is selected for the detection of the second break point, it can be estimated by $\hat t_2 = \argmin_{\hat t_1 + 1 \le \kappa < n} \hat S_{\hat t_1 + 1,n}(\kappa)$.
The consistency of $\hat t_2$ for $t_2^0$ follows from the same argument as above.
\qed

\bigskip
\textbf{Proof of Theorem \ref{thm:normality}}

The consistency of $\hat \delta_n^\mathrm{oracle}$ is straightforward from Theorem \ref{thm:consistency}.
By the first-order condition, Taylor expansion, and the law of large numbers, we have
\begin{align*}
    \sqrt{\frac{N}{2}}(\hat \delta_n^\mathrm{oracle} - \delta_0) = - \E\left[ \partial^2_{\delta \delta^\top} \mcl{L}_n(\delta_0) \right]^{-1} \sqrt{\frac{2}{N}} \sum_{i = 1}^n \sum_{j > i} s^\delta_{i,j}(\delta_0) + o_P(1). 
\end{align*}
Note that, under the assumptions made, $\{s^\delta_{i,j}(\delta_0)\}$ are uniformly bounded, and $\sum_{i = 1}^n \sum_{j > i} s^\delta_{i,j}(\delta_0)$ is a sum of independent random variables.
Thus, the asymptotic normality result follows from the central limit theorem for bounded random variables (see, e.g., Example 27.4 in \cite{billingsley2012probability}).

The asymptotic distributional equivalence result is straightforward from
\begin{align*}
    \Pr\left( \sqrt{\frac{N}{2}}(\hat \delta_n - \delta_0) \in C \right)
    & = \Pr\left( \sqrt{\frac{N}{2}}(\hat \delta_n - \delta_0) \in C, \; (\hat{\mcl{C}}^A_n , \hat{\mcl{C}}^B_n) = (\mcl{C}_0^A, \mcl{C}_0^B) \right) \\
    & \quad + \Pr\left( \sqrt{\frac{N}{2}}(\hat \delta_n - \delta_0) \in C, \; (\hat{\mcl{C}}^A_n , \hat{\mcl{C}}^B_n) \neq (\mcl{C}_0^A, \mcl{C}_0^B) \right) \\
    & = \Pr\left( \sqrt{\frac{N}{2}}(\hat \delta_n^\mathrm{oracle} - \delta_0) \in C \right) + o(1)
\end{align*}
for any $C \subseteq \mbb{R}^{d_z + K^A + K^B + 1}$.
\qed


\section{Supplementary Materials}
\subsection{Explicit form of $\partial^2 \mcl{L}_n^*(\rho)/(\partial \rho)^2$}\label{subsec:rho_unique}

For notational simplicity, let $\mbf{C} = (\beta^\top, \alpha, \bgamma^\top)^\top$ and $\tilde{\mbf{C}}_0(\rho) = (\tilde \beta_0(\rho)^\top, \tilde \alpha_0(\rho), \tilde \bgamma_0(\rho)^\top)^\top$, and write $\mcl{L}_n(\theta, \bgamma)$ equivalently as $\mcl{L}_n(\rho, \mbf{C})$, so that $\tilde{\mbf{C}}_0(\rho) = \argmax_{\mbf{C} \in \mcl{B} \times \mcl{A} \times \mbb{C}_n}\mcl{L}_n(\rho, \mbf{C})$ and $\mcl{L}_n^*(\rho) = \E \mcl{L}_n(\rho, \tilde{\mbf{C}}_0(\rho))$.
By the chain rule, 
\begin{align*}
    \partial_\rho \mcl{L}_n^*(\rho) = \partial_{\rho_1} \E \mcl{L}_n(\rho, \tilde{\mbf{C}}_0(\rho)) + \partial_\mbf{C^\top} \E\mcl{L}_n(\rho, \tilde{\mbf{C}}_0(\rho)) \partial_\rho \tilde{\mbf{C}}_0(\rho) = \partial_{\rho_1} \E \mcl{L}_n(\rho, \tilde{\mbf{C}}_0(\rho)),
\end{align*}
where the second equality holds for any $\rho \in \mcl{R}$ by the first-order condition for $\tilde{\mbf{C}}_0(\rho)$, and we have denoted $\partial_{\rho_1}$ to stand for the partial derivative with respect to the ``first'' $\rho$.
Then, we also have
\begin{align*}
    \partial^2_{\rho\rho} \mcl{L}_n^*(\rho)
    & = \partial^2_{\rho_1 \rho_1} \E \mcl{L}_n(\rho, \tilde{\mbf{C}}_0(\rho)) + \partial^2_{ \rho_1 \mbf{C^\top}} \E \mcl{L}_n(\rho, \tilde{\mbf{C}}_0(\rho)) \partial_\rho \tilde{\mbf{C}}_0(\rho).
\end{align*}
In addition, applying the implicit function theorem to $\partial_\mbf{C} \E\mcl{L}_n(\rho, \tilde{\mbf{C}}_0(\rho)) = \mbf{0}_{(d_z + 2n) \times 1}$ for $\rho \in \mcl{R}$ yields
\begin{align*}
    & \mbf{0}_{(d_z + 2n) \times 1} = \partial^2_{\mbf{C} \rho } \E\mcl{L}_n(\rho, \tilde{\mbf{C}}_0(\rho)) = \partial^2_{ \mbf{C} \rho_1} \E\mcl{L}_n(\rho, \tilde{\mbf{C}}_0(\rho)) + \partial^2_{\mbf{C} \mbf{C}^\top} \E\mcl{L}_n(\rho, \tilde{\mbf{C}}_0(\rho)) \partial_\rho \tilde{\mbf{C}}_0(\rho)\\
    & \Longrightarrow \partial_\rho \tilde{\mbf{C}}_0(\rho) = -\left[ \partial^2_{\mbf{C} \mbf{C}^\top} \E\mcl{L}_n(\rho, \tilde{\mbf{C}}_0(\rho)) \right]^{-1} \partial^2_{ \mbf{C} \rho_1} \E\mcl{L}_n(\rho, \tilde{\mbf{C}}_0(\rho)).
\end{align*}
Therefore, we obtain
\begin{align*}
     \partial^2_{\rho \rho} \mcl{L}_n^*(\rho) = \partial^2_{\rho_1 \rho_1} \E \mcl{L}_n(\rho, \tilde{\mbf{C}}_0(\rho)) - \partial^2_{ \rho_1 \mbf{C}^\top} \E \mcl{L}_n(\rho, \tilde{\mbf{C}}_0(\rho)) \left[ \partial^2_{\mbf{C} \mbf{C}^\top} \E\mcl{L}_n(\rho, \tilde{\mbf{C}}_0(\rho)) \right]^{-1} \partial^2_{ \mbf{C} \rho_1 } \E\mcl{L}_n(\rho, \tilde{\mbf{C}}_0(\rho)).
\end{align*}
Note that since $\partial^2_{\mbf{C} \mbf{C}^\top} \E\mcl{L}_n(\rho, \tilde{\mbf{C}}_0(\rho))$ is negative semidefinite by the second-order condition of maximization, the second term on the right-hand side is non-positive.
Thus, to ensure that $\partial^2_{\rho \rho} \mcl{L}_n^*(\rho)$ is strictly negative, $\partial^2_{\rho_1 \rho_1} \E \mcl{L}_n(\rho, \tilde{\mbf{C}}_0(\rho))$ must also be negative and  larger in magnitude than the second term.

\subsection{Supplementary tables for Section \ref{sec:empiric}}\label{subsec:empiric}

\begin{table}[!h]
    \begin{center}
        \caption{In-degree and Out-degree}
    \begin{tabular}{lccllccllcc}
        \hline\hline
    Country & In-degree & Out-degree &  &  &  &  &  &  &  & \\
        \cline{1-3}
        Armenia & 21 & 46 &  & (cont') &  &  &  & (cont') &  &  \\
        Australia & 43 & 6 &  & Japan & 48 & 14 &  & Papua New Guinea & 18 & 21 \\
        Azerbaijan & 24 & 24 &  & Jordan & 13 & 41 &  & Philippines & 22 & 38 \\
        Bahrain & 33 & 22 &  & Kazakhstan & 29 & 34 &  & Qatar & 29 & 31 \\
        Bangladesh & 6 & 39 &  & Kiribati & 22 & 15 &  & Russia & 37 & 18 \\
        Belarus & 26 & 33 &  & Kuwait & 35 & 23 &  & Saudi Arabia & 31 & 19 \\
        Bhutan & 15 & 2 &  & Kyrgyzstan & 24 & 33 &  & Singapore & 47 & 34 \\
        Brunei & 43 & 26 &  & Laos & 17 & 49 &  & South Korea & 49 & 25 \\
        Cambodia & 15 & 54 & & Latvia & 40 & 19 &   & Sri Lanka & 7 & 53 \\
        China & 24 & 9 &  & Lebanon & 13 & 34 &  & Tajikistan & 23 & 36 \\
        Cyprus & 41 & 19 &  & Lithuania & 31 & 18 &  & Thailand & 31 & 37 \\
        Estonia & 42 & 19 &  & Malaysia & 46 & 46 &  & Tonga & 23 & 21 \\
        Fiji & 23 & 29 &  & Moldova & 26 & 28 &  & Turkey & 35 & 27 \\
        Georgia & 29 & 32 &  & Mongolia & 21 & 20 &  & UAE & 44 & 20 \\
        Hong Kong & 43 & 33 &  & Myanmar & 13 & 16 &  & Ukraine & 36 & 29 \\
        India & 13 & 3 &  & Nauru & 21 & 7 &  & Uzbekistan & 22 & 28 \\
        Indonesia & 25 & 51 &  & Nepal & 9 & 55 &  & Vanuatu & 27 & 33 \\
        Iran & 10 & 48 &  & New Zealand & 45 & 18 &  & Viet Nam & 16 & 14 \\
        Iraq & 5 & 5 &  & Oman & 34 & 5 &  & Yemen & 7 & 9 \\
        Israel & 33 & 21 &  & Pakistan & 5 & 21 &  &  &  &  \\
        \hline\hline
    \end{tabular}
    \label{table:degree}
    \end{center}
\end{table}

\begin{table}[!h]
    \begin{center}
    \caption{Summary Statistics}
    \begin{tabular}{lccccc}
            \hline \hline
    & Mean & Std. Dev. & Min. & Max. & \# Observations \\
            \hline
    $\textit{gdp\_pc}_i$ & 14.4409  & 17.3406   & 0.8071  & 67.5714  & 57 \\
    $\textit{free}_i$ & 4.0789  & 1.9291  & 1  & 7  & 57 \\
    $\textit{export}_{ij}$ & 2.8562  & 2.9874  & 0  & 12.6580  & 3,192 \\
    $\textit{import}_{ij}$ & 2.9219  & 2.9773  & 0  & 12.5461  & 3,192 \\
            \hline \hline
    \end{tabular}
    \label{table:summary}
    \end{center}
\end{table}

\begin{table}[!h]
    \begin{center}
    \caption{BIC Model Selection}
    \begin{tabular}{lc|cccccc}
        \hline \hline
        &  & $K^B$ & & &   &   &  \\
        &  & 2 & 3 & 4 & 5 & 6 & 7 \\
        \hline
        $K^A$ & 2 & 2276.659  & 2078.159  & 2122.551  & 2102.056  & 2104.882  & 2117.091  \\
        & 3 & 2033.604  & 1935.613  & 1909.670  & 1903.809  & 1899.412  & 1909.794  \\
        & 4 & 1985.267  & 1883.387  & 1837.424  & 1834.551  & 1823.058  & 1848.088  \\
        & 5 & 1931.314  & 1837.387  & 1798.743  & 1789.363  & 1787.481  & 1795.286  \\
        & 6 & 1918.471  & 1808.645  & 1775.658  & 1769.636  & 1765.887  & 1768.170  \\
        & 7 & 1897.173  & 1807.099  & 1777.586  & 1766.433  & 1761.725  & 1764.802  \\       
        \hline \hline
    \end{tabular}
    \label{table:BIC}
\end{center}
\end{table}

\bibliography{ref}

\end{document}